%% file: 131.tex
\documentclass[10pt,conference]{IEEEtran}
\usepackage{cite}
\ifCLASSINFOpdf
\else
\fi

\usepackage{mysty}

\begin{document}
\title{Register automata with linear arithmetic}


\author{
\IEEEauthorblockN{Yu-Fang Chen\IEEEauthorrefmark{1}, Ond\v {r}ej Leng\'{a}l\IEEEauthorrefmark{2}, Tony Tan\IEEEauthorrefmark{3} and Zhilin Wu\IEEEauthorrefmark{4}}
\\
\IEEEauthorblockA{\IEEEauthorrefmark{1} Institute of Information Science,
Academia Sinica, Taiwan
}
\IEEEauthorblockA{\IEEEauthorrefmark{2} FIT,
Brno University of Technology,
IT4Innovations Centre of Excellence,
Czech Republic
}
\IEEEauthorblockA{\IEEEauthorrefmark{3} Department of Computer Science and Information Engineering,
National Taiwan University, Taiwan
}
\IEEEauthorblockA{\IEEEauthorrefmark{4} State Key Laboratory of Computer Science, Institute of Software, Chinese Academy of Sciences, China
}
}

\IEEEoverridecommandlockouts
\IEEEpubid{\makebox[\columnwidth]{978-1-5090-3018-7/17/\$31.00~
\copyright2017 IEEE \hfill} \hspace{\columnsep}\makebox[\columnwidth]{ }}

\maketitle

\begin{abstract}
We propose a novel automata model over the alphabet of rational numbers,
which we call {\em register automata over the rationals} ($\raq$).
It reads a sequence of rational numbers and outputs another rational number.
$\raq$ is an extension of the well-known {\em register automata} (RA)
over infinite alphabets, which
are finite automata equipped with a~finite number of registers/variables
for storing values.
Like in the standard RA, the $\raq$ model allows both equality and ordering tests between values. 
It, moreover, allows to perform linear arithmetic between certain variables.
The model is quite expressive:
in addition to the standard RA,
it also generalizes other
well-known models such as affine programs and arithmetic circuits.

The main feature of $\raq$ is that despite the use of
linear arithmetic, the so-called {\em invariant problem}---a~generalization
of the standard non-emptiness problem---is decidable.
We also investigate other natural decision problems, namely, {\em commutativity},
{\em equivalence}, and {\em reachability}.
For deterministic $\raq$, commutativity and equivalence are
polynomial-time inter-reducible with the invariant problem.
\end{abstract}

\IEEEpeerreviewmaketitle

\section{Introduction} \label{sec:intro}

Motivated by various needs and applications, there have occurred many studies
on languages over infinite alphabets.
To name a few, typical applications include database systems, program analysis and verification,
programming languages and theory by itself.
See, e.g.,~\cite{Segoufin06,Bojanczyk10,ChenDaBeast16,NevenSST15,
AartsJUV15,Tzevelekos11,LibkinMV16,GrumbergKS10,DemriL09} and the references therein.
One of the most popular models are arguably register automata (RA)~\cite{ShemeshF94,KaminskiF94}.
Briefly, an~RA is a~finite automaton equipped with a finite number of registers,
where each register can store one symbol at a time.
The automaton then moves from state to state
by comparing the input symbol with those in its registers,
and at the same time may decide to update the content of its registers
by storing a new symbol into one of its registers,
and thus, ``forgetting'' the previously stored symbol.
The simplicity and naturality of RA obviously contribute to their appeal.

So far, the majority of research in this direction has focused on models
where the only operations allowed on the input symbols are equality and order relation.
For many purposes, this abstraction is good enough.
For example, relational algebra-based queries, often used in database systems,
involve only equality tests~\cite{AbiteboulHV95}.
For many simple but common queries, such as 
counting the number of elements or summing up the values in a list,
at least some arithmetic is, however, required.

It is a folklore belief that allowing RA to perform even the~simplest form of arithmetic
on their registers will immediately yield undecidability for the majority of interesting
decision problems.
The evidence is that such RA subsume simple two-counter machines,
which are already Turing-complete~\cite{Minsky67}.
Indeed, the belief holds not only for RA,
but for the majority (if not all) of the models of languages over infinite alphabets.

In this paper, we propose a~novel automaton model over the rational numbers~$\bbQ$, 
named \emph{register automata over the rationals} ($\raq$).
Like in standard RA, an $\raq$ is equipped with a finite number of variables (registers),
each of them is able to store a value.\footnote{Though normally called
registers, for reasons that will be apparent later, we will refer to them as {\em variables} in this paper.}
The $\raq$ model allows to test order and perform linear arithmetic between some variables, 
yet keeps several interesting decision problems decidable.
The key idea is the partitioning of variables into two sets, \emph{control variables} and \emph{data variables}, 
which is inspired by the work of Alur and \v{C}ern\'{y}~\cite{Alur11}.
\emph{Control variables} can be used in transition guards for order ($\leq$) comparison and
can be assigned a~value either from the input or from another control variable.
In contrast, \emph{data variables} can store a value obtained from a~linear
combination of the values of all variables and the value from the input, but
cannot be used in transition guards.
In a~final state, an~$\raq$ outputs a~rational number
obtained by a~linear combination of the values of all variables
(non-final states have no output).
Due to nondeterminism, it is possible that different computation paths
for the same input word produce different output values.
$\raq$ can be used to model, e.g., the following aggregate functions:
finding the smallest and the largest elements,
finding the $k$-th largest element,
counting the number of elements above a~certain threshold,
summing all elements, or
counting the number of occurrences of the largest element in a~list.

The $\raq$ model is a very general model that captures and simulates at least three other well-known models.
The first and obvious one is the standard RA studied in~\cite{KaminskiF94,ShemeshF94,DemriL09,NSVianu04}.
An~RA is simply an~$\raq$ without
data variables that only allows equality test of control variables and
in a~final state outputs the constant~1.
The second one is the affine program (AP) model defined by Karr~\cite{Karr76},
which is commonly used as
a~standard abstract domain in static program analysis~\cite{CousotCousot77,Jeannet2009apron}.
An~AP is a~special case of an~$\raq$ where control variables as well as
values of the input are ignored.
Finally, $\raq$ can also simulate (division-free) arithmetic circuits (AC) without indeterminates.
Originally, AC were introduced as a model for studying algebraic complexity~\cite{ShpilkaY10,BCS97},
but recently gain prominence as a~model for analysing 
the complexity of numerical analysis~\cite{AllenderBKM09}
due to its succinct representation of numbers.
We show that $\raq$ can be used to represent numbers using roughly only twice
as many transitions as the number of edges in the AC that represents the same number.

We study several decision problems for $\raq$.
The first one is the so-called {\em invariant} problem,
which asks if the set of reachable configurations of a given $\raq$ at a given state
is {\em not} contained in a given {\em affine space}.\footnote{Formal definitions will be presented later on, including
the representation of the given affine space.}
This is a typical decision problem in AP, where one would like to find out
the relations among the variables when the program reaches a certain state~\cite{Karr76,MSeidl04}.
We show that the invariant problem for $\raq$ is polynomial-time inter-reducible with another
decision problem called the {\em non-zero} problem,
which asks if a given $\raq$ can output a~non-zero value for some input word.
Note that the non-zero problem is a generalization of the non-emptiness problem of RA,
if we assume that an RA outputs the constant~1 in its final states.
We show that the non-zero problem is decidable in exponential time
(in the number of control variables, while polynomial in other parameters).
Our algorithm is based on the well-known Karr's algorithm~\cite{Karr76,MSeidl04}
for deciding the same problem for AP.

We should remark that the exponential complexity is in the bit model,
i.e., rational numbers are represented in their bit forms.
If we assume that each rational number occupies only a constant space,
e.g., the Blum-Shub-Smale model~\cite{BSS89}, the non-zero problem is $\pspace$-complete,
which matches the non-emptiness problem of standard RA~\cite{DemriL09}.

In addition, we also prove a small model property on the length of the shortest word
leading to a non-zero output.
From that, we derive a polynomial space algorithm for the non-zero problem for 
the so-called~{\em copyless} $\raq$, i.e., 
$\raq$ where reassignments to data variables are copyless\footnote{The 
copyless constraint of $\raq$ is inspired by and in the same flavour of the one for streaming transducers in \cite{Alur11}.}.
In fact, the non-zero problem becomes $\pspace$-complete.
It should be remarked that copyless $\raq$ already subsume standard RA.

The separation of control and data variables is the key
to make the non-zero problem decidable. 
In fact, allowing $\raq$ to access just the {\em least
significant bit} of their data variables is already enough to make them Turing-complete,
and so is allowing order comparison between data variables.
Without control variables, $\raq$ become AP, positioning
their invariant problem in $\ptime$~\cite{Karr76,MSeidl04}.
$\raq$ without data variables are copyless, which makes their invariant problem
$\pspace$-complete (as mentioned above).

We also study the \emph{commutativity} and \emph{equivalence} problems for $\raq$.
The former asks whether a~given $\raq$ is commutative.
A~commutative $\raq$ is an~$\raq$ that, given a word~$w$ as its input, 
outputs the same value on any permutation of~$w$. 
The latter problem asks if two $\raq$ are essentially the same, i.e.,
for every input word, the two $\raq$ output the same set of values.
The equivalence problem is known to be undecidable already for RA~\cite{NSVianu04} 
via a reduction from \emph{Post correspondence problem} (PCP).
The same reduction can be used to show that the commutativity problem for RA is also undecidable.
For deterministic $\raq$, we show that the commutativity, equivalence, and
invariant problems are inter-reducible to each other in polynomial time.
Thus, for deterministic copyless $\raq$
(and therefore also for deterministic RA),
all problems mentioned above can be decided in polynomial space, 
and are, in fact, $\pspace$-complete.

Finally, we also study the \emph{reachability} problem for $\raq$.
This problem asks if a given $\raq$ can output $0$ for some input word.
We show that although the reachability problem is undecidable in general, even when the $\raq$ is
deterministic, 
it is in $\nexptime$ for nondeterministic copyless~$\raq$ 
with non-strict transition guards\footnote{A guard is non-strict if it \emph{does not} contain negations, 
i.e., it is a positive Boolean combination of inequalities $z \le z'$.}.
The decision procedure is obtained by 
a reduction to the configuration coverability problem of
\emph{rational vector addition systems with states} ($\bbQ$-VASS).
Since there is an exponential blow-up in the reduction and the configuration coverability problem 
of $\bbQ$-VASS is in $\nptime$, 
we get a~nondeterministic exponential-time decision procedure for the reachability problem of 
copyless $\raq$ with non-strict transition guards.

\begin{table*}[t]
\caption{Overview of the results
\textup{(SV- means \emph{single-valued},
CL- means \emph{copyless},
NSTG- means \emph{with non-strict transition guards},
\mbox{reachability for (deterministic) RA means \emph{state reachability},
-c means \emph{complete},
UNDEC means \emph{undecidable}})}
}
\begin{center}
\begin{tabular}{|l||r|r|r|r|}
\hline \multicolumn{1}{|c||}{\textbf{Model}}     & \multicolumn{1}{c|}{\textbf{Non-zero (Emptiness)}} &  \multicolumn{1}{c|}{\textbf{Equivalence}}  & \multicolumn{1}{c|}{\textbf{Commutativity}}        &  \multicolumn{1}{c|}{\textbf{Reachability}} \\
\hline
\hline RA~\cite{DemriL09}   & $\pspace$-c~\cite{DemriL09}                   & UNDEC~\cite{NSVianu04}                        & UNDEC (Thm.~\ref{thm:raq-undecidable})       & $\pspace$-c~\cite{DemriL09}         \\
\hline deterministic RA~\cite{DemriL09} & $\pspace$-c~\cite{DemriL09}                   & $\pspace$-c~\cite{DemriL09}        & $\pspace$-c~(Cor.~\ref{theo:detra})        & $\pspace$-c~\cite{DemriL09}         \\
\hline $\raq$               & $\exptime$ (Thm.~\ref{theo:non-zero-exptime}) & UNDEC (Thm.~\ref{thm:raq-undecidable})       & UNDEC (Thm.~\ref{thm:raq-undecidable})       & UNDEC (Thm.~\ref{thm-reach-und})      \\
\hline SV-$\raq$            & $\exptime$ (Thm.~\ref{theo:non-zero-exptime}) & UNDEC (Thm.~\ref{thm:raq-undecidable})       & UNDEC (Thm.~\ref{thm:raq-undecidable})       & UNDEC (Thm.~\ref{thm-reach-und})      \\
\hline CL-$\raq$            & $\pspace$-c (Thm.~\ref{theo:copyless})        & UNDEC (Thm.~\ref{thm:raq-undecidable})       & UNDEC (Thm.~\ref{thm:raq-undecidable})       & $?$                                 \\
\hline deterministic\ $\raq$             & $\exptime$ (Thm.~\ref{theo:non-zero-exptime}) & $\exptime$~(Cor.~\ref{theo:comm-exptime})  & $\exptime$~(Cor.~\ref{theo:comm-exptime})  & UNDEC (Thm.~\ref{thm-reach-und})      \\
\hline deterministic CL-$\raq$          & $\pspace$-c (Thm.~\ref{theo:copyless})        & $\pspace$-c (Cor.~\ref{theo:comm-exptime}) & $\pspace$-c (Cor.~\ref{theo:comm-exptime}) & $?$                                 \\
\hline NSTG-CL-$\raq$       & $\pspace$-c (Thm.~\ref{theo:copyless})        & $?$                                        & $?$                                        & $\nexptime$ (Thm.~\ref{thm-reach-dec}) \\
\hline
\end{tabular}

\end{center}
\label{tab:overview-table}
\end{table*}

An overview of the results obtained in this paper can be found in Table~\ref{tab:overview-table}. All decision problems we consider are natural and have corresponding applications.
%The non-zero problem is a dual problem to the invariant problem.
The invariant, equivalence, and reachability problems are
standard decision problems considered in formal verification.
RA and $\raq$ are natural models of Reducer
programs~\cite{ChenDaBeast16,NevenSST15} in the MapReduce paradigm~\cite{DeanG04}, where commutativity is an important property required for Reducers~\cite{ChenHSW15,ChenDaBeast16,comm-harmful}.

Lastly, let us explain the main differences between the decision
procedures for RA and those presented for $\raq$.
The non-emptiness and reachability problems for RA can essentially be reduced to the
reachability problem in a finite-state system, where one can bound the number
of data values and consider a finite alphabet. The commutativity and
equivalence problems for deterministic RA can then be reduced to the non-emptiness
problem.
On the other hand, due to the use of arithmetic operations, similar techniques
are no longer applicable in $\raq$,
thus, a different set of tools is then required such as Karr's
algorithm and those from algebra and linear programming as used in this paper.

\paragraph*{Organization}
We review some basic linear algebra tools and Karr's algorithm in Section~\ref{sec:prelim}.
In Section~\ref{sec:snt}, we present the formal definition of $\raq$.
We discuss the invariant and non-zero problems in Section~\ref{sec:non-zero},
and the commutativity and equivalence problems in Section~\ref{sec:three-problem}.
In Section~\ref{sec:other} we discuss the reachability problem.
We conclude with some discussions on related works and remarks in
Sections~\ref{sec:related} and~\ref{sec:conclusion}.
All missing technical details and proofs can be found in the appendix.

%%%%%%%%%%%%%%%%%%%%%%%%%%%%%%%%%%%%%%%

\section{Preliminaries}
\label{sec:prelim}

In this paper, a {\em word}~$w$ is a finite sequence of rational numbers $w =
d_1\cdots d_n \in \bbQ^*$.
The \emph{length} of $w$ is $n$, denoted by $|w|$.
The term {\em data value}, or {\em value} for short, means a rational number.
Matrices and vectors are  over the rational numbers $\bbQ$,
where $\bbQ^{m\times n}$ and $\bbQ^{k}$ denote
the sets of matrices of size $m\times n$ and 
column vectors of size $k$ (i.e., $\bbQ^{k}= \bbQ^{k\times 1}$), respectively. 
All vectors in this paper are understood as column vectors.

We use $A,B,\ldots$ to denote matrices, where
$A(i,j)$ is the component in row $i$ and column $j$ of matrix $A$.
We denote the transpose of $A$ by $A^t$, and the determinant of a square matrix~$A$ by $\det(A)$.
We use $\va,\vb,\vu,\vv,\ldots$ to denote vectors,
where $\vu(i)$ is the $i$-th component of vector $\vu$ (numbered from~1).

When $\vu \in \bbQ^k$ and $\vv\in \bbQ^l$,
we write $\myvec {\vu \\ \vv}$ to denote
a vector in $\bbQ^{k+l}$
composed as the concatenation of~$\vu$ and~$\vv$.
Abusing the notation, we write $0$ to also denote both the zero vector and the zero matrix.

For two vectors $\vu,\vv \in \bbQ^k$,
we write $\vu \geq \vv$ when $\vu(i)\geq \vv(i)$
for each component $i=1,\ldots,k$.
The dot product of $\vu$ and $\vv$ is denoted by $\vu\dotprod \vv$.

\subsubsection*{Affine spaces}
Recall that a \emph{vector space} $\bbV$ in $\bbQ^k$ is a subset of $\bbQ^k$ 
that forms a~group under addition $+$ and 
is closed under scalar multiplication, i.e., for all $\vv\in \bbV$ and $\alpha\in \bbQ$,
it holds that $\alpha\vv \in \bbV$.
The dimension of $\bbV$ is denoted by $\dim(\bbV)$.
The \emph{orthogonal complement} of $\bbV$ is the vector space
$\bbV^{\perp}=\{\vu\mid \vu \dotprod \vv = 0 \ \textrm{for every }\vv \in \bbV\}$.
It is known that  $\dim(\bbV^{\perp})+\dim(\bbV)=k$.

An \emph{affine space} $\bbA$ in $\bbQ^k$ is a set of the form $\va+\bbV$,
where $\va \in \bbQ^k$ and $\bbV$ is a vector space in $\bbQ^k$.
Here, $\va + \bbV$ denotes the set $\{\va+\vu \mid \vu \in \bbV\}$.
The dimension of~$\bbA$, denoted
$\dim(\bbA)$, is defined as $\dim(\bbV)$.

A vector $\vu$ is an \emph{affine combination} of $V=\{\va_1,\ldots,\va_n\}$, 
if there are $\lambda_1,\ldots,\lambda_n \in \bbQ$
such that $\sum_{i=1}^n \lambda_i = 1$ and $\vu = \sum_{i=1}^n \lambda_i \va_i$.
We use $\aff(V)$ to denote the
space of all affine combinations of $V$.
It is known that for every affine space $\bbA$,
there is a set $V$ of size~$\dim(\bbA) + 1$ such that $\aff(V)=\bbA$.

An \emph{affine transformation} $T:\bbQ^k \to \bbQ^l$
is defined by a~matrix $M \in \bbQ^{l\times k}$ and a vector $\va \in \bbQ^l$,
such that~$T\vx=M\vx +\va$.
When $\va=0$, $T$~is called a~\emph{linear transformation}.
From basic linear algebra, when $k=l$, it holds that $T$ is a~one-to-one mapping iff
$\det(M)\neq 0$.

For convenience, we simply write \emph{transformation}
to mean affine transformation.
Note that composing two transformations $T_1$ and $T_2$
yields another transformation $\vx \mapsto T_2T_1\vx$,
where $\vx \mapsto T_2T_1\vx$ denotes a~function that maps $\vx$ to $T_2T_1\vx$.

The following two lemmas will be useful.

\begin{lemma}
\label{lem:affine-space}
Let $\bbA\subseteq \bbQ^k$ be an affine space
and $T: \bbQ^{k+1}\to \bbQ^k$ be a transformation.
Suppose there is a~vector~$\vv \in \bbQ^k$ and values $d_1, d_2 \in \bbQ$, where $d_1\neq d_2$, such that
both $T\myvec{\vv \\ d_1}$ and $T \myvec{\vv\\ d_2}$ are in $\bbA$.
Then, $T \myvec{\vv \\ d} \in \bbA$ for every $d\in \bbQ$.
\end{lemma}

\begin{lemma}
\label{lem:finite-dim}
Let $T_1,\ldots,T_{m}$ be transformations and
$\vu_1,\vu_2,\ldots,\vu_{m+1}$ be vectors such that
$\vu_{i+1}= T_i \vu_i$
for every $i = 1,\ldots, m.$
Let $\bbH$ be an affine space such that $\vu_{m+1} \notin \bbH$ and $m\geq \dim(\bbH)+2$.
Then, there is a set of indices $J=\{j_1,\ldots,j_n\}$ 
such that $1\leq j_1 < j_2 < \cdots < j_n \leq m$, $|J| \leq \dim(\bbH)+1$, and
$T_{j_n} T_{j_{n-1}}\cdots T_{j_1}\vu_1 \notin \bbH$.
\end{lemma}

\paragraph*{Affine programs}
An \emph{affine program} (AP) with $n$~variables is a~tuple $\cP =(S,s_0,\mu)$,
where $S$ is a~finite set of states, $s_0 \in S$ is the initial state,
and $\mu$ is a~finite set of transitions of the form
$(s_1,T,s_2)$, where $s_1, s_2 \in S$ and $T:\bbQ^n\to \bbQ^n$ is a~transformation. 
Intuitively, $\cP$~represents a~program with $n$~rational variables, say $z_1,\ldots,z_n$.
A transition $(s_1,T,s_2) \in \mu$ means that the program can
move from state $s_1$ to $s_2$ while reassigning the contents of variables via 
$\vz \mapsto T \vz$, where $\vz$ denotes the column vector of the variables $z_1,\ldots,z_n$.

A \emph{configuration} of $\cP$ is a~pair $(s,\vu) \in S\times \bbQ^n$
where $s$~is a state of~$\cP$ and $\vu$ represents the contents of its variables.
An \emph{initial} configuration is a configuration $(s_0, \vu)$.
A~path in $\cP$ from a configuration $(s,\vu)$ to a~configuration $(s',\vv)$
is a~sequence of transitions $(p_0,T_1,p_1),\ldots,(p_{m-1},T_m,p_m)$ of~$\cP$
such that $p_0=s$, $p_m=s',$ and $\vv = T_m\cdots T_1 \vu$.

The \emph{AP invariant} problem is defined as follows:
{\em Given an AP $\cP$, a~vector $\vu \in \bbQ^n$, a state $s'\in S$, and an affine space~$\bbH$,
decide if there is a path in $\cP$ from $(s_0,\vu)$ to $(s',\vv)$ for some $\vv\notin\bbH$}.
If there is such a~path, then $\bbH$ is not an invariant for $s'$ in $\cP$ w.r.t. the initial configuration $(s_0, \vu)$.
The input affine space $\bbA=\va+\bbV$ can be represented
either as a pair $(\va,V)$ where $V$ is a basis of~$\bbV$,
or as a set of vectors~$V$ where $\aff(V)=\bbA$.
Either representation is fine as one can be easily
transformed to the other.

The AP invariant problem can be solved in a~polynomial time by the so-called Karr's algorithm~\cite{MSeidl04,Karr76}.
The main idea of Karr's algorithm is to compute, for every state~$s \in S$, a~set of vectors~$V_s$
such that the existence of a path from $(s_0, \vu)$ to $(s,\vv)$
implies $\vv\in \aff(V_s)$.
The algorithm works as follows:
At the beginning, it sets $V_{s_0}=\{\vu\}$
and $V_s = \emptyset$ for all other $s\neq s_0$.
Then, using a worklist algorithm, it starts propagating the values of $V_s$
over transitions such that
for each transition $(s_1,T,s_2) \in \mu$ and each vector $\vv \in V_{s_1}$ that has not been
processed before, it adds the vector $T\vv$ into $V_{s_2}$, if $T\vv \notin \aff(V_{s_2})$.
It holds that there is a path from $(s_0,\vu)$ to $(s',\vv)$, for some $\vv \notin\bbH$,
iff $V_{s'} \nsubseteq \bbH$.
Note that we can limit the cardinality of $V_s$ to be at most $n+1$,
hence the algorithm runs in a~polynomial time.
We refer the reader to~\cite{MSeidl04,Karr76} for more details.

\begin{remark}
\label{rem:karr-small-model}
From Karr's algorithm, we can infer a small model property for the invariant problem.
That is, if there is a~path from $(s_0,\vu)$ to $(s',\vv)$, for some $\vv\notin \bbH$,
then there is such a path of length at most $(n+1)|S|$.
Such a~bound can also be derived in a more straightforward manner via Lemma~\ref{lem:finite-dim},
which we believe is interesting on its own.
In fact, if all\linebreak
transformations in an AP are one-to-one,
the bound $(n+1)|S|$ can be lowered to $(\dim(\bbH)+2)|S|$,
which can be particularly useful when $\dim(\bbH)$ is small.
\end{remark}

%%%%%%%%%%%%%%%%%%%%%%%%%%%%%%%%%%%%%%%

\section{Register automata over the rationals ($\raq$)}
\label{sec:snt}

%%%%%%%%%%%%%%%%%%%%%%%%%%%%%%%%%%%%%%%%%%%%%%%%%%%%%%%%%%%%%%%%%%%%%%%%%%%%%%%%

In the following, we fix $X=\{x_1,\ldots,x_k\}$ and $Y=\{y_1,\ldots,y_l\}$,
two disjoint sets of variables called {\em control} and {\em data} variables, respectively.
The vector $\vx$ always denotes a~vector of size $k$ 
where $\vx(i)$ is $x_i$. Likewise, $\vy$ is of size~$l$ and $\vy(i)=y_i$.
We also reserve a special variable $\cur \notin X\cup Y$ to denote the data value
currently read by the automaton.

Each variable in $X\cup Y$ can store a data value (these variables are sometimes called \emph{registers}).
When a vector $\vu \in \bbQ^{k+l}$ is used to represent the contents of variables in $X\cup Y$,
the first $k$ components of $\vu$ represent the contents of control variables, denoted by $\vu\ssX$,
and the last $l$ components represent the contents of data variables, denoted by $\vu\ssY$. We also use $\vu(x_i)$ and $\vu(y_j)$ to denote the contents of $x_i$ and $y_j$ in $\vu$, respectively.

A {\em linear constraint} over $X \cup \{\cur\}$ is
a Boolean combination of atomic formulas of the form $z \leq z'$,
where $z,z' \in X\cup \{\cur\}$.
We write $\cC(X,\cur)$ to denote the set of all
linear constraints over~$X \cup \{\cur\}$. For convenience, we use $z < z'$ as an abbreviation of $\neg(z' \le z)$.
In the following, $\bbP^{k\times (k+1)}$ denotes
the set of all 0-1 matrices in which the number of $1$'s in each row is exactly one. 
Intuitively, a matrix $A\in \bbP^{k\times(k+1)}$
denotes a mapping from $\{x_1,\ldots,x_k \}$ to $\{x_1,\dots, x_k, \cur\}$.

\begin{definition}
\label{def:snt} 
A~\emph{register automaton over the rationals} ($\raq$) with control and data
variables $(X,Y)$ is a~tuple $\cA=\langle Q,q_0,F,\vu_0,\delta,\zeta\rangle$
defined as follows:
\begin{itemize}\itemsep=0pt
\item
$Q$ is a~finite set of states, $q_0 \in Q$ is the initial state, and $F \subseteq Q$
is the set of final states.
\item
$\vu_0 \in \bbQ^{k+l}$ is the initial contents of variables in $X\cup Y$.
\item
$\delta$ is a set of transitions whose elements are of the form
\begin{eqnarray}
\label{eq:transition}
t: \quad (p,\varphi(\vx,\cur)) & \to &
(q,A,B,\vb),
\end{eqnarray}
where $p,q \in Q$ are states, $\varphi(\vx,\cur)$ is a~linear constraint from $\cC(\vx,\cur)$,
and $A\in \bbP^{k\times (k+1)}$, $B \in \bbQ^{l \times (k+l+1)}$, $\vb \in \bbQ^l$.
The formula $\varphi(\vx,\cur)$ is called the {\em guard} of~$t$
and the triple $(A,B,\vb)$ its {\em variable reassignment}.
\item
$\zeta$ is a mapping that maps each final state $q_f$ to 
a linear function/expression $g(\vx,\vy)= \va\cdot \vx + \vb\cdot \vy + c$,
where $\va\in \bbQ^{k}$, $\vb\in \bbQ^{l}$, and $c\in \bbQ$.
\end{itemize}
\end{definition}

The intuitive meaning of the transition in~(\ref{eq:transition}) is as follows.
Suppose the contents of variables in $\vx$ and $\vy$ are $\vu$ and $\vv$, respectively.
If $\cA$ is in state $p$, currently reading data value~$c$,
and the guard $\varphi(\vu,c)$ holds,
then $\cA$~can enter state $q$ and reassign the variables $\vx$ with $A \myvec {\vu \\ c}$
and $\vy$ with $B \myvec {\vu \\ \vv \\ c} + \vb$.

Note that the matrix representation of the reassignment can be equivalently written as
$(i)$~a~reassignment of each control variable in $X$ with a variable
in $X\cup \{\cur\}$, and
$(ii)$~a~reassignment of each data variable in $Y$ with a~linear combination of
variables in $X\cup Y\cup\{\cur\}$ and constants.
We will therefore sometimes write the matrices $A$, $B$, and the vector $\vb$ as reassignments to variables
of the form
$\{x_1 := r_1; \ldots; x_k := r_k; y_1 := s_1; \ldots; y_l := s_l\}$ or
$\{\vx:= f(\vx,\cur);\ \vy:= g(\vx,\vy,\cur)\}$. 
For readability, we omit identity reassignments such as $x_i:= x_i$ or
$y_j: = y_j$ from the first form
since the values of these variables do not change.
%}
A~variable $x_i \in X$ is said  to be \emph{read-only} if, for each
transition~$t$, the reassignment to $x_i$ in $t$ is always of the form
$x_i:=x_i$.

\begin{remark}
Note that in the definition above, the guards only allow comparisons
among the current data value and the contents of variables in~$\vx$. 
One can easily generalize the guards so that comparisons with constants are allowed. 
Such a~generalization does not affect the expressive power of~$\raq$,
since every such a~constant~$c$ can be stored into a~fresh read-only control
variable $x_c$ in the initial assignment~$\vu_0$. This notation is chosen for technical convenience. On the other hand, for readability, in some of the examples later on, we do use comparisons with constants, which, strictly speaking, should be taken as comparisons with read-only control variables.
\end{remark}

A {\em configuration} of $\cA$ is a pair $(q,\vu) \in Q\times \bbQ^{k+l}$,
where $\vu$~denotes the contents of the variables.
The {\em initial configuration} of $\cA$ is $(q_0,\vu_0)$,
while final configurations are those with a~final state in the left-hand component.
A transition $t= (p,\varphi(\vx,\cur))\to(q,A,B,\vb)$ and a value~$d$
entail a binary relation
$(p,\vu) \vdash_{t,d} (q,\vv)$, if
\begin{itemize}\itemsep=0pt
\item
$\varphi(\vu\ssX,d)$ holds and
\item
$\vv \ssX= A \myvec{\vu\ssX\\ d}$ and
$\vv\ssY = B\myvec{\vu\ssX\\ \vu\ssY\\ d} + \vb$.
\end{itemize}
For a sequence of transitions $P=t_1\cdots t_n$,
we write $(q_0,\vu_0) \vdash_{P} (q_n,\vu_n)$
if there is a word $d_1\cdots d_n$ such that
$(q_0,\vu_0)\vdash_{t_1,d_1}
(q_1,\vu_1)\vdash_{t_2,d_2}\cdots \vdash_{t_n,d_n} (q_n,\vu_n)$.
In~this case, we say that {\em $d_1\cdots d_n$ is compatible with $t_1\cdots t_n$}.
As~usual, we write $(q_0,\vu_0)\vdash^{\ast} (q_n,\vu_n)$
if there exists a~sequence of transitions $P$ such that
$(q_0,\vu_0)\vdash_P (q_n,\vu_n)$. 

For an~input word $w=d_1\cdots d_n$,
a {\em run} of $\cA$ on $w$ is a~sequence
$(q_0,\vu_0)\vdash_{t_1,d_1}
(q_1,\vu_1)\vdash_{t_2,d_2}\cdots \vdash_{t_n,d_n} (q_n,\vu_n)$,
where $(q_0,\vu_0)$ is the initial configuration
and $t_1, \ldots, t_n \in \delta$. In this case, we also say $(q_0,\vu_0) \vdash_{\cA, w} (q_n,\vu_n)$. 
The run is {\em accepting} if $q_n \in F$,
in which case $\cA$ outputs the value $g(\vu_n)$,
where $\zeta(q_n)=g$, and we say that $\cA$ accepts $w$.
We write $\cA(w)$ to denote the set of all outputs of $\cA$ on $w$
(i.e., if $\cA$ does not accept $w$, we write $\cA(w)=\emptyset$).

We say that $\cA$ is \emph{deterministic}
if, for any state~$p$ and a~pair of transitions $(p,\varphi(\vx,\cur)) \to (q,A, B,\vb)$ and $(p,\varphi'(\vx,\cur)) \to (q',A', B',\vb')$ starting in~$p$,
the formula $\varphi(\vx,\cur)\wedge \varphi'(\vx,\cur)$ is unsatisfiable. 
$\cA$ is \emph{complete} if, for every state~$p$, the disjunction of guards on
all transitions starting in~$p$ is valid.
We say that $\cA$ is \emph{single-valued} if for every word $w$, $|\cA(w)|\leq 1$.
Note that different input words may yield different outputs.
Evidently, every deterministic $\raq$ is single-valued.

We call $\cA$ {\em copyless}
if the reassignment of its data variables is of the form
$\vy:= A \vy + f(\vx,\cur)$,
where $f$ is a linear function
and $A$ is a~0-1 matrix where 
each column contains at most one $1$.
The intuition is that each variable $y_i$ appears at most once
in the right-hand side of the reassignment.
For example, when $l=2$, the reassignment $\{y_1:=y_1+y_2;y_2:=2x_1\}$
is copyless, while $\{y_1:=y_1+y_2; y_2:=y_1\}$ is not,
since $y_1$ appears twice in the right-hand side.
Our definition of copyless is similar to the one for streaming transducers in~\cite{Alur11}.

Note that the standard register automata (RA) studied in~\cite{KaminskiF94,DemriL09,NSVianu04} 
can be seen as a~special case of $\raq$ without the data variables $Y$.
Moreover, we can view the output function in each final state of an~RA
as a~constant function that always outputs~$1$.
Then, a~standard RA can be seen as a single-valued as well as copyless $\raq$.
Also note that affine programs in the sense of Karr~\cite{MSeidl04,Karr76} are also a special case of $\raq$
in which control variables and input words are ignored.

We present some typical aggregate functions
that can be computed with $\raq$.

\subsubsection*{Computing the minimal value}
The $\raq$ has one control variable $x$, as pictured below.
The output function is $\zeta(q)=x$.
The transition is pictured as $\varphi(\vx,\cur),\{M\}$, where
$\varphi(\vx,\cur)$ is the guard and
$M$ denotes the variable reassignment.
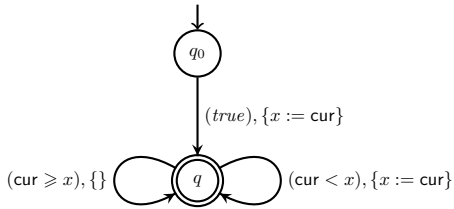
\begin{center}
\begin{tikzpicture}[thick,scale=0.7, every node/.style={transform shape}]

\node[state,initial,initial where=above] (q0)  {$q_0$};
\node[state,,accepting,double distance=1pt] (q) [below=1.5cm of q0] {$q$};

\tikzset{->-/.style={
        decoration={markings,
            mark= at position 0.99 with {\arrow{stealth}} ,
        },
        postaction={decorate}
    }
}

	\draw[->-] (q0) to node[right] {$(\mathit{true}), \{x:=\cur\}$} (q);

	\tikzset{loop/.style={in=330,out=30, looseness=10}}
	\draw[->-,loop below] (q) to node[right] {$(\cur< x),\{x:=\cur\}$} (q);

	\tikzset{loop/.style={in=210,out=150, looseness=10}}
	\draw[->-,loop below] (q) to node[left] {$(\cur\geq x), \{\}$} (q);

\end{tikzpicture}
\end{center}
Intuitively, the $\raq$ starts by storing the first value in $x$.
Every subsequent value $\cur$ is then compared with $x$ and
if $\cur < x$, it is stored in $x$ via the reassignment $x:= \cur$.

\subsubsection*{Computing the second largest element}
The $\raq$ has control variables $x_1$ and $x_2$ and its
output function is $\zeta(q_2)=x_2$.
% The $\raq$ has two control variables $x_1$ and $x_2$ with initial value $0$.
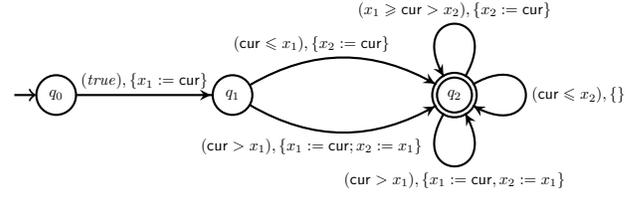
\begin{center}
\begin{tikzpicture}[thick,scale=0.6, every node/.style={transform shape}]

\node[state,initial,initial where=left] (q)  {$q_0$};
\node[state] (q1) [right=3cm of q] {$q_1$};

\node[state,accepting,double distance=1pt] (q2) [right=4cm of q1] {$q_2$};

\tikzset{->-/.style={
        decoration={markings,
            mark= at position 0.99 with {\arrow{stealth}} ,
        },
        postaction={decorate}
    }
}

\draw[->-] (q) to node[above] {$(\mathit{true}), \{x_1:= \cur\}$} (q1);

\draw[->-,bend left=30] (q1) to node[above] {$(\cur \leq x_1), \{x_2:= \cur\}\qquad\qquad$} (q2);
\draw[->-,bend right=30] (q1) to node[below] {$(\cur >  x_1), \{x_1:= \cur;x_2:=x_1\}\qquad\qquad$} (q2);
% \draw[->-] (q1) to node[above] {$(\cur= x_1), \{\}$} (q2);

\tikzset{loop/.style={in=300,out=240, looseness=10}}
\draw[->-,loop below] (q2) to node[below] {$(\cur> x_1),\{x_1:=\cur,x_2:=x_1\}$} (q2);

\tikzset{loop/.style={in=330,out=30, looseness=10}}
\draw[->-,loop below] (q2) to node[right] {$(\cur\leq x_2),\{\}$} (q2);

\tikzset{loop/.style={in=120,out=60, looseness=10}}
\draw[->-,loop below] (q2) to node[above] {$(x_1 \geq \cur > x_2),\{x_2:=\cur\}$} (q2);

\end{tikzpicture}
\end{center}
%\td{The guard $\psi$ is $\cur=x_1 \vee \cur\leq x_2$.}
Intuitively, the $\raq$ stores the first two values in $x_1$ and $x_2$
in a~decreasing order. % with $x_1$ being the largest.
Each subsequent value $\cur$ is compared with $x_1$ and $x_2$,
which are updated if necessary.

\subsubsection*{Computing the number of elements larger than $M$}
The $\raq$ has one control variable $x$ and one data variable $y$ with initial values $M$ and $0$, respectively.
The output function is $\zeta(q_0)=y$.
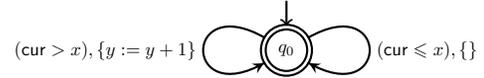
\begin{center}
\begin{tikzpicture}[thick,scale=0.7, every node/.style={transform shape}]

\node[state,initial,accepting,double distance=1pt, initial where=above] (q)  {$q_0$};

\tikzset{->-/.style={
        decoration={markings,
            mark= at position 0.99 with {\arrow{stealth}} ,
        },
        postaction={decorate}
    }
}

	\tikzset{loop/.style={in=330,out=30, looseness=10}}
	\draw[->-,loop below] (q) to node[right] {$(\cur\leq x),\{\}$} (q);

	\tikzset{loop/.style={in=210,out=150, looseness=10}}
	\draw[->-,loop below] (q) to node[left] {$(\cur> x), \{y:=y+1\}$} (q);

\end{tikzpicture}
\end{center}
Intuitively, each input value $\cur$ is compared with $x$.
If $\cur > x$, the $\raq$ increments $y$ by $1$ via the reassignment $y:= y+1$.

\subsubsection*{Computing the number of occurrences of the maximal element}
The $\raq$ has  one control variable $x$ and one data variable $y$, with the initial value $0$.
The output function is $\zeta(q_1)=y$.
\begin{center}
\begin{tikzpicture}[thick,scale=0.6, every node/.style={transform shape}]

\node[state,initial,initial where=left] (q)  {$q_0$};
\node[state,accepting,double distance=1pt] (p) [right=4.5cm of q] {$q_1$};

\tikzset{->-/.style={
        decoration={markings,
            mark= at position 0.99 with {\arrow{stealth}} ,
        },
        postaction={decorate}
    }
}

\draw[->-] (q) to node[above] {$(\mathit{true}), \{x:= \cur;y:=1\}$} (p);

\tikzset{loop/.style={in=300,out=240, looseness=10}}
\draw[->-,loop below] (p) to node[below] {$(\cur< x),\{\}$} (p);

\tikzset{loop/.style={in=330,out=30, looseness=10}}
\draw[->-,loop below] (p) to node[right] {$(\cur= x),\{y:=y+1\}$} (p);

\tikzset{loop/.style={in=120,out=60, looseness=10}}
\draw[->-,loop below] (p) to node[above] {$(\cur> x),\{x:=\cur;y:=1\}$} (p);

\end{tikzpicture}
\end{center}
Intuitively, it stores the first value in $x$
and reassigns $y:=1$.
Every subsequent value $\cur$ is then compared with $x$.
If it is the new largest element, it will be
stored in~$x$ and the contents of $y$ is
reset to $1$.

Proposition~\ref{prop:linear-raq} below will be useful later on.
Intuitively, it states that for a fixed sequence of transitions $t_1\cdots t_n$,
the contents of variables of $\cA$ are a linear combination
of the values in the input word~$d_1\cdots d_n$, 
provided that $d_1\cdots d_n$ is compatible with $t_1\cdots t_n$.
Its proof can be done by a~straightforward induction on~$n$ and is,
therefore, omitted.

\begin{proposition}
\label{prop:linear-raq}
{\bf (Linearity of $\raq$)}
Let $\cA$ be an $\raq$ over $(X,Y)$.
For every sequence $t_1\cdots t_n$ of transitions of $\cA$,
there is a~matrix $M \in \bbQ^{(k+l)\times n}$ and a~vector $\va \in \bbQ^{k+l}$ such that
for every word $d_1\cdots d_n$ compatible with $t_1\cdots t_n$, where
$(q_0,\vu_0)\vdash_{t_1,d_1}\cdots \vdash_{t_n,d_n} (q_n,\vu_n)$,
it holds that
\begin{eqnarray*}
\vu_n & = & M \myvec {d_1\\ \vdots \\ d_n} + \va .
\end{eqnarray*}
\end{proposition}

The following example shows that $\raq$ can be used to represent positive integers succinctly.
Let $p$ and $n$ be positive integers, $k$ be an integer such that $k=\lceil\log n\rceil$, and
$\cA$ be an $\raq$ as illustrated below.
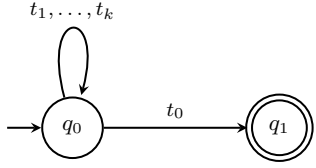
\begin{center}
\resizebox{0.5\linewidth}{!}{
\begin{tikzpicture}[>=stealth,thick,auto,node distance=2cm]

\node[state,initial,initial where=left,inner sep=5pt,minimum size=0pt] (q0) {\small $q_0$};
\node[state,accepting,double distance=1.5pt,inner sep=5pt,minimum size=0pt] (q1) [right=of q0] {\small $q_1$};

\path[->] 
(q0) edge node {\footnotesize $t_{0}$} (q1);

\path[->,every loop/.append style={looseness=15,loop above}]
(q0) edge [loop] node  {\footnotesize $t_{1},\ldots,t_{k}$} (q0);

\end{tikzpicture}}
\end{center}
$\cA$ is over $X=\{x_1,\ldots,x_k\}$ and $Y=\{y\}$.
The initial state of~$\cA$ is~$q_0$ and $q_1$ is its final state.
The initial contents of the variables are $(b_1,\ldots,b_k,1)$,
where $b_k\cdots b_1$ is the binary representation of $n$, i.e.,
$n = \sum_{i=1}^k b_i 2^{i-1}$.

The transition $t_0$ is
$(q_0,\bigwedge_{i=1}^k x_i=0) \to (q_1,\{\})$.
For each $i=1,\ldots,k$, the transition $t_i$ is defined as follows:
\begin{eqnarray*}
& & 
(q_0, x_i = 1 \wedge \bigwedge_{j=1}^{i-1} x_{j}=0)
\\
& &
\to
(q_0,\{x_1:= 1; \ldots; x_{i-1}:=1; x_i := 0; y:= p\cdot y\}).
\end{eqnarray*}
Recall that when a variable's value in a~reassignment is not specified,
its value stays the same.
Therefore, in $t_0$, the contents of all variables stay the same,
while in $t_i$, the contents of $x_{i+1},\ldots,x_k$ stay the same.
We define the output function $\zeta(q_1)$ to output~$y$.

Intuitively, the contents of variables $\vx = (x_1,\ldots,x_k)$ represent
a number between $n$ and~$0$ in binary, where the least significant bit is stored in $x_1$.
$\cA$ starts with $\vx$ containing the binary representation of $n$,
and iterates through all integers from $n$ down to $1$.
On each iteration, it takes one of the transitions $t_1,\ldots,t_k$
that ``decrements'' the number represented by~$\vx$, and
multiplies the contents of $y$ by~$p$.
When the number in $\vx$ reaches $0$, it takes transition $t_0$
and moves to state $q_1$.
Note that the outcome of $\cA$ does not depend on the input word
and it always output $p^n$ regardless on the input.
Moreover, $\cA$~has only $\lceil\log(n)\rceil+1$ transitions.
In fact, one can obtain an $\raq$ that always outputs $p_1^{n_1}\cdots p_k^{n_k}$
by constructing one $\raq$ for each $p_i^{n_i}$
and composing them sequentially.
The final $\raq$ has at most $\sum_{i=1}^k (\lceil\log(n_i)\rceil+1)$ transitions. 

Motivated by the example above, 
we say that {\em an $\raq$ $\cA$ represents a positive integer $n$}
if it has exactly one possible output $n$.
With this representation, $\raq$ can simulate arithmetic circuits as stated below.

\begin{theorem}
\label{theo:circuit}
For every arithmetic circuit $C$ (division-free and without indeterminates),
there is an~$\raq$ $\cA$ that represents the same number as $C$
with the number of transitions linearly proportional to the number of edges in $C$.
If $C$ is additive or multiplicative,
$\cA$ uses only one data variable.
Moreover, $\cA$~can be constructed in time linear
in the size of $C$.
\end{theorem}

The number of transitions in $\cA$ is roughly twice the number of edges in $C$,
plus the number of constants in $C$.

%%%%%%%%%%%%%%%%%%%%%%%%%%%%%%%%%%%%%%%

\section{The invariant problem for $\raq$}
\label{sec:non-zero}

In this section, we will, in the same spirit as Karr~\cite{Karr76}, consider
the {\em invariant} problem for $\raq$.
For an $\raq$ $\cA = \langle Q,q_0,F,\vu_0,\delta,\zeta\rangle$
and a state $q \in Q$,
define $\vAq{\cA}{q} = \{\vv \mid (q_0,\vu_0) \vdash^{\ast} (q,\vv)\}$,
i.e., $\vAq{\cA}{q}$ contains all vectors representing the contents of control and data variables
when $\cA$ reaches state $q$.
The {\em invariant} problem is defined as: 
{\em Given an $\raq$ $\cA$, a state $q$ of $\cA$, and an affine space $\bbH$,
decide whether $\vAq{\cA}{q} \nsubseteq \bbH$}.
Again, an affine space $\bbA=\va+\bbV$ can be represented
either as a pair $(\va,V)$ where $V$ is a basis of $\bbV$,
or as a set $V$ where $\aff(V)=\bbA$. The invariant problem is tightly related
to the \emph{must-constancy} problem for programs~\cite{MR00}, which asks, for
a~given program location $\ell$, a given variable $z$ and a given constant $c$, whether the value of
$z$ in $\ell$ must be equal to $c$. 

Instead of the invariant problem, it will be more convenient to consider another,
but equivalent, problem, which we call the {\em non-zero} problem, defined as
follows. 
{\em Given an $\raq$ $\cA$, decide whether there is $w$ such that $\cA(w)\nsubseteq \{0\}$,
i.e., whether $\cA$ outputs a non-zero value on some word $w$}.
We abuse notation
and simply write $\cA(w) \neq 0$ to denote that there is $c \in \cA(w)$ such that
$c\neq 0$.
The non-zero problem can, therefore, be written as:
{\em Given an $\raq$ $\cA$, decide whether there is $w$ such that $\cA(w)\neq 0$.}

The two problems are, in fact, Karp inter-reducible.
The reduction from the non-zero problem to the invariant problem is as follows.
Let $\cA$ be the input to the non-zero problem,
$q_{1},\ldots,q_{m}$ be the final states of~$\cA$, and $\zeta$ be the mapping that
specifies the output functions for the final states.
Let $\cA'$ be the $\raq$ obtained by adding
a new state $q_f$ into $\cA$, and for every $q_i$ adding
the following transition:
$(q_i, \mathit{true}) \to (q_f, \{y_1 := \zeta(q_i))\})$.
$\cA'$ has only one final state $q_f$, whose
output function yields $y_1$.
The reduction follows by
the fact that there is $w$ such that $\cA(w)\neq 0$ iff
$\vAq{\cA}{q_f} \nsubseteq \bbH$, where $\bbH$ is the space 
of the solutions $\zeta(q_f)(\vx,\vy)=0$.

Vice versa, the invariant problem reduces to the non-zero problem as follows.
Let $\cA$, $q$, and $\bbH=\va+\bbV$ be the input to the invariant problem.
Let $\{\vv_1,\ldots,\vv_m\}$ be a~basis of $\bbV^{\perp}$, the orthogonal complement of $\bbV$,
which can be obtained by Gaussian elimination on a~basis of $\bbV$ in polynomial time.
%\yfc{Define in prelimeary the orthogonal subspace of a vector space. }
Let $\cA'$ be the $\raq$ obtained by adding the following into $\cA$:
\begin{itemize}\itemsep=0pt
\item
$m+1$ new states $q_1,\ldots,q_{m}$ and $p$,
\item
$m+1$ new data variables $y_1,\ldots,y_m$ and $z$,
\item
for each $q_i$, the pair of transitions
$(q, \mathit{true}) \to (q_i, \{y_i := (\myvec{\vx\\ \vy} - \va)\dotprod \vv_i\})$ and
$(q_i, \mathit{true}) \to (p, \{z := y_i\})$.
\end{itemize}
Further, set $p$ as the only final state of~$\cA'$ and set its output function
to yield $z$.
The reduction follows from the fact that
$\vu \in \bbH$ iff $(\vu-\va)\dotprod \vv_i=0$ for every $i=1,\ldots,m$,
thus, $\vAq{\cA}{q} \nsubseteq \bbH$ iff there is a~word~$w$ such that $\cA'(w)\neq 0$.

\subsection{The algorithm and a~small model property}

In this section, we present an exponential-time algorithm
for the non-zero problem of~$\raq$.
Let $\cA=\langle Q,q_0,F,\vu_0,\delta,\zeta\rangle$ be the input $\raq$ over $(X,Y)$,
where $X=\{x_1,\ldots,x_k\}$ and $Y=\{y_1,\ldots,y_l\}$.
The main idea of the presented algorithm is to transform $\cA$ into an affine
program~$\cP_\cA$ and analyse $\cP_\cA$ using Karr's algorithm.

We start with some necessary definitions.
An {\em ordering} of $X$ is a total preorder $\phi$ on~$X$, 
i.e., $\phi=z_{1}\circledast_1 z_{2} \circledast_2 \cdots \circledast_{k-1} z_{k}$,
where $(z_{1},\ldots,z_{k})$ is a permutation of $(x_1,\ldots,x_k)$
and each $\circledast_i$ is either $<$ or $=$.
An ordering $\phi$ is {\em consistent} with 
a transition $(p,\varphi(\vx,\cur))\to (q,A,B,\vb)$,
if $\varphi(\vx,\cur)\wedge \phi$ is satisfiable.
We say that an ordering $\phi$ holds in a~configuration~$(q,\vu)$
if it holds when we substitute $(x_1,\ldots,x_k)$ with $\vu\ssX$.
In this case, we say that the ordering of $(q,\vu)$ is~$\phi$.

The construction of $\cP_\cA$ is based on the following lemma.

\begin{lemma}
\label{lem:reduction-affine}
Let $\bbH$ be an affine space and 
\begin{equation*}
(q_1,\vu_{1}) \ \vdash_{t_1,d_1} \
\cdots\cdots \
\vdash_{t_{m},d_m} \ (q_{m+1},\vu_{m+1})
\end{equation*}
be a run of $\cA$ on a word $d_1\cdots d_m$ 
%for some $t_1,\ldots,t_m$ and $\vu_1,\ldots,\vu_{m+1}$
such that $\vu_{m+1}\notin \bbH$.
Then there is a run of $\cA$ on a word $c_1\cdots c_m$, say 
\begin{equation*}
(q_1,\vv_{1}) \ \vdash_{t_1,c_1} \
\cdots\cdots \
\vdash_{t_{m},c_m} \ (q_{m+1},\vv_{m+1}),
\end{equation*}
such that
$\vu_1=\vv_1$,
$\vv_{m+1}\notin \bbH$, and for every $i=1,\ldots,m+1$ the following holds:
\begin{enumerate}[(a)]\itemsep=0pt
\item
$(q_i,\vu_i)$ and $(q_i,\vv_i)$ have the same ordering.
\item
If $d_i = \vu_{i}\ssX(j)$ for some $j$,
then $c_i = \vv_{i}\ssX(j)$.
\item
If $d_i < \vu_{i}\ssX(j)$, where $\vu_{i}\ssX(j)$ is the minimal value in $\vu_{i}\ssX$,
then either $c_i = \vv_{i}\ssX(j)-1$ or $c_i=\vv_{i}\ssX(j)-2$.
\item
If $d_i > \vu_{i}\ssX(j)$, where $\vu_{i}\ssX(j)$ is the maximal value in $\vu_{i}\ssX$,
then either $c_i = \vv_{i}\ssX(j)+1$ or $c_i=\vv_{i}\ssX(j)+2$.
\item
If $\vu_{i}\ssX(j)< d_i < \vu_{i}\ssX(j')$, where

$\vu_{i}\ssX(j)$ is the maximal value in $\vu_{i}\ssX$
less than $d_i$ and

$\vu_{i}\ssX(j')$ is the minimal value in $\vu_{i}\ssX$
greater than $d_i$,

then either $c_i = \frac{1}{3}\vv_{i}\ssX(j) + \frac{2}{3}\vv_{i}\ssX(j')$ 
or $c_i=\frac{2}{3}\vv_{i}\ssX(j) + \frac{1}{3}\vv_{i}\ssX(j')$.
\end{enumerate}
\end{lemma}

Intuitively,
Lemma~\ref{lem:reduction-affine} states that
% for a run $(q_1,\vu_1)\vdash^* (q_{m+1},\vu_{m+1})$ such that
% OL: |-* is not a run, but reachability relation
for a run $(q_1,\vu_1)\vdash_{\cA, w} (q_{m+1},\vu_{m+1})$ on a word $w=d_1\cdots d_m$ such that
$\vu_{m+1}$ does not belong to the affine space $\bbH$,
we can assume that 
each $d_i$ is a linear combination of components in $\vu_i$.
We note that
the choice of the constants $\pm1, \pm2, \frac{1}{3}, \frac{2}{3}$ in items (c)-(e) is arbitrary and was made to ensure that there are at least two possible different values for~$c_i$,
since by Lemma~\ref{lem:affine-space}, one of them is guaranteed to hit outside~$\bbH$.
Other constants satisfying the right conditions would work, too.

With Lemma~\ref{lem:reduction-affine}, we then transform $\cA$ to an affine program~$\cP_\cA$
and apply Karr's algorithm on $\cP_\cA$ to decide the non-zero problem.
Essentially, the set of states in $\cP_\cA$ is the Cartesian product of $Q$ and the set of orderings of $X\cup\{\cur\}$.
The number of variables in $\cP$ is $k+l$.
There are altogether $2^k(k+1)!$ orderings of $X\cup \{\cur\}$,
so the algorithm runs in an exponential time, as stated in Theorem~\ref{theo:non-zero-exptime} below.

\begin{theorem}
\label{theo:non-zero-exptime}
The non-zero problem for $\raq$ is in $\exptime$.
\end{theorem}

Note that the number of states in the affine program $\cP_\cA$
is $|Q|2^k (k+1)!$.
By Remark~\ref{rem:karr-small-model},
we can obtain a~similar small model property for~$\raq$,
as stated below.

\begin{theorem}
\label{theo:small-model-snt}
{\bf (A small model property for $\raq$)}
If there is a~word $w \in \bbQ^*$ such that $\cA(w)\neq 0$,
then there is a~word $w' \in \bbQ^*$ such that $\cA(w')\neq 0$ and $|w'|\leq |Q|(k+l+1)2^k (k+1)!$.
\end{theorem}

One can also prove Theorem~\ref{theo:small-model-snt} without relying on Karr's
algorithm.
Moreover, we can show that the exponential bound is, in fact, tight
(see the appendix for proofs of both claims).

We should remark that the exponential complexity is in the bit model,
i.e., rational numbers are represented in their bit forms.
As stated in Theorem~\ref{theo:circuit},
an $\raq$ can simulate an~arithmetic circuit and store in its data variables
values that are doubly-exponentially large (w.r.t.~the number of control variables), which
occupy an~exponential space.
For example, if the initial value of a~data variable~$y$ is 1 and
every transition contains the reassignment $y := 2y$, 
the final value of~$y$ may be up to $2^{|Q|(k+l+1)2^k (k+1)!}$.
However, if we assume that rational numbers between $-1$ and $1$ occupy only constant space,
the non-zero problem is in $\pspace$ (by guessing a path of exponential length as in Theorem~\ref{theo:small-model-snt}),
which matches the non-emptiness problem of standard RA~\cite{DemriL09}.

\subsection{Polynomial-space algorithm for copyless $\raq$}

In the following, let $\cA$ be a copyless $\raq$ with $k$ control variables
and $l$ data variables.
W.l.o.g.,
we assume that every transition of~$\cA$ is of the form $(p,\varphi(\vx,\cur))\to (q,A,B,0)$,
i.e., $\vb=0$ (every $\raq$ can be transformed to this form by adding new
control variables to store the non-zero constants in $\vb$).
Recall that copyless $\raq$ are still a~generalization of standard RA,
thus, the non-zero problem is $\pspace$-hard.
In the following we will show that the non-zero problem is in $\pspace$.
We need the following lemma.

\begin{lemma}
\label{lem:forgetful-raq}
Let $\bbH$ be an affine space and 
\begin{equation*}
(q_1,\vu_{1}) \ \vdash_{t_1,d_1} \
\cdots\cdots \
\vdash_{t_{m},d_m} \ (q_{m+1},\vu_{m+1}),
\end{equation*}
be a run of $\cA$ on a word $d_1 \cdots d_m$ 
such that $\vu_{m+1}\notin \bbH$.
Then, there exists a run of $\cA$ on a word $c_1 \cdots c_m$, say
\begin{equation*}
(q_1,\vv_{1}) \ \vdash_{t_1,c_1} \
\cdots\cdots \
\vdash_{t_{m},c_m} \ (q_{m+1},\vv_{m+1}),
\end{equation*}
such that
$\vu_1=\vv_1$,
$\vv_{m+1}\notin \bbH$, and for every $i=1,\ldots,m+1$,
if $c_i$ does not appear in $\vv_i\ssX$,
then $c_i \neq c_{j}$ for every $j \leq i-1$.
\end{lemma}

Intuitively, Lemma~\ref{lem:forgetful-raq} states that for the non-zero problem,
it~is sufficient to consider only words $c_1\cdots c_m$
such that if $\cA$~encounters a value $c_i$ that is not in its control variables,
then $c_i$ is indeed new, i.e., it has not appeared in~$c_1\cdots c_{i-1}$.
Another way of looking at it is that
once $\cA$ ``forgets'' a value~$c_i$,
i.e., $c_i$ no longer appears in its control variables,
then $c_i$ will never appear again in the future.
 
We start with the following observation.
Let $w=d_1\cdots d_m$ be a word.
Suppose $(q_0,\vu_0)\vdash_{\cA,w} (q_m,\vu_m)$,
where $q_m$ is a final state.
By linearity of $\raq$ (cf.~Proposition~\ref{prop:linear-raq}),
$\vu_m = M \myvec {d_1 \cdots d_m}^t + \va$,
for some~$M$ and~$\va$.
Let $\zeta(q_m)(\vx,\vy)=\vc\dotprod \myvec{\vx \\ \vy}+b$.
Then, for some $\alpha_1,\ldots,\alpha_m,\beta$,
\begin{eqnarray*}
\zeta(q_m)(\vu_m) & = &
\alpha_1 d_1 + \cdots + \alpha_m d_m +\beta.
\end{eqnarray*}
Suppose $d_1',\ldots,d_n'$ are the distinct values occurring in $d_1\cdots d_m$.
Therefore, for some $\alpha_1',\ldots,\alpha_n'$, it holds that
\begin{eqnarray*}
\zeta(q_m)(\vu_m) & = &
\alpha_1' d_1' + \cdots + \alpha_n' d_n' +\beta .
\end{eqnarray*}
We can assume that the values $d_1',\ldots,d'_{k}$ are
the initial contents of control variables.
For simplicity, we can also assume that the initial contents of control
variables are pairwise different and that all initial values stored in control variables occur
in $d_1\cdots d_m$.
We observe the following:
\begin{itemize}\itemsep=0pt
\item
If there is $i> k$ such that $\alpha_i' \neq 0$,
then we can assume that $\zeta(q_m)(\vu_m) \neq 0$.
To show why, suppose to the contrary that $\zeta(q_m)(\vu_m) = 0$.
From the assumption, the value $d_i'$~does not appear in the initial contents of control variables.
When $d_i'$ first appears in the input word, 
by density of rational numbers,
we can increase $d_i'$ by some small number $\epsilon>0$,
i.e., replace $d_i'$ with $d_i'+\epsilon$,
and still obtain a run from $q_0$ to $q_m$.
The output will now be
$$
\alpha_1' d_1' + \cdots + \alpha_i'(d_i'+\epsilon)+\cdots + \alpha_n' d_n' +\beta
=\alpha_i' \epsilon,
$$
which will be non-zero, since both $\epsilon$ and $\alpha_i'$ are non-zero.
\item
If, for all $i>k$, it holds that $\alpha_i' = 0$,
then $\zeta(q_m)(\vu_m)\neq 0$ if and only if
$\alpha_1' d_1'+\cdots+\alpha_k'd'_{k} +\beta \neq 0$.
\end{itemize}
Note that our observation above holds for general $\raq$.
In general, the number of bits for storing $\alpha_i'$ can be exponentially large,
but as we will see later, is polynomially bound for copyless $\raq$.

The algorithm works by
simulating a~run of $\cA$ of length at most $|Q|(k+l+1)2^{k}(k+1)!$
starting from the initial configuration.
During the simulation, when $\cA$ assigns new values into control variables,
the algorithm only remembers the ordering of control variables, not the actual
data values assigned.
Such an ordering is sufficient to simulate a~run.
The algorithm will then try to nondeterministically guess the first position of the word where a value $d_i'$  such that $i > k$ and $\alpha_i' \neq 0$ occurs.
Again, the algorithm does not guess the actual value $d_i'$
but, instead, only remembers the names of the control variables that~$d'_i$ is assigned to.

In the rest of the simulation, the algorithm keeps track of how many ``copies'' of $d_i'$
have been added to each data variable.
For example, suppose $d_i'$ is stored in a~control variable~$x_j$
and the reassignment for~$y$ in a transition~$t$ is of the form $y:= y+ y' + c' x_j$.
Then, the number of copies of $d_i'$ in $y$ is obtained by adding the number of copies of~$d_i'$ in~$y$ and~$y'$,
plus $c'$ copies of $d_i'$.
Note that due to being copyless, 
the assignment to $y'$ in $t$ cannot use the original value of $y'$, which is lost.

When the value $d_i'$ is forgotten in $\cA$, i.e., $d'_i$ is not stored in any
control variable any more, we can assume $d'_i$ will not appear again in the
input word (by Lemma~\ref{lem:forgetful-raq}).
During the simulation, the algorithm keeps for every data variable~$y$ a~track of how many copies of $d_i'$
are stored in~$y$.
Every time the algorithm reaches a~final state, it applies the output function
of the state and checks whether~$\alpha_i' \neq 0$.

Based on the property that $\cA$ is copyless,
we notice that in one step, the sum of {\em all} data variables may increase by at most $f(\vx,\cur)$,
for some linear function $f$, where the constants in $f$ come from
those in the transition.
By Theorem~\ref{theo:small-model-snt}, during any run,
the number of bits occupied by the sum of 
the ``coefficients'' of $d'_i$ in {\em all} data variables is at most
$c\cdot \log(|Q|(k+l+1)2^k(k+1)!)$,
where $c$ is the sum of all bits occupied by the constants in the transitions.
Thus, each $\alpha_i'$ occupies only a~polynomial space. 

If for all $i > k$ it holds that $\alpha_i'= 0$,
the algorithm counts $\alpha_1',\ldots,\alpha_k'$ instead.
Recall that $d_1',\ldots,d_k'$ are the data values of the initial contents of control variables.
The algorithm performs the counting during the simulation of a~run in a~similar way as above.
When $d_i'$ no longer appears in any control variable,
the counting stops and the simulation simply continues
by remembering the state and the ordering of control variables.
When a~final state is reached, the algorithm verifies that
$\alpha_1'd_1' + \cdots + \alpha_k'd_k' +\beta \neq 0$.
Again, each $\alpha_i'$ occupies only a~polynomial space,
thus, the whole algorithm runs in a~polynomial space.
Since the non-emptiness problem for standard RA
is already $\pspace$-hard,
we conclude with the following theorem.

\begin{theorem}
\label{theo:copyless}
The non-zero problem for copyless $\raq$
is $\pspace$-complete.
\end{theorem}

It is tempting to directly use Theorem~\ref{theo:small-model-snt}
to prove Theorem~\ref{theo:copyless} by simulating the run directly instead of tracing the coefficients of input values.
In doing so, however,
the number of bits may increase in each step.
For example, suppose an~$\raq$ has two control variables $x_1$ and $x_2$
storing $0.01$ and $0.1$ (in binary), respectively,
and its transitions have the guard $x_1< \cur < x_2$ with 
the reassignment $\{x_1:= \cur\}$.
Straightforward guessing by adding one bit $1$ at the end of $x_1$
will result in the number of bits in $x_1$ increasing by one in each step.
A~similar thing can happen if
guessing a path in the affine program generated from
Lemma~\ref{lem:reduction-affine}: the numbers cannot be represented in
a~polynomial space because of the multiplication with $\frac{1}{3}$ and~$\frac{2}{3}$.

Furthermore, note that the algorithm is correct also for general (i.e., not only
copyless) $\raq$.
The restriction to copyless~$\raq$ allows us to guarantee that the space used
by the algorithm is polynomial.
If we applied the algorithm to general $\raq$, it would require an exponential space.

\subsection{Some remarks on the non-zero problem}

In this section, we have shown that the non-zero problem for $\raq$ is in $\exptime$,
and the problem itself is a generalization of the non-emptiness problem
for standard RA.
Our algorithm relies heavily on the fact that
the variables are partitioned into two groups:
control variables, which ``control'' the computation flow,
and data variables, which accumulate data about the input word.
Without control variables, an~$\raq$ is similar to an affine program,
thus the non-zero problem drops to $\ptime$.
Without data variables,
the non-zero problem becomes $\pspace$-complete,
as an~$\raq$ without data variables is a special case of a~copyless $\raq$
but still a generalization of a~standard~RA.

It is also important that $\raq$ have no access to data variables at all.
If we allow comparison between two data variables,
the non-zero problem becomes undecidable.
In fact, even if we allow $\raq$ to access only one bit of information from data variables,
say, the least significant bit of the integer part of a rational number,
$\raq$ can simulate Turing machines.
In particular,
the contents of a Turing machine (two-way infinite) tape
can be represented as a rational number.
For example, if the contents of the tape is $\sqcup^{\omega}$0\underline{0}1$\sqcup^{\omega}$
(with the underline indicating the position of the head and $\sqcup$ denoting a~blank space),
its representation by a rational number can be e.g.~$1010.11$,
where $0,1$, and $\sqcup$ are encoded by $10,11$, and $00$, respectively.
The head moving right and left can be simulated by multiplying the data variable
by 4 and $\frac{1}{4}$, respectively.

%%%%%%%%%%%%%%%%%%%%%%%%%%%%%%%%%%%%%%%

\section{The equivalence and commutativity problems}
\label{sec:three-problem}

In this section we study the equivalence and the commutativity problems.
The equivalence problem is defined as follows:
{\em Given two $\raq$ $\cA$ and $\cA'$, decide if $\cA(w)=\cA'(w)$ for all words $w$}.
On the other hand, the commutativity problem is defined as follows:
{\em Given an $\raq$ $\cA$, decide if for all words $w$ and $w'$ such that
$w'\in perm(w)$, it holds that $\cA(w)=\cA(w')$,
where $perm(w)$ is the set of all permutations of the word $w$.}

\begin{theorem}
For single-valued $\raq$, the commutativity and equivalence problems are both undecidable.\label{thm:raq-undecidable}
\end{theorem}
Theorem~\ref{thm:raq-undecidable} can be proved by a~reduction similar to the one used
in~\cite{NSVianu04} for proving undecidability of the equivalence problem for
standard RA. 
For deterministic $\raq$, however,
both problems become inter-reducible (via a~Cook reduction) with the non-zero problem,
as stated below.

\begin{theorem}
\label{thm:detraq}
For deterministic $\raq$, 
the equivalence problem,
the commutativity problem,
the non-zero problem, and
the invariant problem are all inter-reducible
in polynomial time. 
\end{theorem}

In the following paragraph, we present the main ideas of the proofs.

\subsubsection*{From non-zero to invariant and vice versa}

Note that the Karp reductions from Section~\ref{sec:non-zero} cannot be used,
because they construct nondeterministic~$\raq$.
Instead, we modify them into Cook reductions such that for every final state of
the $\raq$ in the non-zero problem, we create one invariant test.
On the other hand, for the invariant problem, suppose that $\bbH = \va + \bbV$ and
let us take a basis $\{\vv_1,\ldots, \vv_m\}$ for $\bbV^\bot$ (the orthogonal
complement of $\bbV$).
Then for each $\vv_i$ in the basis, we create a~new
$\raq$ with a single final
state with a corresponding output function.
Notice that the reductions preserve the (deterministic) structure of the $\raq$.

\subsubsection*{From equivalence to non-zero}
The proof is by a~standard product construction.
Given two deterministic $\raq$ $\cA_1$ and~$\cA_2$ (w.l.o.g. we assume they are
both complete), we can construct in a~polynomial time a deterministic~$\raq$ $\cA$
such that $\cA_1$ and $\cA_2$ are equivalent iff $\cA(w)=0$ for all $w$. 
The states of $\cA$ are of the form $(q_1,q_2)$, where $q_1$ is a state from $\cA_1$ and
$q_2$ from $\cA_2$.
A~state~$(q_1,q_2)$ is final iff at least one of $q_1$ and $q_2$ is final, and the
output function
is defined as either $(i)$ the difference of the outputs of $q_1$ and $q_2$ if both
$q_1$ and $q_2$ are final, or
$(ii)$ the constant $1$ if exactly one of them is final.

\subsubsection*{From non-zero to commutativity}
Let $\cA$ be a deterministic $\raq$.
We assume w.l.o.g. that
for all $w$, $|\cA(w)| =1 $, i.e., $\cA$ outputs a~value on all inputs~$w$.
We construct a deterministic $\raq$ $\cA'$
with outputs defined as follows:
\begin{itemize}\itemsep=0pt
\item
For words where the first and the second values are $1$ and $2$ respectively, i.e., words
of the form $v = 12w$, we define $\cA'(v) = \cA(w)$
\item
For all other words, $\cA'$ outputs $0$.
\end{itemize}
The construction of $\cA'$ takes only linear time
by adding two new states that check the first two values and
a~new final state that outputs the constant $0$ for words not in the form of $12w$.
If~there is a~word~$w$ such that $\cA(w)\neq 0$, then $\cA'(12w) \neq 0$, so
$\cA'$ is not commutative (because $\cA'(21w) = 0$).
On the other hand, if $\cA'$ is not commutative, it means that there is an
input for which the output is not 0.

\subsubsection*{From commutativity to equivalence}
The idea of the proof is similar to the one used in~\cite{ChenHSW15} to prove decidability of
the commutativity problem of two-way finite automata.
A~similar idea was also used in~\cite{ChenDaBeast16} for the same problem
over \emph{symbolic numerical transducers}, which are a strict subclass of~$\raq$.

We define two permutation functions $\pi_{1}$ and $\pi_{2}$ on words as follows:
let
$\pi_{1}(d_1d_2\cdots d_n) = d_2d_1\cdots d_n$ (swap the first two symbols)
and
$\pi_{2}(d_1d_2\cdots d_n) = d_2\cdots d_nd_1$ (move the first symbol to the end of the input word).
It is known that every permutation is a composition of $\pi_1$ and $\pi_2$~\cite{Hungerford03}.

Given a deterministic $\raq$~$\cA$, it holds that
$\cA$ is commutative iff
the following equations hold for every word $w$:
\begin{equation*}
  \cA(w) \ = \ \cA(\pi_1(w)) \ = \ \cA(\pi_2(w))
\end{equation*}
As a~consequence, we can reduce
the commutativity problem to the equivalence problem
by constructing deterministic $\raq$ $\cA_1$ and $\cA_2$ such that for every word $w$,
$\cA_1(w)=\cA(\pi_1(w))$ and $\cA_2(w)=\cA(\pi_2(w))$.

While the construction of $\cA_1$ is straightforward, the construction of
$\cA_2$ is more involved.
The standard way to construct $\cA_2$ involves nondeterminism
to ``guess'' that the next transition is the last one.
However, here we require $\cA_2$ to be deterministic.
Our trick is to use a~new set of variables to simulate the process of guessing
in a~deterministic way.

\begin{corollary}
\label{theo:comm-exptime}
The commutativity and equivalence problems for deterministic $\raq$ are in
$\exptime$.
They become $\pspace$-complete for deterministic copyless $\raq$.
\end{corollary}

All upper bounds follow from the results in Section~\ref{sec:non-zero}.
The $\pspace$-hardness can be obtained using a~reduction similar to the one
in~\cite{DemriL09}.
Moreover,
we can use the ideas in the proof of Theorem~\ref{thm:detraq} also for standard RA to obtain the following corollary.

\begin{corollary}
\label{theo:detra}
The commutativity problem for deterministic
RA is $\pspace$-complete.
\end{corollary}

%%%%%%%%%%%%%%%%%%%%%%%%%%%%%%%%%%%%%%%

\section{The reachability problem}
\label{sec:other}

The reachability problem is defined as follows:
{\em Given an $\raq$ $\cA$, decide if there is a~word $w$ such that $0 \in \cA(w)$.}
The reachability problem is tightly related to the \emph{may-constancy} problem
for programs \cite{MR00}, which asks, for a~given program location $\ell$,
a~given variable $z$, and a~given constant $c$,  whether the value of $z$ in
$\ell$ may be equal to $c$, that is, there is an execution path leading to
$\ell$ such that the value of $z$ in $\ell$ is~$c$.

\begin{theorem}\label{thm-reach-und}
The reachability problem for $\raq$ is undecidable, even for deterministic~$\raq$.
\end{theorem}

The proof of Theorem~\ref{thm-reach-und} is obtained by a reduction from
PCP~\cite{Post46}. 
On the other hand, we show that for copyless $\raq$ with non-strict guards, the reachability problem is decidable.

Let $X$ be a set of control variables. A transition guard $\varphi(\vx, \cur)$ is \emph{non-strict}
if it does not contain negations, i.e., it is a positive Boolean combination of inequalities 
$z \leq z'$ for $z, z' \in X \cup \{\cur\}$. 
An $\raq$ $\cA$  is said to have \emph{non-strict transition guards} if the guard in each transition of $\cA$ 
is non-strict.

\begin{theorem}\label{thm-reach-dec}
The reachability problem for (nondeterministic) copyless $\raq$ with non-strict transition guards is in \nexptime.
\end{theorem}

The rest of this section is devoted to the proof of Theorem~\ref{thm-reach-dec}.
Suppose $\cA=\langle Q,q_0,F,\vu_0,\delta,\zeta\rangle$
is a copyless $\raq$ with non-strict transition guards over $(X,Y)$, where $X=\{x_1,\ldots,x_k\}$ and $Y=\{y_1,\ldots,y_l\}$. Let
$\cN$ be the set of constants appearing in $\vu_0\ssX$.
For simplicity, we assume that all control variables initially
contain different values.

Suppose there is a word $w=d_1\cdots d_n$ that leads to a~zero output.
Let 
$(q_0,\vu_0)\vdash_{t_1,d_1}
(q_1,\vu_1)\vdash_{t_2,d_2}\cdots \vdash_{t_n,d_n} (q_n,\vu_n)$
be the run of $\cA$ on $w$.
By Proposition~\ref{prop:linear-raq}, there are $M$ and $\vb$ such that
\begin{eqnarray*}
\vu_n & = & M \myvec {d_1\\ \vdots \\ d_n} + \vb.
\end{eqnarray*}
The values $d_1,\ldots,d_n$ satisfy a~set of inequalities
imposed by the transitions $t_1,\ldots,t_n$. Let $\Phi(\vz)$ denote the conjunction of those inequalities, where $\vz = (z_1,\ldots, z_n)^t$ are variables representing the data values $d_1,\ldots,d_n$. 
For simplicity, we assume that the guards in $t_1,\dots, t_n$ contain \emph{no disjunctions}, which means that the set of points (vectors) satisfying $\Phi(\vz)$ is a~convex polyhedron.

Suppose the output function of $q_n$ is $\zeta(q_n) = \va\cdot \myvec{\vx\\ \vy} + a'$.
We define the following function:
\begin{eqnarray*}
f(\vz) & = & \va \cdot M \myvec {z_1\\ \vdots\\ z_n} +\va\cdot \vb +a'.
\end{eqnarray*}
Thus, by our assumption that $d_1\cdots d_n$ leads to zero,
we have:
\begin{eqnarray*}
f((d_1,\ldots,d_n)^t) = 0& \wedge & \Phi((d_1,\ldots,d_n)^t)=\ltrue,
\end{eqnarray*}
which is equivalent to: 
\begin{eqnarray}\label{eq:reach_bound}
	\exists\vz_1,\vz_2 \in \bbQ^{n}: f(\vz_1)  \leq   0  \leq  f(\vz_2) \wedge  \Phi(\vz_1)\wedge\Phi(\vz_2).
\end{eqnarray}
Observe that (\ref{eq:reach_bound}) holds iff the following two constraints hold simultaneously:
\begin{description}\itemsep=0pt
\item[{\bf [F1]}]
the infimum of $f(\vz)$ w.r.t. $\Phi(\vz)$ is $\leq 0$,
\item[{\bf [F2]}]
the supremum of $f(\vz)$ w.r.t. $\Phi(\vz)$ is $\geq 0$.
\end{description}
From the Simplex algorithm for linear programming~\cite{chvatal},
we know that the points that yield the optimum, i.e., the infimum and the supremum,
are at the ``corner'' points of convex polyhedra.
The constraints in $\Phi(\vz)$ only contain the constants from $\cN$ (as a result of the fact that the initial contents of control variables are a~fixed vector of constants),
so the corner points of the convex polyhedron of $\Phi(\vz)$ only take values from the set $\cN\cup\{-\infty,+\infty\}$.

To establish constraints {\bf F1} and {\bf F2}, it is sufficient to find 
two corner points $\vz_1$ and $\vz_2$ such that
$f(\vz_1)\leq 0 \leq f(\vz_2)$.
To find these two points, we will construct
a~corresponding \mbox{$\bbQ$-VASS} (rational vector addition systems with states), where the configuration reachability can be decided in $\nptime$.

In the following, we shows how to construct the $\bbQ$-VASS
from  $\cA$.
For simplicity of presentation, we make the following assumptions:
\begin{itemize}
\item $\cA$ is {\em order-preserving} on $X$.
That is, at all times the contents of control variables must satisfy the constraint

\smallskip
\hspace{18mm}$
x_1 \ \le \ x_2 \ \le \ \cdots \ \le \ x_k$.
\smallskip

\item The reassignments of data variables
are of the form $y_j := y_j + f(\vx,\cur)$ for each $y_j \in Y$.
\end{itemize}
The construction can be generalized to arbitrary copyless $\raq$ with non-strict guards without the two assumptions.

Moreover, we can ``split'' each transition of $\cA$ into several ones by pinpointing the place of $\cur$ w.r.t. the linear order $x_1 \le x_2 \le \cdots \le x_k$, so that the guard in each transition is of the form: $\cur= x_i$,
$\cur \le x_1$, $x_i \le \cur \le x_{i+1}$, or $x_k \le \cur$.

Let $\vu_0\ssX=(c_1,\dots,c_k)^t$. Then $\cN= \{c_1,\ldots,c_k\}$ and $c_1 < \dots < c_k$.
Let $\cN_{\infty}= \{-\infty, +\infty\}\cup \cN$.
A {\em specification} is a mapping $\eta$ from $X$ to $\cN_{\infty}$
that respects the ordering of $\cN_{\infty}$,
i.e., for $i \le j$, $\eta(x_i) \le \eta(x_j)$.
Intuitively, $\eta$ encodes the value of $x_i$ in a corner point.
We have $\eta(x_i)=c_j$ when $x_i$ is either assigned to $c_j$ or to a~value arbitrarily close to $c_j$.

We will construct a $2l$-dimensional $\bbQ$-VASS $(S,\Delta)$
with variables $\vy_1 = (y_{1,1},\ldots,y_{1,l})$ 
and $\vy_2 = (y_{2,1},\ldots,y_{2,l})$ as follows.
The set of states $S$ of the $\bbQ$-VASS is 
$Q\times \{(\eta_1,\eta_2) \mid \eta_1,\eta_2 \ \textrm{are specifications}\}$.
A \emph{configuration} is of the form $((q,\eta_1,\eta_2),\vy_1,\vy_2)$, where
$(\eta_1,\vy_1)$ and $(\eta_2,\vy_2)$ summarize the information of the
components of the two corner points that have been acquired so far (in other
words, the input data values that have been read by the~$\raq$ so far). 
The details of the transition relation $\Delta$ can be found in the appendix.

Consider the initial configuration 
$((q_0,\eta,\eta),\vu_0\ssY,\vu_0\ssY)$,
where $\eta(x_i)=\vu_0(x_i)$ for each $x_i \in X$.
It holds that there is $w$ such that $0 \in \cA(w)$
iff there is a configuration
$((q',\eta_1,\eta_2),\vv_1,\vv_2)$
reachable from the initial configuration
such that $q' \in F$ and one of the following holds.

\medskip
\hspace{1cm} $
\zeta(q')(\eta_1(\vx),\vv_1) \ \leq \ 0
\ \leq \
\zeta(q')(\eta_2(\vx),\vv_2)
$
\medskip

or

\smallskip
\hspace{12mm}$
\zeta(q')(\eta_2(\vx),\vv_2) \ \leq \ 0
\ \leq \
\zeta(q')(\eta_1(\vx),\vv_1).
$
\smallskip

The existence of such a~configuration
can be encoded as configuration reachability in the constructed $\bbQ$-VASS,
which, in turn, can be reduced to satisfiability of 
an existential Presburger formula.

%%%%%%%%%%%%%%%%%%%%%%%%%%%%%%%%%%%%%%%

\section{Related Work}\label{sec:related}

The literature provides many different formal models with registers or arithmetics.
Here we just mention those that are closely related to $\raq$.
One of the most general models with registers and arithmetics are \emph{counter automata}~\cite{Minsky61} (over finite alphabets), which
are essentially finite automata equipped with a~bounded number of registers
capable of holding an integer, which can be tested and updated using
Presburger-definable relations.
General counter automata with two or more registers are
Turing-complete~\cite{Minsky61}, which makes any of their non-trivial problems
undecidable.

One way of restricting the expressiveness of counter automata to obtain
a~decidable model are the so-called \emph{integer vector addition systems with
states} ($\bbZ$-VASS)~\cite{Haase14}, where testing values of registers is forbidden and the only allowed updates to a~register are addition or
subtraction of a~constant from its value.
This restriction makes the configuration reachability problem for $\bbZ$-VASS much easier
($\nptime$-complete) and the \emph{equivalence of reachability sets} problem decidable
(co$\nexptime$-complete).
For completeness, we also mention \emph{vector addition systems with states}
(denoted as VASS without the initial~$\bbZ$), where registers can only
hold values from~$\bbN$ (and thus
transitions that would decrease the current value below zero are disabled).
This makes VASS equivalent to Petri nets.
In VASS, configuration reachability is $\expspace$-hard~\cite{Lipton76} (but
decidable~\cite{Mayr84}) and equivalence is undecidable~\cite{Hack76}.

Another way of restricting counter automata to decidable subclasses is via their
structure.
One important subclass of this kind are the so-called \emph{flat counter
automata}~\cite{LG05}, i.e., counter automata without nested loops, where
configuration reachability and equivalence are decidable.

\emph{Register automata}
(RA)~\cite{ShemeshF94,KaminskiF94,NSVianu04,DemriL09}---sometimes also called
\emph{finite-memory automata}---
is a model of automata over infinite alphabets
where registers can store values
copied from the input and transition guards
can only test equality between the input value and the values stored in registers.
For (nondeterministic) RA, the emptiness problem is $\pspace$-complete, while the inclusion, equivalence, and universality problems are all
undecidable.
Register automata can also be extended~\cite{Figueira10} to allow transition
guards to test the order relation between data values (denoted by RA$_\le$), in which case they are able to simulate \emph{timed
automata}~\cite{Alur94} by encoding timed words with data words. The model of $\raq$ can be seen as an extension of RA$_\le$ with data variables and linear arithmetics on them.
There is also another RA model over the alphabet $\bbN$
with order and successor relations in guards, but no arithmetic on the input word~\cite{BLT17}.

As mentioned in the introduction, the model of $\raq$ is inspired by the model of
\emph{streaming data string transducers} (SDST), proposed by Alur and
\v{C}ern\'{y} in \cite{Alur11}.
SDST are an extension of  
deterministic RA$_\le$ with \emph{data string variables (registers)}, which 
%that are write-only 
can hold data strings obtained by concatenating some of the input values that
have been read so far.
There are two major restrictions imposed on the data strings variables of SDST:
$(i)$~they are \emph{write-only}, in the sense that they are
forbidden to occur in transition guards, and
$(ii)$~the reassignments that update them are \emph{copyless}.
These two restrictions are essential for obtaining the $\pspace$-completeness
result of the equivalence problem for SDST. 

\emph{Cost register automata} (CRA)~\cite{Alur13} is a model over finite
alphabets where a~finite number of \emph{cost registers} are used to store
values from a (possibly infinite) cost domain, and these cost registers are
updated by using the operations specified by \emph{cost grammars}.
% which are interpreted on cost domains.
A cost domain and a cost grammar, together with its interpretation on the cost
domain, are called a \emph{cost model}.
An example of a~cost model is $(\bbQ, +)$, where the cost domain is $\bbQ$, the
set of rational numbers, and the cost grammar is the set of linear arithmetic
expressions on $\bbQ$, with $+$ interpreted as the addition operation on
$\bbQ$.
Decidability and complexity of decision problems for CRA depend on the underlying \emph{cost model}.
For instance, the equivalence problem for CRA over the $(\bbQ, +)$ cost model is
decidable in $\ptime$, while, on the other hand, for CRA over the $(\bbN,
\min, +c)$ cost model (which are equivalent to weighted automata), the equivalence problem becomes undecidable.

The work related closest to $\raq$ are \emph{streaming numerical transducers}
(SNT) introduced in our previous work~\cite{ChenDaBeast16} for investigating
the commutativity problem of Reducer programs in the MapReduce
framework~\cite{DeanG04}.
The model of SNT is a~strict subclass of $\raq$ that satisfies several
additional constraints; in particular, SNT are copyless and their transition
graph is \emph{deterministic} and  \emph{generalized flat} (any two loops share
at most one state).
In~\cite{ChenDaBeast16},  by using a completely different proof strategy than
in the current paper, we provided an exponential-time algorithm for the
non-zero, equivalence, and commutativity problems of SNT. 
We did not consider the reachability problem for SNT in~\cite{ChenDaBeast16}.

\emph{Weighted register automata} (WRA)~\cite{Babari16} is a model that
combines register automata with \emph{weighted automata}~\cite{Schutzenberger61}.
Using the framework of this paper, the model of WRA can be seen as a variant of
$\raq$ with
%where a finite set of control variables are used for transition guards
exactly one data variable that is used to store the weight, with the following
differences:
$(i)$~the input data values in WRA can be compared using an arbitrary collection of
binary data relations in the data domain, and
$(ii)$~the data variable can be updated using an arbitrary collection of binary
data functions from the data domain to the weight domain.
The work~\cite{Babari16} focused on the expressibility issues and did not
investigate the decision problems.

Finally, let us mention \emph{symbolic automata} and \emph{symbolic
transducers}~\cite{Veanes12,Veanes13,DAntoni15}. They are extensions of finite
automata and transducers where guards in transitions are predicates from an
alphabet theory (which is a~parameter of the model), thus preserving many of
their nice properties.
Extending these models with registers in a~straightforward way yields
undecidable models.
Imposing a~register access policy (such as that a~register always holds the
previous value) can bring some decision problems back to the realm of
decidability~\cite{DAntoni15,Czyba15}.
It is an interesting open problem to find a way of combining symbolic automata with
$\raq$.

%%%%%%%%%%%%%%%%%%%%%%%%%%%%%%%%%%%%%%%

\section{Concluding Remarks}\label{sec:conclusion}

In this paper, we defined $\raq$ over the rationals.
To the best of our knowledge, it is the first such model over infinite alphabets 
that allows arithmetic on the input word, while keeping some interesting decision problems decidable. 
We study some natural decision problems such as
the invariant/non-zero problem, which is a generalization of
the standard non-emptiness problem,
as well as the equivalence, commutativity, and reachability problems.
$\raq$ is also quite a general model subsuming
at least three well-known models, i.e.,
the standard RA, affine programs, and arithmetic circuits.

It will be interesting to investigate 
the configuration reachability and coverability problems for copyless $\raq$.
Both of them subsume the corresponding problems for $\bbZ$- and $\bbQ$-VASS,
since such VASS can be viewed as $\raq$ where data variables represent the counters in the VASS.
From Theorem~\ref{thm-reach-und},
we can already deduce that they are undecidable for general $\raq$.
We leave the corresponding problems for copyless $\raq$ as future work.

% %%%%%%%%%%%%%%%%%%%%%%%%%%%%%%%%%%%%%%%%%%%%%%%%%%%%%%%%%%%%%%%%%%%%%%%%%%%%%%%%

\subsubsection*{Acknowledgement}

% %%%%%%%%%%%%%%%%%%%%%%%%%%%%%%%%%%%%%%%%%%%%%%%%%%%%%%%%%%%%%%%%%%%%%%%%%%%%%%%%
We thank Rajeev Alur for valuable discussions and
the anonymous reviewers for their helpful suggestions about how to
improve the presentation of the paper.
Yu-Fang Chen is supported by the MOST grant No.\ 103-2221-E-001-019-MY3.
Ond\v {r}ej Leng\'{a}l is supported by
the Czech Science Foundation (project 17-12465S),
the BUT FIT project FIT-S-17-4014,
and the IT4IXS: IT4Innovations Excellence in Science project (LQ1602).
Tony Tan is supported by the MOST grant No.\ 105-2221-E-002-145-MY2.
Zhilin Wu is supported by the NSFC grants No.\ 61472474, 61572478, 61100062, and 61272135.


%\IEEEtriggeratref{19}



%\bibliographystyle{IEEEtran}

\bibliographystyle{IEEEtran}

\newpage
\onecolumn
\appendix

\input{app-tony.tex}

\input{app-yufang.tex}
\input{app-ondrej.tex}
\input{app-zhilin.tex}

\end{document}

%% file: app-tony.tex
%!TEX root = snt-main.tex

\section{Missing proofs}

%\captionsetup[subfigure]{subrefformat=simple,labelformat=simple,listofformat=subsimple}
%\renewcommand\thesubfigure{(\alph{subfigure})}

\subsection{Proof of Lemma~\ref{lem:affine-space}}
\label{app:proof:lem:affine-space}

Let $\bbA=\vu+\bbV$, where $\bbV$ is a vector space.
Let $T\vx = M\vx +\va$.
Furthermore, let $M = [M_0 \mid \vb]$, i.e., $\vb$ is the last column of matrix $M$.
Then, $T \myvec{\vv\\ d_i} = M_0\vv  +  d_i \vb + \va \ \in\ \bbA$,
and hence, $M_0\vv  +  d_i \vb + \va-\vu \in \bbV$, for every $i=1,2$.
Subtracting one from the other, we get $(d_1-d_2)\vb \in \bbV$.

Since $\alpha(d_1-d_2)\vb \in \bbV$, for every $\alpha \in \bbQ$,
and $T \myvec{\vv\\ d_i} \in \bbA$, 
\begin{eqnarray*}
T \myvec{\vv\\ d_i} + \alpha(d_1-d_2)\vb & \in & \bbA,
\qquad\qquad\mbox{for every}\ \alpha \in \bbQ.
\end{eqnarray*}
For every $d \in \bbQ$, if we take $\alpha = (d-d_1)/(d_1-d_2)$,
we have $T \myvec{\vv\\ d}  = 
T \myvec{\vv\\ d_1} + \alpha(d_1-d_2)\vb$, which by the equation above, is in $\bbA$.
This completes our proof of Lemma~\ref{lem:affine-space}.

\subsection{Proof of Lemma~\ref{lem:finite-dim}}
\label{app:proof:lem:finite-dim}

It suffices to show that for $m\geq \dim(\bbH)+2$,
there is such a set $J$ such that $|J|\leq m-1$.
The proof is by induction on $\dim(\bbH)$.
The base case is $\dim(\bbH)=0$, i.e., $\bbH$ consists of only one point, say $\vc$.
We pick the smallest index $j\in \{1,\ldots,m+1\}$ such that $\vu_j\neq \vc$.
If $j=1$, we set $J=\emptyset$, and 
the claim holds trivially.
If $j\neq 1$, then $\vu_1=\dots = \vu_{j-1}=\vc$, we can set $J=\{j-1\}$,
where $T_{j-1} \vu_1 = T_{j-1} \vc \notin \bbH$.

For the induction step, let
$m\geq \dim(\bbH)+2$ and $\vu_1,\ldots,\vu_{m+1}$ be vectors such that $\vu_{i+1}= T_i \vu_i$.
If $\vu_{m} \notin \bbH$,
we can take $J=\{1,\ldots,m-1\}$.
So, we assume that $\vu_{m}\in \bbH$.

Let $\bbK$ be the pre-image of $\bbH$ under $T_m$,
i.e., $\bbK=\{\vu \mid T_m \vu \in \bbH\}$.
Since $T_m(\vu_{m}) \notin \bbH$,
we have $\vu_{m} \ \notin \ \bbK$.
Thus, $\bbH\cap\bbK \subsetneq \bbH$, and therefore,
$\dim(\bbH\cap \bbK)\leq \dim(\bbH)-1$.
Applying the induction hypothesis,
we have $J_0=\{j_1,\ldots,j_n\}$ such that $|J_0|\leq m-2$ and $T_{j_n}\cdots T_{j_1} \vu_1 \notin \bbH\cap \bbK$.

Let $\vz = T_{j_n}\cdots T_{j_1} \vu_1$.
If $\vz \notin \bbH$, the set $J=J_0$ is as desired.
Assume that $\vz \in \bbH$.
Since $\vz \notin \bbH\cap \bbK$, 
it should be that $T_m \vz \notin \bbH$.
So, the set $J= J_0 \cup \{m\}$ is as desired.
This completes our proof of Lemma~\ref{lem:finite-dim}.

\subsection{Proof for Remark~\ref{rem:karr-small-model}}
\label{app:proof:rem:karr-small-model}

Let $\cP=(S,s_0,\mu)$ be an AP.
Let the path from $(s_0,\vu)$ to $(s',\vv)$, for some $\vv\notin \bbH$, as follows.
$$
(s_0,T_1,s_1),(s_1,T_2,s_2),\ldots,(s_{\ell-1},T_{\ell},s_{\ell})
$$
If $\ell > (n+1)|S|$,
there is a state that appears at least $n+2$ times in the path.
Let us denote by $p$ such state.

Let $T_0'$ be the composition of the transformations
from $s_0$ to the first appearance of $p$,
and $T_1'$ the composition of the transformations
from the first appearance of $p$ to the second, and 
$T_2'$ the composition of the transformations from the second appearance of $p$
to the third, and so on.
Let $m\geq n+2$ be the number of appearances of $p$ in the path.
So, we have:
$$
R T_{m-1}' \cdots T_1'T_0'\vu\ =\ \vv\ \notin\ \bbH
$$
where $R$ is the composition of the transformations from the last appearance of $p$ to $s'$.

Let $\bbK$ be the pre-image of $\bbH$ under $R$.
So, 
$$
T_{m-1}' \cdots T_{1}'T_0'\vu\ \notin\ \bbK
$$
Now, $m-1 \geq n+1$.
On the other hand, $\dim(\bbK)\leq n-1$,
as otherwise, $\bbK$ is the whole space $\bbQ^n$.
So, $m-1 \geq \dim(\bbK)+2$.
By Lemma~\ref{lem:finite-dim}, there is $J=\{j_1,\ldots,j_n\}\subseteq\{1,\ldots,m-1\}$ 
such that $|J|\leq \dim(\bbK)+1$ and
$$
T_{j_n}' \cdots T_{j_1}'T_0'\vu\ \notin\ \bbK
$$
and thus, $RT_{j_n}' \cdots T_{j_1}'T_0'\vu\ \notin\ \bbH$. 
The path for such composition of transformations contains the state $p$ only
for at most $\dim(\bbK)+2$ number of times.
So, we have the bound that the path contains the state $p$ at most $n+1$ times.

Note that if all the transformations in $\cP$ are one-to-one,
$R$ is also one-to-one, and thus, $\dim(\bbK)=\dim(\bbH)$.
So, the path contains $p$ at most $\dim(\bbH)+2$ times.

%The case when not all the transformations in $\mu$ are 1-1 can be reduced to the case above.
%For each $T_i$, let $M_i$ and $\va_i$ be the matrix and vector that
%define $T_i$, i.e., $T_i \vx = M_i \vx +\va$.
%Let $\bbK_0$ be the kernel of the composition of the product of the matrices $\bbK_0=\ker(M_{\ell}\cdots M_1)$.
%Now, the product $M_{\ell}\cdots M_1$ can be viewed as
%1-1 linear mapping on the quotient space $\bbQ^n/\bbK_0$,
%which means each of them is 1-1 mapping.
%This implies that each $T_i$ is also 1-1 mapping on $\bbQ^n/\bbK_0$.
%Hence,
%$$
%T_\ell \cdots T_1 \vu \ \notin \ \bbH/\bbK_0.
%$$
%Since $\dim(\bbH/\bbK_0) \leq \dim(\bbH)$, the reduction and the remark follow.

\subsection{A brief review of arithmetic circuits and proof of Theorem~\ref{theo:circuit}}
\label{app:proof:theo:circuit}

Briefly, a (division-free) arithmetic circuit (AC) is  a directed acyclic graph with nodes
labeled with constants from $\bbZ$, or with some indeterminates $X_1,\ldots,X_m$,
or with the operators $+,-,\times$.
The nodes labeled with constants are called constant nodes,
while those labeled with indeterminates are called input nodes.
Both constant and input nodes don't have incoming edges.
Internal nodes are those labeled with $+,-,\times$.
Output node is one without out-going edges.

We assume that all the operators $+,-,\times$ are binary,
so all the internal nodes have in-degree $2$.
We also assume that there is only one output node.
Each node $u$ in a circuit represents a multivariate polynomial $\bbZ[X_1,\ldots,X_m]$, denoted by $\val_u$.
Here we are only interested in arithmetic circuits without input nodes,
thus, $\val_u$ is an integer.

In the following, $\oplus$-nodes and $\otimes$-nodes
refer to nodes with label $+$ and $\times$, respectively.
Since $-$ can be rewritten with $+$ and multiplication with $-1$,
we assume our AC contains only internal $\oplus$- and $\otimes$-nodes.
A circuit is called {\em multiplicative}, if all the internal nodes are $\otimes$-nodes.
Likewise, it is additive, if all the internal nodes are $\oplus$-nodes.

%We consider only arithmetic circuit without indeterminates and whose internal nodes are all labeled with $\times$.
%Briefly, such circuit is  a directed acyclic graph with nodes
%labeled with constants from $\bbZ$ or the operators $\times$.
%The nodes labeled with constants are called constant nodes,
%and they don't have incoming edges.
%Internal nodes are those labeled with $\times$.
%Output node is one without out-going edges.
%
%We assume that $\times$ is a binary operator,
%so all the internal nodes have in-degree $2$.
%We also assume that there is only one output node.
%Note that each node $u$ in AC represents an integer, denoted by $\val_u$.

%An alternative, but equivalent, model to AC is a {\em straight-line program} (SLP),
%which is a sequence of instructions corresponding
%to a sequential evaluation of an arithmetic circuit:
%\begin{equation}
%\label{eq:slp}
%\qquad z_1 := e_1;\qquad \cdots \qquad z_n:= e_n;
%\end{equation}
%where the variables $z_1,\ldots,z_n$ are all different.
%Each $e_i$ is either a constant or an arithmetic expression of the form $z\circledast z'$,
%where $z,z'$ are variables that already appears previously and $\circledast$
%is one of $+$, $-$, $\times$.
%The value stored in $z_n$ is taken to be the integer
%represented by the SLP.

The rest of this section is devoted to the proof of Theorem~\ref{theo:circuit}.
We will first describe the proof for multiplicative/additive circuits.
The general case will follow after that.

We need some terminologies.
For an edge $e=(u,v)$, we say that $u$ is the {\em source node} of $e$,
while $v$ is the {\em target node} of $e$.
A node {\em $u$ is below $v$} in a circuit $C$, 
if either $u=v$ or there is a path from $u$ to $v$ in $C$.
An edge $e$ is below $v$, if the target node of $e$ is below $v$.
Intuitively, if $u$ is below $v$, the value of $v$ depends on the value of $u$.
Likewise, if an edge $e$ is below $v$, the value of $v$ depends
on the value of both the source and the target nodes of $v$.

To avoid cumbersome case-by-case analysis, we assume that
there is an ``imaginary'' loop on each constant node in the AC,
as it will help reduce the complication in our proof.
With such loops, all nodes in AC have incoming edges.

\paragraph*{Multiplicative/additive circuits}
We will only describe our proof for multiplicative circuits.
The additive case can be treated in a similar manner.

Before we present the formal detail,
we would like to explain our proof via the example in Figure~\ref{fig:ac-raq}.
Figure~\ref{fig:ac-raq1} shows an AC, where $\otimes,u_i$ indicate the node $u_i$ with label $\times$.
Figure~\ref{fig:ac-raq2} shows the topology of the constructed $\raq$
with initial state $q_1$ corresponding to the output node $u_1$, 
final state $q_f$ and four states $q_2,q_3$ and $q_{c_1},q_{c_2}$
corresponding to the internal nodes $u_2,u_3$ and the constant nodes $c_1,c_2$ in AC, respectively.

\begin{figure}
\centering

\subfigure[]{\label{fig:ac-raq1}
%\subfloat[]{\label{fig:ac-raq1}
\begin{tikzpicture}[scale=0.8,every node/.style={minimum size=1pt}]

    \node (u1) at (0,5) {$\otimes$\scriptsize{,$u_1$}};
    \node (u2) at (3,2.5) {$\otimes$\scriptsize{,$u_2$}};
    \node (u3) at (0,0) {$\otimes$\scriptsize{,$u_3$}};

    \node (c1) at (0,-2) {$c_1$\scriptsize{,$u_1'$}};
    \node (c2) at (3,-2) {$c_2$\scriptsize{,$u_2'$}};

	\node (x) at (0,-3) {};

    \draw[->] (u3) to (u2);
    \draw[->] (u3) to (u1);
    \draw[->] (u2) to (u1);
    \draw[->] (c1) to (u3);
    \draw[->] (c2) to (u2);
    \draw[->] (c2) to (u3);

	\draw[->,loop below] (c1) to (c1);
	\draw[->,loop below] (c2) to (c2);
\end{tikzpicture}
}\hspace{3cm}
\subfigure[]{\label{fig:ac-raq2}
%\subfloat[]{\label{fig:ac-raq2}
\begin{tikzpicture}[scale=0.7,every node/.style={minimum size=1pt}]

    \node (q1) at (0,6) {$q_1$};
    \node (q2) at (4,3) {$q_2$};
    \node (q3) at (0,0) {$q_3$};

    \node (c1) at (0,-3) {$q_{c_1}$};
    \node (c2) at (4,-3) {$q_{c_2}$};

	\node (qf) at (4,6) {$q_f$};

\tikzset{->-/.style={
        decoration={markings,
            mark= at position 0.5 with {\arrow{latex}} ,
        },
        postaction={decorate}
    }
}
%    \draw[->-] (q0) to node[above] {\scriptsize$t_0$} (q1);
    \draw[->-] (q1) to node[above] {\scriptsize$t_f$} (qf);

    \draw[->-, bend right=10] (q1) to node[left] {\scriptsize$t_1$} (q3);
    \draw[->-, bend right=10] (q3) to node[right] {\scriptsize$t_1^r$} (q1);

    \draw[->-, bend right=10] (q1) to node[left] {\scriptsize$t_2$} (q2);
    \draw[->-, bend right=10] (q2) to node[right] {\scriptsize$t_2^r$} (q1);

    \draw[->-, bend right=10] (q2) to node[left] {\scriptsize$t_3$} (q3);
    \draw[->-, bend right=10] (q3) to node[right] {\scriptsize$t_3^r$} (q2);

    \draw[->-, bend right=10] (q3) to node[left] {\scriptsize$t_4$} (c1);
    \draw[->-, bend right=10] (c1) to node[right] {\scriptsize$t_4^r$} (q3);

    \draw[->-, bend right=10] (q3) to node[left] {\scriptsize$t_5$} (c2);
    \draw[->-, bend right=10] (c2) to node[right] {\scriptsize$t_5^r$} (q3);

    \draw[->-, bend right=10] (q2) to node[left] {\scriptsize$t_6$} (c2);
    \draw[->-, bend right=10] (c2) to node[right] {\scriptsize$t_6^r$} (q2);

\tikzset{->-/.style={
        decoration={markings,
            mark= at position 0.99 with {\arrow{latex}} ,
        },
        postaction={decorate}
    }
}

	\tikzset{loop/.style={min distance=0mm,in=330,out=30,looseness=10}}
	\draw[->-,loop below] (c2) to node[right] {\scriptsize$t_{c_2}$} (c2);

	\tikzset{loop/.style={min distance=0mm,in=210,out=150,looseness=10}}
	\draw[->-,loop below] (c1) to node[left] {\scriptsize$t_{c_1}$} (c1);
%
%	\tikzset{loop/.style={min distance=1mm,in=220,out=150,looseness=10}}
%	\draw[->,loop below] (q3) to node[left] {\scriptsize$t_{c_1}$} (q3);

\end{tikzpicture}
}
\caption{\label{fig:ac-raq}%
An AC on the left and the constructed $\raq$ on the right.
The loops on the constant nodes in AC are imaginary loops whose sole purpose
is to make the nodes ``uniform'' in that all of them have incoming edges,
thus, avoid cumbersome case analysis in our construction.}
\end{figure}
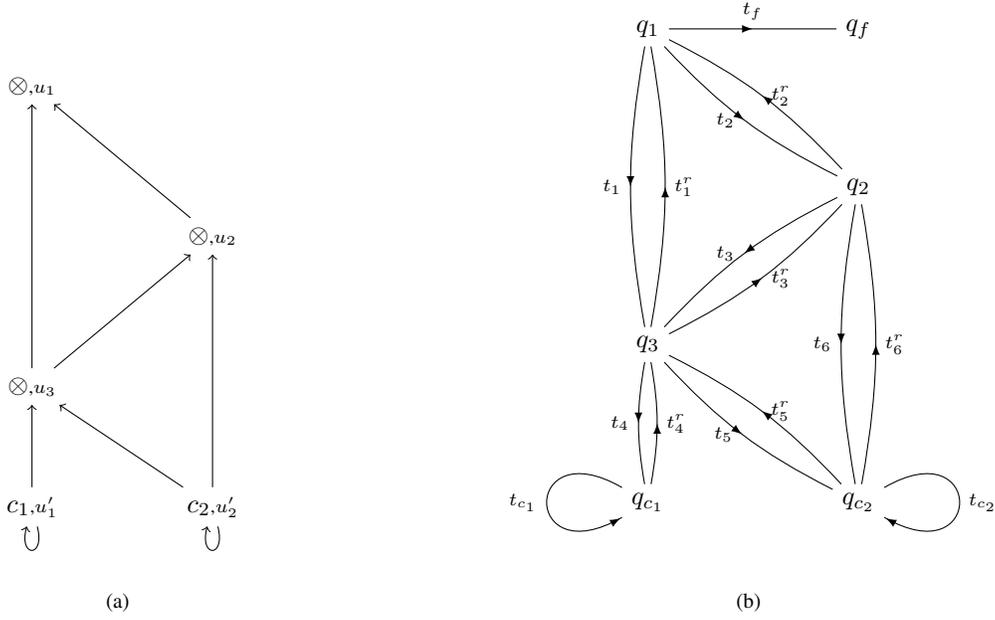

Intuitively, the $\raq$ has one data variable $y$ with initial content $1$.
The output function in the final state $q_f$ will output the value in data variable $y$.
To reach $q_f$, the $\raq$ will be ``forced'' to traverse along the following path:
$$
q_1,\ q_3,\ q_{c_1},\ q_3,\ q_{c_2},\ q_3,\
q_1,\ q_2,\ q_3,\ q_{c_1},\ q_3,\ q_{c_2},\ q_3,\
q_2,\ q_{c_2},\ q_2,\ q_1,\ q_f 
$$
When it reaches $q_{c_1}$, it will take transition $t_{c_1}$
in which there is a reassignment $y:= c_1\cdot y$, before
it can leave $q_{c_1}$.
Likewise, when it reaches $q_{c_2}$,
it will take transition $t_{c_2}$ with reassignment $y:=c_2\cdot y$.
In the transitions $t_0,\ldots,t_6,t_f$ the reassignment is $y:= y$ (i.e. the value of $y$ is unchanged).
Note that in the path above, the $\raq$ enters $q_3$ once from $q_1$
and once from $q_2$.
Each time it enters $q_3$, it has to enter both $q_{c_1}$ and $q_{c_2}$ to ensure $y$ is multiplied with $c_1$ and $c_2$.

\hide{\yfc{The original text}
To make sure that the $\raq$ traverses through all its transitions correctly,
we will employ control variable $x_e$
for each edge $e$,
that stores 0-1 value to remember which edge yet to traverse
and which edge to backtrack to.
}

To make sure that the $\raq$ traverses through all its transitions correctly,
we will employ control variables $x_e$ and $z_e$
for the forward and reverse direction of each edge $e$, respectively,
that stores 0-1 value to remember which edge yet to traverse
and which edge to backtrack to. The $0$ value of a variable $x_e$ or $z_e$ means the corresponding transition can be taken and $1$ otherwise.
The formal construction is as follows.
Let $C=(V,E)$ be an AC where all the internal nodes are $\otimes$-nodes,
and $E$ is the set of edges which includes the loops on the constant nodes.

Let $\cA = \langle Q,q_0,F,\vu_0,\delta,\zeta\rangle$ be the following $\raq$.
\begin{itemize}\itemsep=0pt
\item
There is one data variable $y$ with initial content $1$.

\item
For each edge $e$ in $E$, there are two control variables $x_e$ and $z_e$
with initial content $0$.

\item
The set $Q$ of states is $\{q_f\}\cup \{q_u \mid u \in V\}$,
where $q_u$ is a state corresponding to the node $u$ in $V$.

\item
The initial state is $q_u$, where $u$ is the output node in $C$,
and there is only final state $q_f$.
\item
The output function in $q_f$ outputs the value $y$.
\end{itemize}
The transitions in $\delta$ are defined as follows.

In the following we consider a loop $e$ on a node $u$
as both its incoming and outgoing edge.
For each internal node $u$, we fix one of its incoming edge as its {\em left-edge}
and the other one as its {\em right-edge}.
\begin{enumerate}[{\bf [T1]}]\itemsep=0pt
\item
Let $u$ be the output node in $C$ and $e_1,e_2$ be its left- and right-edges, respectively.
Then, $\delta$ contains the transition:
\begin{eqnarray*}
(q_u, x_{e_1}=1 \wedge x_{e_2}=1) & \to & (q_f,\{\})
\end{eqnarray*}
%The intuitive meaning of the first transition should be clear.
Intuitively, the transition means that when $\cA$ is in state $q_u$
and it has traversed both the edges $e_1$ and $e_2$
(indicated by them containing value $1$),
then it can go to the final state $q_f$.

\item
For each internal node $u$ in $C$, the transitions from the state $q_u$ are 
defined as follows.

Let $e_1,e_2$ be left- and right-edges of $u$, respectively,
and let $v_1,v_2$ be their source nodes.

The following transitions are in $\delta$:
\begin{eqnarray*}
(q_u,x_{e_1}=0) & \to & (q_{v_1},\{z_{e_1}:=1\}\cup\{ x_{e}:=0 \mid e \ \textrm{is an incoming edge of}\ v_1\})
\\
(q_u, x_{e_1}=1 \wedge x_{e_2}=0) & \to & (q_{v_2},\{z_{e_2}:=1\}\cup\{ x_{e}:=0 \mid e \ \textrm{is an incoming edge of}\ v_2\})
\end{eqnarray*}
Note that if $v_i$ is a constant node, then there is only one incoming edge $e$ to $v_i$,
which is the loop on $v_i$, and $x_{e}$ is assigned $0$, as well.

Let $e_1',\ldots,e_m'$ be the outgoing edges of $u$,
and let $v_1',\ldots,v_m'$ be their target nodes, respectively.

For each $i=1,2,\ldots,m$, the following transition is in $\delta$:
\begin{eqnarray*}
(q_u, x_{e_1}=1 \wedge x_{e_2}=1 \wedge z_{e_i'}=1)
& \to & 
(q_{v_i'},\{ z_{e_i'}:=0; x_{e_i'}:=1\})
\end{eqnarray*}

\item
For each constant node $u$ in $C$, the transition from state $q_u$ is defined as follows.

Let $p$ be the integer label of $u$.

Let $e$ be the loop on $u$.
The following transition is in $\delta$:
\begin{eqnarray*}
(q_u,x_{e}=0) & \to & (q_{u},\{ x_{e}:=1; y:= p\cdot y\})
\end{eqnarray*}

Let $e_1',\ldots,e_m'$ be the outgoing edges of $u$,
and let $v_1',\ldots,v_m'$ be their target nodes, respectively.

For each $i=1,2,\ldots,m$, the following transition is in $\delta$:
\begin{eqnarray*}
(q_u, x_{e}=1 \wedge z_{e_i'}=1)
& \to & 
(q_{v_i'},\{ z_{e_i'}:=0; x_{e_i'}:=1\})
\end{eqnarray*}

\end{enumerate}

The following claim immediately implies the correctness of our $\raq$.
Intuitively, it states that for each node $u$ in the circuit $C$,
if $\cA$ ``enters'' $q_u$ with the initial content of $y$ being, say $N$,
then $\cA$ ``exits'' $q_u$ with the content of $y$ being $\val_{u}\cdot N$.

\begin{claim}
Let $u$ be a node in the circuit $C$.
Let $\vv_1$ be the contents of the variables such that $\vv_1(x_e)=0$, for every incoming edge $e$ of $u$.
Then, 
$$
(q_u,\vv_1)\vdash^* (q_u,\vv_2)
$$
for some $\vv_2$ such that
\begin{itemize}\itemsep=0pt
\item
$\vv_2(y)=\val_u \times \vv_1(y)$;
\item
$\vv_2(x_e)=
\left\{
\begin{array}{ll}
1, & \textrm{for every edge $e$ below $u$}
\\
 \vv_1(x_e), & \textrm{for other $e$.}
\end{array}
\right.
$
\end{itemize}
\end{claim}

We omit the proof which is by straightforward induction on the distance between $u$
and its furthest constant node.

\paragraph*{Proof of Theorem~\ref{theo:circuit}}
Let $C$ be an AC.
The construction of $\cA$ follows similar idea as above
by evaluating $\val_u$, for each node $u$ in $C$.
When $u$ is $\otimes$-node,
the evaluation is similar to the above.
When $u$ is a $\oplus$-node,
the evaluation is done via the distributive law: $\alpha(\beta+\gamma)=\alpha\beta+\alpha\gamma$.

Let $\cA = \langle Q,q_0,F,\vu_0,\delta,\zeta\rangle$ be the following $\raq$.
\begin{itemize}\itemsep=0pt
\item
$\cA$ contains the following data variables.
\begin{itemize}\itemsep=0pt
\item
For each $\otimes$-node $u$, there is one data variable $y_u$ with initial content $1$.
\item
For each $\oplus$-node $u$, there are three data variables $y_u,y_u^{\myleft}, y_{u}^{\myright}$
with initial content $1$. 
\item
For each constant node $u$, there is one data variable $y_u$
with initial content $1$.
\end{itemize}
In fact, as we will see later, the initial contents of data variables are not important,
except the one associated with the output node.

Intuitively, the meaning of the data variable for $\oplus$-node is as follows.
Let $u$ be a $\oplus$-node and $e$ be an arc  from $u$ to $v$, when $\cA$ goes
from a state $q_v$ to $q_u$ in some transition,  $y_u$ will be assigned with the
value of $y_v$, moreover, the value of $y_u$ will keep unchanged before
returning to the state $q_v$. In addition, let $v_1, v_2$ be the source nodes of
the left- and right-edge of $u$ respectively, then after returning to the state
$q_u$ from the state $q_{v_1}$ for the first time, $y_u^{\myleft} = y_u \cdot
\val_{v_1}$, similarly, after returning to the state $q_u$ from the state
$q_{v_2}$ for the first time, $y_u^{\myright} = y_u \cdot \val_{v_2}$.

\item
For each edge $e$ in $E$, there are two control variables $x_e$ and $z_e$
with initial content $0$.

\item
The set $Q$ of states is $\{q_f\}\cup \{q_u \mid u \in V\}$,
where $q_u$ is a state corresponds to the node $u$ in $V$.

\item
The initial state is $q_u$, where $u$ is the output node in $C$,
and there is only final state $q_f$.
\item
If the output node $u$ is $\otimes$-node, the function $\zeta(q_f)$ outputs $y_u$.

If the output node $u$ is $\oplus$-node, the function $\zeta(q_f)$ outputs $y_u^{\myleft}+y_u^{\myright}$.
\end{itemize}
The transitions in $\delta$ are defined as follows.
\begin{enumerate}[{\bf [T1]}]\itemsep=0pt
\item
Let $u$ be the output node in $C$ and $e_1,e_2$ be its left- and right-edges, respectively.
Then, $\delta$ contains the transition:
\begin{eqnarray*}
(q_u, x_{e_1}=1 \wedge x_{e_2}=1) & \to & (q_f,\{\})
\end{eqnarray*}
\item
For each internal $\otimes$-node $u$ in $C$, the transitions from state $q_u$ are 
defined as follows.

Let $e_1,e_2$ be left- and right-edges of $u$, respectively,
and let $v_1,v_2$ be their source nodes.

The following transitions are in $\delta$:
\begin{eqnarray*}
(q_u,x_{e_1}=0) & \to & (q_{v_1},\{z_{e_1}:=1; y_{v_1}:=y_u\}\cup\{ x_{e}:=0 \mid e \ \textrm{is an incoming edge of}\ v_1\})
\\
(q_u, x_{e_1}=1 \wedge x_{e_2}=0) & \to & (q_{v_2},\{z_{e_2}:=1; y_{v_2}:=y_u\}\cup\{ x_{e}:=0 \mid e \ \textrm{is an incoming edge of}\ v_2\})
\end{eqnarray*}

Let $e_1',\ldots,e_m'$ be the outgoing edges of $u$,
and let $v_1',\ldots,v_m'$ be their target nodes, respectively.

For each $i=1,2,\ldots,m$, the following transition is in $\delta$:
\begin{itemize}\itemsep=0pt
\item
If $v_i'$ is $\otimes$-node:
\begin{eqnarray*}
(q_u, x_{e_1}=1 \wedge x_{e_2}=1 \wedge z_{e_i'}=1)
& \to & 
(q_{v_i'},\{ z_{e_i'}:=0; x_{e_i'}:=1; y_{v_i'}:= y_u\})
\end{eqnarray*}
\item
If $v_i'$ is $\oplus$-node and $e_i'$ is its left-edge:
\begin{eqnarray*}
(q_u, x_{e_1}=1 \wedge x_{e_2}=1 \wedge z_{e_i'}=1)
& \to & 
(q_{v_i'},\{ z_{e_i'}:=0; x_{e_i'}:=1; y_{v_i'}^{\myleft}:= y_u\})
\end{eqnarray*}
\item
If $v_i'$ is $\oplus$-node and $e_i'$ is its right-edge:
\begin{eqnarray*}
(q_u, x_{e_1}=1 \wedge x_{e_2}=1 \wedge z_{e_i'}=1)
& \to & 
(q_{v_i'},\{ z_{e_i'}:=0; x_{e_i'}:=1; y_{v_i'}^{\myright}:= y_u\})
\end{eqnarray*}
\end{itemize}

\item
For each internal $\oplus$-node $u$ in $C$, the transitions from state $q_u$ are 
defined as follows.

Let $e_1,e_2$ be left- and right-edges of $u$, respectively,
and let $v_1,v_2$ be their source nodes.

The following transitions are in $\delta$:
\begin{eqnarray*}
(q_u,x_{e_1}=0) & \to & (q_{v_1},\{z_{e_1}:=1; y_{v_1}:=y_u\}\cup\{ x_{e}:=0 \mid e \ \textrm{is an incoming edge of}\ v_1\})
\\
(q_u, x_{e_1}=1 \wedge x_{e_2}=0) & \to & (q_{v_2},\{z_{e_2}:=1; y_{v_2}:=y_u\}\cup\{ x_{e}:=0 \mid e \ \textrm{is an incoming edge of}\ v_2\})
\end{eqnarray*}

Let $e_1',\ldots,e_m'$ be the outgoing edges of $u$,
and let $v_1',\ldots,v_m'$ be their target nodes, respectively.

For each $i=1,2,\ldots,m$, the following transition is in $\delta$:
\begin{itemize}\itemsep=0pt
\item
If $v_i'$ is $\otimes$-node:
\begin{eqnarray*}
(q_u, x_{e_1}=1 \wedge x_{e_2}=1 \wedge z_{e_i'}=1)
& \to & 
(q_{v_i'},\{ z_{e_i'}:=0; x_{e_i'}:=1; y_{v_i'}:= y_u^{\myleft}+y_{u}^{\myright}\})
\end{eqnarray*}
\item
If $v_i'$ is $\oplus$-node and $e_i'$ is its left-edge:
\begin{eqnarray*}
(q_u, x_{e_1}=1 \wedge x_{e_2}=1 \wedge z_{e_i'}=1)
& \to & 
(q_{v_i'},\{ z_{e_i'}:=0; x_{e_i'}:=1; y_{v_i'}^{\myleft}:= y_u^{\myleft}+y_{u}^{\myright}\})
\end{eqnarray*}
\item
If $v_i'$ is $\oplus$-node and $e_i'$ is its right-edge:
\begin{eqnarray*}
(q_u, x_{e_1}=1 \wedge x_{e_2}=1 \wedge z_{e_i'}=1)
& \to & 
(q_{v_i'},\{ z_{e_i'}:=0; x_{e_i'}:=1; y_{v_i'}^{\myright}:= y_u^{\myleft}+y_{u}^{\myright}\})
\end{eqnarray*}
\end{itemize}

\item
For each constant node $u$ in $C$, the transition from state $q_u$ is defined as follows.

Let $p$ be the integer label of $u$ and $e$ be the loop on $u$.
The following transition is in $\delta$:
\begin{eqnarray*}
(q_u,x_{e}=0) & \to & (q_{u},\{ x_{e}:=1; y_u:= p\cdot y_u\})
\end{eqnarray*}
Let $e_1',\ldots,e_m'$ be the outgoing edges of $u$,
and let $v_1',\ldots,v_m'$ be their target nodes, respectively.

For each $i=1,2,\ldots,m$, the following transition is in $\delta$:
\begin{itemize}\itemsep=0pt
\item
If $v_i'$ is $\otimes$-node:
\begin{eqnarray*}
(q_u, x_{e_1}=1 \wedge x_{e_2}=1 \wedge z_{e_i'}=1)
& \to & 
(q_{v_i'},\{ z_{e_i'}:=0; x_{e_i'}:=1; y_{v_i'}:= y_u\})
\end{eqnarray*}
\item
If $v_i'$ is $\oplus$-node and $e_i'$ is its left-edge:
\begin{eqnarray*}
(q_u, x_{e_1}=1 \wedge x_{e_2}=1 \wedge z_{e_i'}=1)
& \to & 
(q_{v_i'},\{ z_{e_i'}:=0; x_{e_i'}:=1; y_{v_i'}^{\myleft}:= y_u\})
\end{eqnarray*}
\item
If $v_i'$ is $\oplus$-node and $e_i'$ is its right-edge:
\begin{eqnarray*}
(q_u, x_{e_1}=1 \wedge x_{e_2}=1 \wedge z_{e_i'}=1)
& \to & 
(q_{v_i'},\{ z_{e_i'}:=0; x_{e_i'}:=1; y_{v_i'}^{\myright}:= y_u\})
\end{eqnarray*}
\end{itemize}

\end{enumerate}

The following claim immediately implies the correctness of our construction.
The proof is by straightforward induction on the distance between $u$
and its furthest constant node.

\begin{claim}
Let $u$ be a node in the circuit $C$.
Let $\vv_1$ be the content of the variables of $\cA$
such that $\vv_1(x_e)=0$, for every incoming edge $e$ of $u$.
Then, 
$$
(q_u,\vv_1)\vdash^* (q_u,\vv_2)
$$
for some $\vv_2$ such that the following holds.
\begin{itemize}\itemsep=0pt
\item
$u$ is $\otimes$-node:
\begin{itemize}\itemsep=0pt
\item
$\vv_2(y_u)=\val_u \cdot \vv_1(y_u)$;
\item
$\vv_2(x_e)=
\left\{
\begin{array}{ll}
1, & \textrm{for every edge $e$ below $u$}
\\
 \vv_1(x_e), & \textrm{for other $e$.}
\end{array}
\right.
$
\end{itemize}

\item
$u$ is $\oplus$-node:
\begin{itemize}\itemsep=0pt
\item
$\vv_2(y_u^{\myleft})+\vv_2(y_u^{\myright})=\val_u \cdot \vv_1(y)$;
\item
$\vv_2(x_e)=
\left\{
\begin{array}{ll}
1, & \textrm{for every edge $e$ below $u$}
\\
\vv_1(x_e), & \textrm{for other $e$.}
\end{array}
\right.
$
\end{itemize}

\end{itemize}
\end{claim}

\subsection{Proof of Lemma~\ref{lem:reduction-affine}}
\label{app:proof:lem:reduction-affine}

We will need the following terminologies.
We say that $\cur$ is {\em bound} in a constraint/ordering $\phi$,
if $\phi$ implies $\cur=x_i$, for some $x_i$,
in which case, we say that $\cur$ is bound to $x_i$.
Otherwise, we say that $\cur$ is {\em free}.
For example, in $(x_2< \cur= x_1)$ and $(\cur= x_2 < x_1)$,
$\cur$ is bound to $x_1$ and $x_2$, respectively, while in $(\cur < x_2 < x_1)$ and $(x_1 < \cur < x_2)$,
$\cur$ is free.

The proof of Lemma~\ref{lem:reduction-affine} is by induction on $m$.
The base case $m=0$ is trivial.
For the induction step, suppose the following.
\begin{eqnarray*}
(q_1,\vu_{1}) \ \vdash_{t_1,d_1} \
\cdots\cdots \
\vdash_{t_{m},d_m} \ (q_{m+1},\vu_{m+1})
& \mbox{and} & \vu_{m+1}\notin \bbH.
\end{eqnarray*}

Let $t_m$ be $(p,\varphi(\vx,\cur))\to (q,A,B,\vb)$.
It induces a mapping, denoted by the same symbol $t_m$, that maps an affine space $\bbA$ to another as follows.
\begin{itemize}\itemsep=0pt
\item
If $\cur$ is bound in $\varphi$ to $x_i$ in $\vx$,
then 
\begin{eqnarray*}
t(\bbA) =
\left\{
\begin{array}{l|l}
\vv
&
\vv\ssX = A \myvec {\vu\ssX\\ \vu(i)}, \ \vv\ssY = B \myvec {\vu\\ \vu(i)} + \vb\
\mbox{where} \
 \vu \in \bbA

\end{array}
\right\}
\end{eqnarray*}
\item
If $\cur$ is free in $\varphi$,
then
\begin{eqnarray*}
t(\bbA) & = &
\left\{
\begin{array}{l|l}
\vv
&
\vv\ssX = A \myvec {\vu\ssX\\ \alpha}, \ \vv\ssY = B \myvec{\vu \\ \alpha} + \vb \
\mbox{where} \ \vu \in \bbA \ \mbox{and} \ \alpha \in \bbQ
\end{array}
\right\}
\end{eqnarray*}
\end{itemize}
Let $\bbK$ be the pre-image of $\bbH$ under $t_m$.
Thus, $\vu_{m}\notin \bbK$.
For the the induction hypothesis,
we assume that there is $c_1\cdots c_{m-1}$ such that
\begin{eqnarray*}
(q_1,\vv_{1}) \ \vdash_{t_1,c_1} \
\cdots\cdots \
\vdash_{t_{m},c_{m-1}} \ (q_{m},\vv_{m})
& \mbox{and} & \vv_{m}\notin \bbK,
\end{eqnarray*}
where all the conditions (a)--(e) holds for each $i=1,\ldots,m$.
In particular, $(q_m,\vu_{m})$ and $(q_m,\vv_m)$ have the same ordering.

We will show that there is $c_m$ such that $(q_{m},\vv_{m})\vdash_{t_m,c_m} (q_{m+1},\vv_{m+1})$
such that all the conditions (a)--(e) holds.
There are a few cases.
\begin{itemize}\itemsep=0pt
\item
Case~1: $d_m = \vu_{m}\ssX(j)$ for some $1\leq j\leq k$.

We fix $c_m = \vv_{m}\ssX(j)$.
Since $\vu_m$ and $\vv_m$ have the same ordering,
$\varphi(\vv_m,c_m)$ holds too.
Thus, let $\vv_{m+1}$ be such that
$(q_m,\vv_m)\vdash_{t_m,c_m} (q_{m+1},\vv_{m+1})$.
Since they are obtained by taking the same transition from configurations with the same ordering,
$(q_{m+1},\vv_{m+1})$ and $(q_{m+1},\vu_{m+1})$ have the same ordering, as well.
\item
Case~2: $d_m < \vu_{m}\ssX(j)$, where $\vu_{m}\ssX(j)=\min (\vu_{m}\ssX)$, where $\min (\vu_{m}\ssX)$ is the minimum component of $\vu_{m}\ssX$.

Since $(q_m,\vu_m)$ and $(q_m,\vv_m)$ have the same ordering,
$\varphi(\vv_m,c)$ holds for either $c = \vv_{m}\ssX(j)-1$ or $c = \vv_{m}\ssX(j)-2$.

Now, $\vv_m \notin \bbK$.
Hence, $t_m(\{\vv_m\})\nsubseteq \bbH$.
By Lemma~\ref{lem:affine-space},
for at least one of $c = \vv_{m}(j)-1$ or $c = \vv_{m}(j)-2$,
\begin{eqnarray*}
\vv_{m+1}& \notin & \bbH,
\end{eqnarray*}
where
$\vv_{m+1} \ssX= A \myvec{\vu\ssX\\ c}$
and
$\vv_{m+1}\ssY = B\myvec{\vu\ssX\\ \vu\ssY\\ c} + \vb$.
Let $c_{m+1}$ be such $c$.
Hence, $(q_m,\vv_m)\vdash_{t_m,c_m} (q_{m+1},\vv_{m+1})$,
which also implies
$(q_{m+1},\vv_{m+1})$ and $(q_{m+1},\vu_{m+1})$ having the same ordering.

\item
The other two cases can be proved in a similar manner as in Case~2.
\end{itemize}
This completes our proof of Lemma~\ref{lem:reduction-affine}.

\subsection{Construction of the affine program and application of Karr's algorithm for Theorem~\ref{theo:non-zero-exptime}}
\label{app:proof:theo:non-zero-exptime}

Define the affine program $\cP = (S,s_0, \Delta)$ as follows.
The variables in $\cP$ are $X\cup Y$.
The set $S$ consists of states of the form $(q,\phi)$,
where $q\in Q$ and $\phi$ is an ordering of $X\cup \{\cur\}$.
Let $\phi_0$ be the ordering of $\vu_0$ and $s_0=(q_0,\phi_0)$.
For each transition $t=(p,\varphi(\vx,\cur))\to(q,A,B,\vb)$ in $\delta$,
for each ordering $\phi$ of $X\cup\{\cur\}$ that is consistent
with $t$,
we add the following transitions to $\Delta$.
Note that since $A\in\bbP^{k\times (k+1)}$ simply selects $k$ elements
from $X\cup\{\cur\}$,
we can infer a set of orderings that are consistent
with the result of the application of $A$ on $\phi$.
In the following let $\phi'$ be an ordering that is consistent with
the application of $A$ on $\phi$.
\begin{itemize}\itemsep=0pt
\item
If $\cur$ is bound to $x_i$ in $\phi$,
we add $((p,\phi),T,(q,\phi'))$,
where $T$ represents the reassignment:
\begin{eqnarray*}
\vx := A\myvec {\vx\\ x_i}
& &
\vy := B\myvec {\vx \\ \vy \\ x_i}
\end{eqnarray*}
\item
If $x_i < \cur < x_j$ appears in $\phi$,
we add $((p,\phi),T_1,(q,\phi'))$ and $((p,\phi),T_2,(q,\phi'))$,
where $T_1$ represents the reassignment:
\begin{eqnarray*}
\vx := A\myvec {\vx\\ \frac{1}{3}x_i+\frac{2}{3}x_j}
& &
\vy := B\myvec {\vx \\ \vy \\ \frac{1}{3}x_i+\frac{2}{3}x_j} + \vb
\end{eqnarray*}
and $T_2$ represents the reassignment:
\begin{eqnarray*}
\vx := A\myvec {\vx\\ \frac{2}{3}x_i+\frac{1}{3}x_j}
& &
\vy := B\myvec {\vx \\ \vy \\ \frac{2}{3}x_i+\frac{1}{3}x_j} + \vb
\end{eqnarray*}
\item
If $\cur$ is minimal with $\cur < x_i$ appearing in $\phi$,
we add $((p,\phi),T_1,(q,\phi'))$ and $((p,\phi),T_2,(q,\phi'))$,
where $T_1$ represents the reassignment:
\begin{eqnarray*}
\vx := A\myvec {\vx\\ x_i-1}
& &
\vy := B\myvec {\vx \\ \vy \\ x_i-1} + \vb
\end{eqnarray*}
and $T_2$ represents the reassignment:
\begin{eqnarray*}
\vx := A\myvec {\vx\\ x_i-2}
& &
\vy := B\myvec {\vx \\ \vy \\ x_i-2} + \vb
\end{eqnarray*}
\item
If $\cur$ is maximal with $x_i < \cur$ appearing in $\phi$,
we add $((p,\phi),T_1,(q,\phi'))$ and $((p,\phi),T_2,(q,\phi'))$,
where $T_1$ represents the reassignment:
\begin{eqnarray*}
\vx := A\myvec {\vx\\ x_i+1}
& &
\vy := B\myvec {\vx \\ \vy \\ x_i+1} + \vb
\end{eqnarray*}
and $T_2$ represents the reassignment:
\begin{eqnarray*}
\vx := A\myvec {\vx\\ x_i+2}
& &
\vy := B\myvec {\vx \\ \vy \\ x_i+2} + \vb
\end{eqnarray*}
\end{itemize}

We can apply Karr's algorithm starting
with $V_{(q_0,\phi_0)}=\{\vu_0\}$ and $V_{(q,\phi)}=\emptyset$, for the other $(q,\phi)$.
%For each state $(q,\phi)$,
%let $V_{(q,\phi)}$ be the set of vectors
%obtained by applying Karr's algorithm. 
By Lemma~\ref{lem:reduction-affine},
there is $w$ such that $\cA(w)\neq 0$
if and only if
there is a final state $q_f$ and an ordering $\phi$ such that
$\aff(V_{(q_f,\phi)})\nsubseteq \bbH$,
where $\bbH$ is the space of the solutions of $\zeta(q_f)(\vx,\vy)=0$.

Since there are altogether $2^k(k+1)!$ ordering of $X\cup \{\cur\}$,
$\cP$ has $|Q|2^k(k+1)!$ states with $k+l$ variables.
So, overall our algorithm runs in exponential time.

\subsection{Polynomial space algorithm for non-zero problem with 
constant space assumption for rational numbers in $[-1,1]$}
\label{app:proof:constant-space}

Let $\cA=\langle Q,q_0,F,\vu_0,\delta,\zeta\rangle$ 
be the input $\raq$ over $(X,Y)$,
where $X=\{x_1,\ldots,x_k\}$ and $Y=\{y_1,\ldots,y_l\}$.
Without loss of generality, we can normalize $\cA$ so that
it has the following properties.
\begin{itemize}\itemsep=0pt
\item
There is only one final state $q_F$ whose output function is $\zeta(q_F)=y_1$.
It implies that the output of $\cA$ on a word $w$
is simply the content of the variable $y_1$ when it enters the final state $q_F$.

This can be achieved by adding new transitions that reach $q_F$,
where $y_1$ is assigned the original output function,
and all the variables $x_1,\ldots,x_k,y_2,\ldots,y_l$
are assigned zero.

\item
Every transition in $\delta$ is of the form $(p,\varphi(\vx,\cur))\to (q,A,B,0)$,
i.e., $\vb=0$.

This can achieved by adding new control variables to store the non-zero constants in $\vb$.

\item
We add two new control variables containing $1$ and $2$.

In every transition the content of these two new variables
never change, except in the transition entering the final state $q_F$,
where every variable, except $y_1$, is assigned $0$.
\end{itemize}
The algorithm works non-deterministically as follows.
\begin{enumerate}[(1)]\itemsep=0pt
\item
Guess a positive integer $N \leq |Q|(k+l+1)2^k (k+1)!$.
\item
For each integer $i\leq N$,
compute $d_i$ and $t_i$, $(q_i,\vu_i)$ such that
$(q_{i-1},\vu_{i-1})\vdash_{t_i,d_i} (q_i,\vu_i)$
as follows.

Let $\phi= x_{j_1} \circledast_1 \cdots \circledast_{k-1} x_{j_k}$
be the ordering of $\vu_{i-1}\ssX$,
where each $\circledast_{j}$ is $<$ or $=$.

Guess a transition $t_i=(q_{i-1},\varphi(\vx,\cur))\to (q_i,A,B,0)$ 
that is consistent with the ordering $\phi$ of $\vu_{i-1}$.

Guess $d_i$ according to one of the following cases.
\begin{enumerate}[(a)]\itemsep=0pt
\item
If $\cur$ is bound to some $x_j$ in $\varphi(\vx,\cur)$,
then $d_i = \vu_{i-1}(j)$.
\item
If $\cur < x_{j_1}$ is consistent with $\varphi(\vx,\cur)$,
then either $d_i = \vu_{i-1}(j_1)-1$ or $d_i=\vu_{i-1}(j_1)-2$.
\item
If $\cur > x_{j_k}$ is consistent with $\varphi(\vx,\cur)$,
then either $d_i=\vu_{i-1}(j_k)+1$ or $d_i=\vu_{i-1}(j_k)+2$.
\item
If $x_{j}< \cur < x_{j'}$ is consistent with $\varphi(\vx,\cur)$,
then either $d_i=\frac{1}{3}\vu_{i-1}(j) + \frac{2}{3}\vu_{i-1}(j')$ 
or $d_i=\frac{2}{3}\vu_{i-1}(j) + \frac{1}{3}\vu_{i-1}(j')$ for $d_i$.
\end{enumerate}
Let $\vu_i$ be such that $\vu_{i}\ssX = A \myvec {\vu_{i-1}\ssX\\ d_i}$ 
and $\vu_{i}\ssY = B \myvec{\vu_{i-1} \\ d_i}$.

Divide $\vu_i$ by $\max_j |\vu_{i}(j)|$,
if the absolute values of some components in $\vu_i$ are bigger than $1$.

\item
Verify that $q_N \in F$ and $\zeta(q_N)(\vu_N)\neq 0$.
\end{enumerate}

For the complexity analysis, we assume each rational number between $-1$ and $1$ occupies constant number of bits.
The number $N$ occupies polynomial space.
Moreover, on each $i=1,\ldots,N$, 
the Turing machine only needs to remember the last two configurations 
$(q_{i-1},\vu_{i-1})$ and $(q_i,\vu_i)$.
Hence, the algorithm runs in polynomial space.

The proof of the correctness of our algorithm is as follows.
Dividing each $\vu_i$ by $\max_j |\vu_{i}(j)|$
does not effect the correctness of our algorithm, since $A (\alpha\vu) = \alpha A\vu$,
for every non-zero $\alpha$.
By Theorem~\ref{theo:small-model-snt}, if there is $w \in L(\cA)$ such that 
$\cA(w)\neq 0$, we can assume that $|w|\leq |Q|(k+l+1)2^k (k+1)!$.
The correctness immediately follows from Lemma~\ref{lem:reduction-affine}.

\subsection{Proof of Theorem~\ref{theo:small-model-snt} and the tightness of the bound}
\label{app:proof:theo:small-model-snt}

\paragraph*{Proof of Theorem~\ref{theo:small-model-snt}}
Note that there can be only $2^k(k+1)!$ different orderings of the content of the control variables.
Let $w = d_1\cdots d_N$, where $N > |Q|(k+l+1)2^k (k+1)!$
with its accepting run as follows:
\begin{eqnarray*}
(q_0,\vu_{0}) \ \vdash_{t_1,d_1} \
\cdots\cdots \
\vdash_{t_{N},d_N} \ (q_N,\vu_{N})
& \mbox{and} & q_N\in F,
\end{eqnarray*}
where $\zeta(q_N)(\vu_N)\neq 0$.
Since $N > |Q|(k+l+1)2^k (k+1)!$,
there are at least $k+l+2$ configurations in the run with the same ordering, say $\phi$.
Let $(p,\vv_1),\ldots,(p,\vv_{m+1})$ be such configurations with the ordering $\phi$, where $m\geq k+l+1$.
Let us also denote by $T_0,\ldots,T_{m+1}$ the sequences of transitions such that:
$$
(q_0,\vu_{0}) \ \vdash_{T_0} \
(p,\vv_1) \ \vdash_{T_1} \
\cdots\cdots \
\vdash_{T_m} \ (p,\vv_{m+1}) \
\vdash_{T_{m+1}} \ (q_N,\vu_{N})
$$
Let $\bbA$ be the affine space $\{\vz \mid \zeta(q_N)(\vz)=0\}$,
and let $\bbH$ be the pre-image of $\bbA$ under $T_{m+1}$. Since $\vu_{N} \not \in \bbA$, we have $\vv_{m+1} \not \in \bbH$. 
From this, we deduce that $\dim(\bbH)\leq k+l-1$, so $m\geq dim(\bbH)+2$.
By Lemma~\ref{lem:finite-dim}
there is a set $J=\{j_1,\ldots,j_n\}$ such that $|J|\leq k+l$
and $\vz_1,\ldots,\vz_{n+1}$ such that $\vz_1 = \vv_1$, 
\begin{eqnarray*}
 (q_0,\vu_{0}) \ \vdash_{T_0} \
(p,\vz_1) \ \vdash_{T_{j_1}} \
\cdots\cdots \
\vdash_{T_{j_n}} \ (p,\vz_{n+1}) ,
\end{eqnarray*}
and 
$\vz_{n+1}\notin \bbH$.
Furthermore, because each of $T_1, \dots, T_m$ transforms each configuration $(p, \vu)$ of the ordering $\phi$ to another configuration $(p, \vv)$ with the same ordering,  $\vv_1,\ldots,\vv_{m+1},\vz_1,\ldots,\vz_{n+1}$ have the same orderings.
%By Proposition~\ref{prop:indistinguishable},
Since $\bbH$ is the pre-image of $\bbA=\{\vz \mid \zeta(q_N)(\vz)=0\}$ under $T_{m+1}$, 
we have:
$$
(q_0,\vu_{0}) \ \vdash_{T_0} \
(p,\vz_1) \ \vdash_{T_{j_1}} \
\cdots\cdots \
\vdash_{T_{j_n}} \ (p,\vz_{n+1}) \
\vdash_{T_{m+1}} \ (q_N,\va) 
$$
for some $\va \notin \bbA$, and therefore, $\zeta(q_N)(\va)\neq 0$.
This completes our proof.

\paragraph*{Tightness of the exponential bound}

The exponential bound in Theorem~\ref{theo:small-model-snt} above is tight.
The idea is almost the same as the $\raq$ that represents the number $p^n$ in Section~\ref{sec:snt}.
Consider the following $\raq$ $\cA=\langle Q,q_0,F,\vu_0,\delta,\zeta\rangle$
over $(X,Y)$, where $X=\{x_1,\ldots,x_k,x_{k+1},x_{k+2}\}$, $Y=\{y_1,\ldots,y_l\}$,
$Q=\{q_0,\ldots,q_n\}$, $F=\{q_n\}$ and the transitions
$t_{i,j}$, where $0\leq i \leq n$ and $0\leq j \leq k$, as illustrated below.

\begin{center}
%\begin{figure}

\resizebox{0.7\linewidth}{!}{
\begin{tikzpicture}[->,auto,node distance=2cm]

%\node[state,initial,initial where=left,inner sep=5pt,minimum size=0pt] (q0)  {\small $q_0$};
\node[state,initial,initial where=left,inner sep=5pt,minimum size=0pt] (q0) {\small $q_0$};
\node[state,inner sep=5pt,minimum size=0pt] (q1) [right=of q0] {\small $q_1$};
\node (p) [right=of q1] {$\cdots\cdots$};
\node[state,inner sep=2pt,minimum size=0pt] (qn1) [right=of p] {\small $q_{n-1}$};
\node[state,accepting,inner sep=5pt,minimum size=0pt] (qn) [right=of qn1] {\small $q_{n}$};

%\path[->] 
%(q0) edge node {\scriptsize $t_{0,0}$} (q1);

\path[->] 
(q0) edge node {\footnotesize $t_{0,0}$} (q1);

\path[->] 
(qn1) edge node {\footnotesize $t_{n-1,0}$} (qn);

\path[->] 
(qn) edge [bend left] node [above] {\footnotesize $t_{n,0}$} (q0);

\path[->,every loop/.append style={looseness=15,loop above}]
(q0) edge [loop] node  {\footnotesize $t_{0,1},\ldots,t_{0,k}$} (q0);

\path[->,every loop/.append style={looseness=15,loop above}]
(q1) edge [loop above] node [bend left] {\footnotesize $t_{1,1},\ldots,t_{1,k}$} (q1);

\path[->,every loop/.append style={looseness=15,loop above}]
(qn1) edge [loop above] node [bend left] {\footnotesize $t_{n-1,1},\ldots,t_{n-1,k}$} (qn1);

\path[->,every loop/.append style={looseness=15,loop above}]
(qn) edge [loop above] node [bend left] {\footnotesize $t_{n,1},\ldots,t_{n,k}$} (qn);

\end{tikzpicture}}
%\caption{The illustration of the $\raq$ whose path that leads to non-zero is exponential.}
%\label{fig:exponential}
%\end{figure}
\end{center}
During the computation, the contents of the variables $x_{k+1}$ and $x_{k+2}$ never change,
and they contain $0$ and $1$, respectively.
\begin{itemize}\itemsep=0pt
\item
The initial content of the control variables $(x_1,\ldots,x_k,x_{k+1},x_{k+2})=(0,\ldots,0,0,1)$, i.e.,
all $x_i$'s are initially $0$, except $x_{k+2}$ which is initially $1$.
\item
The initial content of the data variables $(y_1,y_2,\ldots,y_l)=(1,0,\ldots,0)$, i.e.,
$y_1$ is initially $1$, and the rest are initially $0$. 
\item
The function $\zeta(q_n)= y_l$, i.e., the output function associated with $q_n$ is the function $g(\vx,\vy)= y_l$.

\item
For each $0\leq i \leq n$, the transition $t_{i,0}$ is
$$
(q_i,\varphi_k) \to 
(q_{(i+1) \bmod (n+1)},A_{i,0},B,0),
$$ where 
\begin{itemize}\itemsep=0pt
\item
the guard $\varphi_k$ is $\bigwedge_{h=1}^{k} x_h=1$;
\item
the matrix $A_{i,0}$ represents the assignment $\{x_1:=0;\ldots;x_k:= 0\}$;
\item
the matrix $B$ represents the ``shift right'' of the content of the data variables,
i.e., $\{y_1:= y_l;y_2:= y_1;\ldots;y_l:=y_{l-1}\}$.
\end{itemize}
Here the equality and assignment of the variables with $0$ and $1$
can be done via the variables $x_{k+1}$ and $x_{k+2}$,
whose contents are always $0$ and $1$, respectively.

\item
For each $0\leq i \leq n$ and $1\leq j \leq k$, the transition $t_{i,j}$ is
\begin{eqnarray*}
(q_i,\varphi_j) & \to &
(q_{(i+1) \bmod (n+1)},A_{i,j},B,0),
\end{eqnarray*}
where
\begin{itemize}\itemsep=0pt
\item
the guard $\varphi_j$ is $x_j=0 \ \wedge \ \bigwedge_{h=1}^{j-1} x_h=1$;
\item
the matrix $A_{i,j}$ represents the assignment $\{x_1:=0;\ldots,x_{j-1}:= 0;x_{j}:= 1\}$;
\item
the matrix $B$ represents the assignment $\{y_1:= y_1, \ldots, y_l:= y_l\}$,
i.e., it does not change the content of the data variables $y_1,\ldots,y_l$.
\end{itemize}
Again, the equality and assignment of the variables with $0$ and $1$
can be done via the variables $x_{k+1}$ and $x_{k+2}$,
whose contents are always $0$ and $1$, respectively.

\end{itemize}
If $(n+1)$ and $l$ are coprime,
then the length of any path from the initial configuration $(q_0,\vu_0)$ to an accepting configuration
that outputs non-zero is a multiple of $l(n+1)2^k$.

%\tony{I think here we should explain more about the intuition of the $\raq$ and the bound. 
%I'll leave it as it is now, and move on to other ``more important'' stuff. I'll come back to it later.}

\subsection{Proof of Lemma~\ref{lem:forgetful-raq}}
\label{app:proof:lem:forgetful-raq}

The proof is similar to the proof of Lemma~\ref{lem:reduction-affine} in
Appendix~\ref{app:proof:lem:reduction-affine}.

The idea is that if $(q_{j},\vu_{j})\vdash_{t_j,d_j} (q_{j+1},\vu_{j+1})$, where $\vu_{j+1}\notin\bbH$
and $d_j$ does not appear in $\vu_{j}\ssX$,
then there are infinitely many other values $c$ such that
$(q_{j},\vu_{j})\vdash_{t_j,c} (q_{j+1},\vv)$, where $\vv \notin \bbH$.
Indeed, if $d_j$ does not appear in $\vu_{j}\ssX$,
by the density of rational numbers,
there are infinitely many values $c$ for $\cur$ such that
the guard of $t_j$ is satisfied.
By Lemma~\ref{lem:affine-space},
for all, but one, such $c$,
we have $(q_{j},\vu_{j})\vdash_{t_j,c} (q_{j+1},\vv)$, where $\vv \notin \bbH$.

\subsection{The $\pspace$-hardness reductions}
\label{app:hardness-det-ra}

In this appendix, we will establish all the $\pspace$-hardness mentioned in this paper,
i.e., for the following problems.
\begin{itemize}\itemsep=0pt
\item
The non-zero, equivalence, and commutativity problems for copyless $\raq$.
\item
The equivalence and commutativity problems for deterministic RA.
\end{itemize}
It has been shown that the non-emptiness problem
for deterministic standard RA is $\pspace$-complete~\cite{DemriL09}.
All of our reductions either trivially follow their reduction,
or are simply slight modifications of theirs.

For clarity and completeness, we will repeat their reduction here,
and then present our reductions.
It is from the following $\pspace$-complete problem.
Let $c>0$ be a constant.
{\em On input $\lfloor \cM\rfloor$ (a Turing machine description)
and a word $w \in \{0,1\}^*$,
decide whether $\cM$ accepts $w$ using at most $c|w|$ space.}

\paragraph*{Reduction~1: To the non-emptiness of deterministic standard RA~\cite{DemriL09}}
Construct an RA $\cA$ with $c|w|+3$ registers.
The first three registers contain $0$, $1$ and $\sqcup$, respectively,
where $\sqcup$ represents the blank symbol in Turing machine.
The next $|w|$ registers contain the bits of $w$.
Define the states of $\cA$ as pairs $(p,j)$,
where $p$ is a state of $\cM$ and $1\leq j \leq c|w|$ indicating
the position of the head of $\cM$.
On state $(p,j)$, $\cA$ can use transitions
to check whether the content of register $j+3$ is one of $0$, $1$ or $\sqcup$
by comparing its content with the first three registers
and simulate the transitions of $\cM$.
The final states are those $(p,j)$, where $p$ is the accepting state of $\cM$.
Thus, $\cM$ accepts $w$ using at most $c|w|$ space
if and only if $\cA$ can reach one of its final states.

\paragraph*{Reduction~2: To the equivalence problem for deterministic standard RA}
Note that an RA accepts some word
iff it is not equivalent to a trivial RA that accepts nothing.
Since $\pspace$ is closed under complement, the hardness follows.

\paragraph*{Reduction~3: To the commutativity problem for deterministic standard RA}
Indeed, we can modify the constructed RA $\cA$ in reduction~1
such that it enforces the first two values are $0$
and the subsequent values are all different from the $0$.
Obviously, if the machine $\cM$ accepts $w$ using at most $c|w|$ space,
then $\cA$ accepts some word, and $\cA$ is not commutative,
since it requires the first two values are $0$ and the rest are non-zero.
On the other hand, if $\cM$ does not accept $w$ using at most $c|w|$ space,
$\cA$ accepts nothing, which implies that $\cA$ is commutative.
This completes the reduction.

\paragraph*{Reduction~4: To the non-zero problem for copyless $\raq$}
Standard RA are copyless $\raq$ that outputs constant 1 in the final states.
So, it follows from reduction~1.

\paragraph*{Reduction~5: To the equivalence problem for copyless $\raq$}
Like in reduction~2, this follows from the fact that $\pspace$ is closed under complement.
An $\raq$ $\cA$ outputs a non-zero value for some input word
iff it is not equivalent to a trivial $\raq$ that always outputs 0. 

\paragraph*{Reduction~6: To the commutativity problem for copyless $\raq$}
Standard RA are copyless $\raq$ that outputs constant 1 in the final states.
So, it follows from reduction~3.

%% file: app-yufang.tex
%!TEX root = snt-main.tex

\subsection{Proof of Theorem~\ref{thm:raq-undecidable}}
\label{app:proof:thm:raq-undecidable}

First, both the commutativity and equivalence problems of non-deterministic RA
can be reduced to the corresponding problems of single-valued $\raq$, because
non-deterministic RA is a special case of single-valued $\raq$. 

Second, the equivalence/universality problem of non-deterministic register
automata has been proved to be undecidable in Theorem 5.1 of~\cite{NSVianu04},
by constructing an RA that accepts the set of words that encode all
non-solutions to a PCP instance.
Thus, if the PCP has a solution, the RA is not commutative,
because the permutations of the solution are simply arbitrary string.
Likewise, if the RA is not commutative, hence, not universal,
the PCP has a solution.
Therefore, the PCP instance has a solution if and only if the constructed RA is not commutative.

\subsection{Proof of Theorem~\ref{thm:detraq}}
\label{app:proof:thm:detraq}

%-------------------------------------------------------------------------------
\vspace{-0.0mm}
\subsection*{From invariant to non-zero}
\vspace{-0.0mm}
%-------------------------------------------------------------------------------
Note that the Karp reductions between the invariant problem and the non-zero problem
from Section~\ref{sec:non-zero} cannot be used,
because they construct nondeterministic~$\raq$.
Therefore, we modify the said reductions into the following Cook reductions
that preserve determinism.

Given an $\raq$ $\cA$, a state $q$, and $\bbH = \va + \bbV$, similarly to the reduction from
Section~\ref{sec:non-zero}, we compute a~basis $\{\vv_1, \ldots, \vv_m\}$ of the
orthogonal complement of $\bbV$ and transform $\cA$ into $m$~new $\raq$ $\cA_1,
\ldots, \cA_m$, all having a~single final state~$q$, such that the output
function of $\cA_i$ on $q$ is $(\myvec{\vx\\ \vy} - \va)\dotprod \vv_i$.
Each of these automata is then tested for the non-zero problem.

%-------------------------------------------------------------------------------
\vspace{-0.0mm}
\subsection*{From non-zero to invariant}
\vspace{-0.0mm}
%-------------------------------------------------------------------------------
For an~$\raq$ $\cA$ with final states $\{q_1, \ldots, q_n\}$ and output
function~$\zeta$, we test for every $q_i$ whether
% $\cA$, $q_i$, and $\bbH$,
$\vAq{\cA}{q_i} \subseteq \bbH_i$,
where $\bbH_i$ is the space of solutions of the system
$\zeta(q_i)(\vx,\vy)=0$.

\subsection*{From equivalence to non-zero}
Let $\cA_1 =\langle Q_i,q_{0,1},F_1,\vu_{0,1},\delta_1,\zeta_1\rangle$ over $(X_1,Y_1)$
and $\cA_2 =\langle Q_2,q_{0,2},F_2,\vu_{0,2},\delta_2,\zeta_2\rangle$ over $(X_2,Y_2)$
be deterministic $\raq$. Here we use a function instead of a vector to represent the initial content of variables to ease presentation.
Assume w.l.o.g. that the sets of variables in $\cA_1$ and $\cA_2$ are disjoint.

First, we make each $\cA_i$ complete.
% That is, for all outgoing transitions from the same state $q$, 
% the disjunction of their guards is a valid formula 
% (i.e., is evaluated to $\ltrue$ for all $\vx,\cur$).
This can be done in a~linear time by adding a new sink state without affecting the semantics of $\cA_i$.
Next,
we construct a~new $\raq$ $\cA$ over $(X_1\cup X_2, Y_1\cup Y_2)$ such that
$\cA_1$ and $\cA_2$ are equivalent iff $\cA(w)=0$ for all $w$.
The construction is the standard product construction.
Note that $\cA$ is deterministic.
\begin{itemize}\itemsep=0pt
\item
The set of states is $Q_1\times Q_2$.
\item
The initial state is $(q_{0,1},q_{0,2})$.
\item
The set of final states is $Q_1\times F_2 \cup  F_1 \times Q_2$.
\item
The initial content of the variables is the function $\{z\mapsto \vu_{0,1}(z)\mid z\in X_1 \cup Y_1\} \cup \{z\mapsto \vu_{0,2}(z)\mid z\in X_2 \cup Y_2\}$.

\item
The set of transitions $\delta$ consists of the following.

For each pair of transitions $t_1,t_2$,
where $t_1=(p_1,g_1) \rightarrow (q_1, M_1) \in \delta_1$ and
$t_2=(p_2,g_2) \rightarrow (q_2, M_2) \in \delta_2$, and $M_1,M_2$ represent the reassignments,
we add a new transition 
$((p_1,p_2),g_1\wedge g_2) \rightarrow ((q_1,q_2), M_1\cup M_2)$ in $\delta$.
\item
The output function on state $(q_1,q_2)\in F_1\times F_2$ is 
$\zeta_1(q_1)-\zeta_2(q_2)$. 
Otherwise, for pairs $(q_1,q_2)$, where one of them is non-final, 
the output is the constant $1$. 
\end{itemize}

\subsection*{From non-zero to commutativity}
Let $\cA$ be an $\raq$. We first apply the following changes to $\cA$:
\begin{itemize}
	\item we make $\cA$ complete (see the previous reduction) and
	\item we change all non-final states in $\cA$ into final states
	that output the constant $0$.
\end{itemize}
Both changes can be done in a~linear time and will not affect the result of the non-zero problem on $\cA$. Now we have $|\cA(w)| =1 $ for all $w$, i.e., the output of $\cA$ is ``defined'' on all inputs $w$. 

The idea is as follows.
We construct an $\raq$ $\cA'$
such that $\cA'$ outputs the following.
\begin{itemize}\itemsep=0pt
	\item
	For a~word $v$ where the first and second values are $1$ and $2$, respectively,
	i.e., $v= 12 w$ for some $w$, we define
	$\cA'(v)=\cA(w)$.
	\item
	For all the other words, $\cA'$ outputs $0$.
\end{itemize}
We construct $\cA$ in a~linear time
by adding two new states that check the first two input values and
a new final state that outputs the constant $0$ for words failed the check, i.e., those that do not begin with $12$.

Below we show that there is a~word $w$ such that $\cA(w)\neq 0$ if and only if
$\cA'$ is not commutative.

\begin{itemize}
	\item Assume there is $w$ such that $\cA(w)\neq 0$: Then
	$\cA'(12w)=\cA(w)\neq 0$. By definition, $\cA'(21w)=0$ and
	$21w\in perm(12w)$. It follows that $\cA'$ is not commutative.
	\item Assume that $\cA(w)= 0$ for all $w$:  
	\begin{itemize}
		\item For any word $v$ of the form
		$12w$, it holds that $\cA'(v)=\cA'(12w)=\cA(w)=0$. 
		\item For any word $v$ not in form of
		$12w$, it holds that $\cA'(v)=0$ (by definition).
	\end{itemize}
	It follows that $\cA'$ outputs $0$ on all inputs and, hence, $\cA'$ is commutative.
\end{itemize}

\subsection*{From commutativity to equivalence}
Let $\cA=\langle Q,q_0,F, \vu_0,\delta,\zeta\rangle$ be an $\raq$ over $X\cup Y$.
(Again we use a function instead of a vector to represent the initial contents of variables.)
W.l.o.g we assume that $\cA(w)=\emptyset$ for all $w$ such that $|w|\leq 1$ (commutativity of words of lengths zero or one is trivial).
We show the construction of deterministic $\raq$ $\cA_1$ and $\cA_2$
such that, for every word $w$, $\cA_1(w)=\cA(\pi_2(w))$ and $\cA_2(w)=\cA(\pi_*(w))$.

The construction of $\cA_1$ is straightforward, thus omitted.
We focus on the construction of $\cA_2$.
First, we make a copy $X'$ of $X$, and $Y'$ of $Y$.
That is, for every $x\in X$, we have its corresponding primed version $x' \in X'$.
Likewise for every $y \in Y$. 

Assuming that the states in $Q$ are of the form $q_i$,
define a set $F'=\{q_{(i\to j)}\mid q_i\in Q,q_j\in F\}$.
Intuitively, the set $F'$ has a copy of each state $q_i$ for each final state $q_j$.

We construct an $\raq$ $\cA_2 = \langle Q_2,q_{0,2},F_2,\vu_{0,2},\delta_2,\zeta_2\rangle$
over $(X_2,Y_2)$, where the components of $\cA_2$ and $(X_2, Y_2)$ are defined as follows:
\begin{itemize}\itemsep=0pt
\item
$X_2 = X \cup X' \cup \{x_t\}$, where $x_t$ is a new control variable,
and $Y_2 = Y \cup Y'$.
\item
$Q_2 = Q\cup F' \cup \{q_{0,2}\}$, where $q_{0,2}$ is a new state and $F'$ is as defined above.
\item
The new state $q_{0,2}$ is the initial state.
\item
$F_2 = F'$ as defined above is the set of final states.
\item
$\vu_{0,2}= \{z\mapsto \vu_0(z) \mid z\in X \cup Y\} \cup \{z \mapsto 0 \mid z\in X'\cup Y'\cup \{x_t\}\}$.
\item
$\zeta_2$ is defined as follows.

For all $q_{(i\to j)} \in F'$, the output function $\zeta_2(q_{(i\to j)})$ 
is the same as $\zeta(q_j)$ but substituting all the variables with their primed versions. 
For example, if $\zeta(q_j)=x_1+3y_3$ then $\zeta_2(q_{(i\to j)})=x_1'+3y_3'$.
\end{itemize}

Before the construction of the transition rules $\delta_2$, let us first explain the intuition behind.
\begin{itemize}
\item For each state $q_i$ in $\cA$, there are $|F|+1$ copies in $\cA_2$, that is, the state $q_i$, and the states $q_{i,j}$ with $q_j \in F$. 
\item For each state $q_{i,j}$ in $\cA_2$ and every two transitions $t_1,t_2$ in $\cA$ from some state $q_{i_0}$ to $q_i$ and from $q_i$ to $q_j$ respectively,  $\cA_2$ includes a transition $t'$ from $q_{i_0}$ to $q_{i,j}$ which reassigns $\vx'$ and $\vy'$ with the new values of $\vx$ and $\vy$ after the two transitions (that is, $t_1$ and $t_2$), while still reassigns $\vx$ and $\vy$ with their values after $t_1$ only. 
\item The only purpose of the variables $\vx'$ and $\vy'$ is to be used in the output. (Recall that $\zeta(q_{i,j})$ is defined as $\zeta(q_j)$, with $\vx$ and $\vy$ substituted by $\vx'$ and $\vy'$ respectively).
\item In addition, for each state $q_{i_1}$ in $\cA$ and each copy of $q_{i_1}$ in $\cA_2$, say $q_{i_1, j}$ for instance, each transition $t=(q_{i_1}, \varphi(\vx,\cur))\to(q_{i_2},A,B,\vb)$ in $\cA$ is split into multiple transitions of $\cA_2$ out of $q_{i_1, j}$, by splitting the  guard $\varphi(\vx,\cur)$ into mutually disjoint constraints, in order to make sure that $\cA_2$ is still deterministic.
\end{itemize}

The transitions in $\delta_2$ are defined as follows.
\begin{enumerate}[{\bf [T1]}]\itemsep=0pt
	\item $\delta_2$ contains the following transition to store the first $\cur$ in $x_t$
	\begin{eqnarray*}
		(q_{0,2},
		\ltrue) & \to & (q_0, \{x_t:=\cur\})
	\end{eqnarray*}
	\item For each pair of transitions 
	$t_1= (q_{i_1},\varphi(\vx,\cur))\to(q_{i_2},A_1,B_1,\vb_1), t_2=(q_{i_2},\varphi'(\vx,\cur))\to(q_{i_3},A_2,B_2,\vb_2) \in \delta$ such that $q_{i_3}\in F$, $\delta_2$ contains the following transitions to summarize the effect of $t_1$ and $t_2$:
	\begin{eqnarray*}
		(q_{i_1}, \left(\varphi(\vx,\cur)\wedge\varphi'(A_1 \myvec{\vx\\ \cur},x_t)\right)) & \to & (q_{(i_2\to i_3)},M)\\
		\forall q_j \in F, (q_{(i_1\to j)}, \left(\varphi(\vx,\cur)\wedge\varphi'(A_1\myvec{\vx\\ \cur},x_t)\right)) & \to & (q_{(i_2\to i_3)},M),
	\end{eqnarray*}	
	
	where $M$ encodes 
	$\vx:=A_1\myvec{\vx\\ \cur}$, 
	$\vy:=B_1\myvec{\vx \\ \vy \\ \cur}+\vb_1$, 
	$\vx':=A_2 \myvec{A_1\myvec{\vx\\ \cur}\\ x_t}$, 
	$\vy':=  B_2\myvec{A_1\myvec{\vx\\ \cur}\\ B_1\myvec{\vx \\ \vy \\ \cur}+\vb_1\\ x_t}+\vb_2$, and 
	$x_t:=x_t$.
	
	Intuitively, $M$ updates variables in $X\cup Y$ in the same way as $t_1$ does, updates variables in $X'\cup Y'$ to summarize the effect of $t_1$ and $t_2$ using $\cur$ as the first and $x_t$ as the second input symbol. The summarization is achievable by Proposition~\ref{prop:linear-raq}.

    \item We define $T_{q,F} \subseteq \delta$ as the set of transitions starting from $q$ and ending at a final state $F$.
    For each transition $t_1= (q_{i_1},\varphi(\vx,\cur))\to(q_{i_2},A,B,\vb)$, $\delta_2$ contains the following transitions
	\begin{eqnarray*}
		(q_{i_1},\varphi''(\vx, x_t,\cur)) & \to &(q_{i_2},M)\\
		\forall q_j \in F, (q_{(i_1\to j)},\varphi''(\vx, x_t,\cur))& \to &(q_{i_2},M),
	\end{eqnarray*}
   where $\varphi''(\vx, x_t,\cur)$ is a predicate defines as  $$\varphi(\vx,\cur)\wedge\bigwedge
   \{\neg\varphi'(A\myvec{\vx\\\cur},x_t)\mid 
   \varphi'(\vx,\cur) \mbox{ is the guard of a transition in } T_{q_{i_2},F}\},$$ 
   and $M$ encodes  $\vx:=A\myvec{\vx\\ \cur}$, 
   $\vy:=B\myvec{\vx \\ \vy \\ \cur}+\vb$, 
   $\vx':=\vx'$, 
   $\vy':= \vy'$, and 
   $x_t:=x_t$.
   
   Intuitively, the guard enforces that $\cA_2$ can take the {\bf [T3]}-type transition from a configuration $(q_{i_1}, \vu)$ only when all {\bf [T2]}-type transitions from $q_{i_1}$ cannot be taken.
   The variables in $X\cup Y$ are updated in the same way as $t_1$ does and the values of other variables remain unchanged.
\end{enumerate}

Observe that if we reset the initial state of $\cA_2$ to $q_0$, reset the final states of $\cA_2$
to $F$, and use $\zeta$ as the output function, then $\cA_2$ and $\cA$ are equivalent.

\hide{
\subsection{Results on the coverability problem}
\label{app:cov}

\begin{theorem}
	The invariant problem of $\raq$ can be reduced to the coverability problem of $\raq$, which can be further reduced to the reachability problem of $\raq$. Both reductions are in polynomial-time.
\end{theorem}

Let $\cA$ be an $\raq$.
The first reduction is done by creating (1) an $\raq$ $\cA'$ such that $\exists v \in\cA(w): v\geq 0$ for some $w$ iff $\exists v \in\cA'(w): v> 0$ for some $w$ by adding an arbitrary positive value to the output of $\cA$, e.g., add a $\cur>0$ and
(2) another $\raq$ $\cA''$ by negating all output expressions of $\cA'$. The answer to the non-zero problem of $\cA$ is positive iff the answers to the coverability problems of $\cA'$ and $\cA''$ are both positive.
The second reduction is done by adding a transition $(q,\cur\geq 0)\rightarrow (q', y_1:=\zeta(q)-\cur)$ from all final states $q$ to a new state $q'$ for some data variable $y_1$ and setting $q'$ as the only new final state with the output expression $y$.

\begin{corollary}
	The coverability problem of copyless $\raq$s is in \nexptime. 
\end{corollary}

\subsection{Results on configuration reachability and coverability}
\label{app:conf-cov-rea}

For Petri-net or VASS, people are also interested in configuration reachability and configuration coverability problems. The corresponding problems in $\raq$ are defined as follows.
Given an $\raq$ $\cA=\langle Q,q_0,F,\vu_0,\delta,\zeta\rangle$ and a configuration $(q_n,\vu_n)$, the configuration reachability problem asks if $(q_0,\vu_0)\vdash^{\ast} (q_n,\vu_n)$ and the configuration coverability problem asks if 
$(q_0,\vu_0)\vdash^{\ast} (q_n,\vu'_n)$ and $\vu'_n \geq \vu_n$. 
We show that the two problems in $\raq$ are inter-reducible in polynomial-time and they are not easier than the reachability problem, which is undecidable (Theorem~\ref{thm-reach-und}).

\begin{lemma}
	The configuration reachability problem of $\raq$ can be reduced to the configuration coverability problem of $\raq$, and vice versa. 
\end{lemma}
Let $\cA$ be an $\raq$ over $(X,Y)$ and the two problem are targeting the configuration $(q,\vv)$. 
The first reduction is done by creating an $\raq$ $\cA'$ over $(X\cup X', Y\cup Y')$ such that $X'$ and $Y'$ compute the negation of $X$ and $Y$, using a similar construction of the proof of Lemma 2 in~\cite{Haase14}. The reverse direction is done by adding the transition $(q, \bigwedge_{x_i\in X} x=\vv(X)[i])\rightarrow(q', \{\})$ and for all $y\in Y$ the transition $(q',\cur\geq 0) \rightarrow (q', \{y:=y-cur\})$ and targeting the configuration $(q',\vv)$ instead.

\begin{lemma}
	The coverability and reachability problems of $\raq$ can be reduced to the configuration coverability and reachability problems of $\raq$, respectively.
\end{lemma} 
The reduction is done by adding a transition with the assignment $y:=\zeta(q)$ from all final states $q$ to a new state $q'$ and use the configuration $(q',0)$ in the corresponding configuration coverability and reachability problems.
}

%% file: app-ondrej.tex
%%%%%%%%%%%%%%%%%%%%%%%%%%%%%%%%%%%%%%%%%%%%%%%%%%%%%%%%%%%%%%%%%%%%%%%%%%%%%%%%
\vspace{-0.0mm}
\subsection{Proof of Theorem~\ref{thm-reach-und}}\label{app:reachability-undec}
\vspace{-0.0mm}
%%%%%%%%%%%%%%%%%%%%%%%%%%%%%%%%%%%%%%%%%%%%%%%%%%%%%%%%%%%%%%%%%%%%%%%%%%%%%%%%

We show undecidability of the reachability problem
by a~reduction from the Post correspondence problem.
Let $\cP = \{(u_i, v_i)\}_{1 \leq i \leq n}$ be a~set of pairs of
sequences over the alphabet $\Sigma = \{1, \ldots, b-1\}$ for some base~$b$,
i.e., for all $1 \leq j \leq n$ it holds that $u_i, v_i \in \Sigma^*$.
The \emph{Post correspondence problem} (PCP) asks whether there exists a
sequence $i_1 \cdots i_k \in \{1, \ldots, n\}^+$ of indices such that
$u_{i_1} \cdots u_{i_k} = v_{i_1} \cdots v_{i_k}$.
The PCP is well known to be undecidable~\cite{Post46}.

Given $\cP$, we build a~deterministic $\raq$ $\cA_{\cP} = \langle
\{q_0,q_1\},q_0,\{q_1\},\vu_0,\delta,\zeta\rangle$ over $X=\emptyset$ and $Y= \{y_1, y_2\}$, where $\vu_0
= \myvec{0,0}$ (i.e., the initial configuration of registers is $y_1 = 0$ and
$y_2 = 0$), the output function is $\zeta(q_1) = y_1 - y_2$, and $\delta$ is
constructed as described later.
The reduction is based on treating strings over~$\Sigma$ as natural numbers in
base-$b$ encoding and representing concatenation of strings using
linear expressions, e.g., if we assume $b=10$, the concatenation of strings $1.2.3
\in \Sigma^*$ and $4.5.6 \in \Sigma^*$ can be represented using natural numbers
as $123 \cdot 10^3 + 456 = 123456$.
% encoding strings over an $m$-symbol alphabet as integer numbers in the
% base~$m$ numeral system and representing concatenation of an~$l$-symbol-long
% string $b_1\cdots b_l$ as a~multiplication by the $l$-th power of~$m$,
% following by addition of the number representing $b_1 \cdots b_l$.
In particular, for every $1 \leq i \leq n$, we translate the pair $(u_i,
v_i) \in \cP$ into the pair of transitions
\begin{align*}
&(q_0, (\cur = i)) \to (q_1, \myvec{\omega_1; \omega_2;}) \in \delta \\
&(q_1, (\cur = i)) \to (q_1, \myvec{\omega_1; \omega_2;}) \in \delta,
\end{align*}
such that, if $u_i = g_1 \cdots g_r$ and
$v_i = h_1 \cdots h_s$, the term $\omega_1$ is the assignment $y_1
:= b^r \cdot y_1 + e(g_1 \cdots g_r)$ and $\omega_2$ is the assignment $y_2 := b^s
\cdot y_2 + e(h_1 \cdots h_s)$ where we use $e(g_1 \cdots g_r)$ to represent the
number $g_1 b^{r-1} + \ldots + g_r b^{0}$ (similarly for $e(h_1 \cdots
h_s)$); for instance, for $b = 5$, we have $e(4 . 2) = 42_5$
(note that $e: \Sigma^* \to \bbN$).
In Figure~\ref{fig:pcp-example}, we show an example of an encoding of a~PCP
instance into an~SNT.
The transitions from $q_0$ to $q_1$ are necessary in order to
``bootstrap'' the computation (a~solution to a~PCP cannot be an~empty sequence
of indices).
The crucial property of~$\cA_\cP$ is that for a~sequence $i_1 \cdots i_k \in
\{1, \ldots, n\}^+$, it holds that $u_{i_1} \cdots u_{i_k} =
v_{i_1} \cdots v_{i_k}$ (i.e., $i_1 \cdots i_k$ is a~solution of~$\cP$) iff
$\cA_{\cP}(i_1 \cdots i_k) = 0$, which we prove in the following.

\begin{figure}[t]
\begin{center}
\usetikzlibrary{arrows}
\usetikzlibrary{automata}

\begin{tikzpicture}
[->,>=stealth']
\tikzstyle{state}=[circle,draw,inner sep=0.8mm]
%\tikzstyle{trans}=[scale=0.7,rectangle split,rectangle split parts=2,rectangle split draw splits,draw]
\tikzstyle{trans}=[scale=0.7]

\node (init) {};
\node[state] (q0) [right of=init] {$q_0$};
\node[state,accepting] (q1) [right of=q0,node distance=85mm] {$q_1$};
\node (output) [right of=q1] {};

\draw (init) edge (q0);

\draw (q1) .. controls +(10mm,10mm) and +(-10mm,10mm) .. node[above,trans] {$\begin{array}{c}\cur{} = 1 \\\hline \\[-4mm] \left[\begin{array}{lrlr}y_1 := &100 \cdot y_1 &+& 12;\\ y_2 := &10 \cdot y_2 &+ &1;\end{array}\right]\end{array}$} (q1);

\draw (q1) .. controls +(10mm,-10mm) and +(-10mm,-10mm) .. node[below,trans] {$\begin{array}{c}\cur{} = 2 \\\hline \\[-4mm] \left[\begin{array}{lrlr}y_1 := &10 \cdot y_1 &+& 3;\\ y_2 := &100 \cdot y_2 &+ &23;\end{array}\right]\end{array}$} (q1);

\draw (q0) edge node[above,trans] {$\begin{array}{c}\cur{} = 1 \\\hline \\[-4mm] \left[\begin{array}{lrlr}y_1 := &100 \cdot y_1 &+& 12;\\ y_2 := &10 \cdot y_2 &+ &1;\end{array}\right]\end{array}$} node[below,trans] {$\begin{array}{c}\cur{} = 2 \\\hline \\[-4mm] \left[\begin{array}{lrlr}y_1 := &10 \cdot y_1 &+& 3;\\ y_2 := &100 \cdot y_2 &+ &23;\end{array}\right]\end{array}$} (q1);

\draw (q1) edge node[below,xshift=7mm,trans] {$\zeta = y_1 - y_2$} (output);

\end{tikzpicture}

\end{center}
\caption{An example of our encoding of the PCP instance $\{(1.2, 1),
  (3, 2.3)\}$ over the alphabet $\{1, \ldots, 9\}$ into an SNT (all
  numbers are given in base-4).}
\label{fig:pcp-example}
\end{figure}
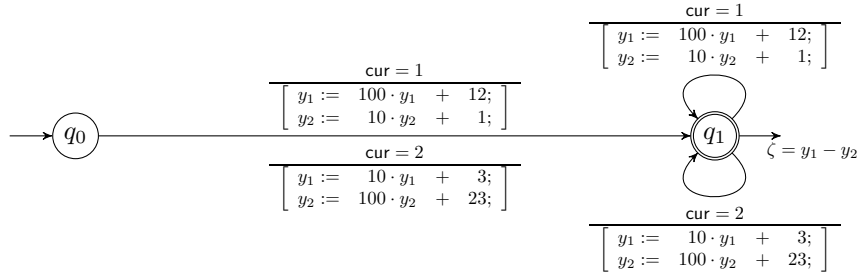

First, we want to show that after reading $i_1 \cdots i_l$, for $l \geq 1$ and
$\forall 1 \leq k \leq l: 1 \leq i_k \leq n$, $\cA$ is in th state $q_1$ and the
content of the variable~$y_1$ is $e(u_{i_1}\,.\,\cdots \,.\, u_{i_l})$ and the
content of the variable~$y_2$ is $e(v_{i_1}\,.\, \cdots\,.\,v_{i_l})$.
We show the previous by induction on the length of $i_1 \cdots i_l$.
For $l = 1$, after reading $i_1$, $\cA_{\cP}$ will move from the state~$q_0$ to
the state~$q_1$, taking the transition $(q_0, (\cur = i_1)) \to (q_1,
\myvec{\omega_1; \omega_2;}) \in \delta$, which sets the new content of the
variable~$y_1$ to $0 + e(u_{i_1})$ and the content of the variable $y_2$ to
$0 + e(v_{i_1})$.
For the inductive step, assume that $\cA_{\cP}$ is after reading $i_1 \cdots
i_l$ in the state $q_1$, the content of the registers $y_1$ and $y_2$ being
$e(u_{i_1}\,.\,\cdots \,.\, u_{i_l})$ and $e(v_{i_1}\,.\,\cdots
\,.\, v_{i_l})$ respectively.
If $\cA_{\cP}$ now reads the symbol $i_{l+1}$, it takes the transition $(q_1,
(\cur = i_{l+1})) \to (q_1, \myvec{\omega_1; \omega_2;}) \in \delta$
where~$\omega_1$ is of the form $y_1 := y_1 b^r + e(u_{i_{l+1}} = g_1
\cdots g_r)$ and $\omega_2$ is of the form $y_2 := y_2 b^s + e(v_{i_{l+1}} =
h_1 \cdots h_s)$.
The new value of the variable~$y_1$ will therefore be
$e(u_{i_1}\,.\,\cdots \,.\, u_{i_l}) \cdot b^r + e(g_1 \cdots
g_r) = e(u_{i_1}\,.\,\cdots \,.\, u_{i_l}\,.\, u_{i_{l+1}})$,
the new value of the variable~$y_2$ will be
$e(v_{i_1}\,.\,\cdots \,.\, v_{i_l}) \cdot b^s + e(h_1 \cdots
h_s) = e(v_{i_1}\,.\,\cdots \,.\, v_{i_l}\,.\, v_{i_{l+1}})$,
and the next state will be~$q_1$, which concludes this part of the proof.

Based on the previous paragraph and the fact that the only output of $\cA_{\cP}$
is in the state~$q_1$ and has the value computed as $y_1 - y_2$, we infer that
$$
\cA_{\cP}(i_1 \cdots i_k) = e(u_{i_1}\,.\,\cdots \,.\, u_{i_k}) -
e(v_{i_1}\,.\,\cdots \,.\, v_{i_k}).
$$
What remains to show is the following:
\begin{equation*}
u_{i_1}\,.\,\cdots \,.\, u_{i_k} = v_{i_1}\,.\,\cdots \,.\, v_{i_k}
\qquad \text{iff}
\qquad
e(u_{i_1}\,.\,\cdots \,.\, u_{i_k}) = e(v_{i_1}\,.\,\cdots \,.\, v_{i_k}) ,
\end{equation*}
which holds because $e$~is a~bijection.
%%%%%%%%%%%%%%%%%%%%%%%%%%%%%%%%%%
\hide{
The proof that the configuration reachability problem is also undecidable is
very similar.
For a~given PCP~$\cP$, we construct a~deterministic $\raq$ $\cA_\cP' = \langle
\{q_0,q_1,q_2\},q_0,\emptyset,\vu_0,\delta',\emptyset\rangle$ where $\vu_0$ is
the same as in~$\cA_\cP$ and
$$
\delta' = \delta \cup \{(q_1, (\cur = \#)) \to (q_2, \myvec{y_1 := y_1 - y_2;
y_2 := 0;})\}
$$
where $\# = b$.
Then it holds that $i_1 \cdots i_k$ is a~solution of~$\cP$ iff after $\cA_\cP'$
reads $i_1 \cdots i_k \#$, it will be in the configuration $(q_2, \myvec{0,0})$,
i.e., $\cP$ has a~solution iff the configuration $(q_2, \myvec{0,0})$ is
reachable in $\cA_\cP'$.
}
%%%%%%%%%%%%%%%%%%%%%%%%%%%%%%%%%%

A more succint proof can also be obtained by noticing that
the reachability problem of $\raq$s can encode
instances of the scalar reachability problem of matrices, which is
undecidable~\cite{Halava07}.

%% file: app-zhilin.tex
%!TEX root = snt-main.tex

%%%%%%%%%%%%%%%%%%%%%%%%%%%%%%%%%%%%%%%%%%%%%%%%%%%%%%%%
%%%%%%%%%%%%%%%%%%%%%%%%%%%%%%%%%%%%%%%%%%%%%%%%%%%%%%%%
\hide{
\subsection{Proposition~\ref{prop-po-one-interval}: Its proof and generalizations}\label{app-pf-prop-po-one-interval}

We first give a proof for Proposition~\ref{prop-po-one-interval}, then discuss its generalizations.

\begin{proof}[Proof of Proposition~\ref{prop-po-one-interval}]

Suppose $c, d \in \ratnum$ such that $c < d \le 0$ or $0 \le c < d$, $Z=\{z_1,\dots, z_m\}$ is a set of variables, $\preceq$ is a partial order on $Z$, and $a_1, \dots, a_m \in \ratnum$. Our goal is to analyze the set of values of $a_1 \eta(z_1) + \dots + a_m \eta(z_m)$, for the assignments $\eta: Z \rightarrow (c,d)$ that are consistent with $\preceq$.
Without loss of generality, we assume that for each $i \in [m]$, $a_i \neq 0$.  
Otherwise, we can remove those $a_i z_i$ such that $a_i = 0$, without affecting the value of  $a_1 z_1 + \dots + a_m z_m$. 

Let $\pow^\uparrow_{\preceq}(Z)$ denote the set of upward-closed subsets of $Z$ w.r.t. $\preceq$, and let $\alpha, \beta$ be the minimum resp. maximum of the values $c \sum \limits_{z_j  \in Z \setminus Z' } a_j + d \sum \limits_{z_j \in Z'} a_j$, such that $Z'$ ranges over $\pow^\uparrow_{\preceq}(Z)$. 

In order to show 
$$
\left\{ \sum \limits_{j \in [m]} a_j\eta(z_j) \  \big\vert\  \eta: Z \rightarrow (c, d) \mbox{ is consistent with } \preceq \right\}= (\alpha, \beta),$$
 it is sufficient to show the following two facts.
\begin{enumerate}
\item For each assignment $\eta: Z \rightarrow (c, d)$ such that $\eta$ is consistent with $\preceq$, there are $Z'_{1}, Z'_{2} \in \pow^\uparrow_{\preceq}(Z)$ such that 
$$c \sum \limits_{z_j  \in Z \setminus Z'_{1}} a_j + d \sum \limits_{z_j \in Z'_{1}} a_j  < \sum \limits_{z_j \in Z} a_j \eta(z_j) <  c \sum \limits_{z_j \in Z \setminus  Z'_{2}} a_j + d \sum \limits_{z_j \in Z'_{2}} a_j.$$

\item For each $d' \in (\alpha, \beta)$, 
there is an assignment $\eta: Z \rightarrow (c, d)$ such that $\eta$ is consistent with $\preceq$ and $d' = \sum \limits_{z_j \in Z} a_j  \eta(z_j)$.
\end{enumerate}

\smallskip

\noindent {\bf \large The first fact}.

Suppose $\eta: Z \rightarrow (c, d)$ is an assignment such that $\eta$ is consistent with $\preceq$. Without loss of generality, suppose $\eta(z_1) \le \dots \le \eta(z_m)$. This implies that for each $z_{i_1}, z_{i_2} \in Z$ such that $z_{i_1} \preceq z_{i_2}$, we have $i_1 \le i_2$. Then for each $i \in [m]$, the following two facts hold,
\begin{itemize}
\item $\{z_{i}, \dots, z_{m}\} \in \pow^\uparrow_{\preceq}(Z)$, 
\item for each  $Z' \in \pow^\uparrow_{\preceq}(\{z_1,\dots, z_{i-1}\})$, it holds that $Z' \cup \{z_i, \dots, z_m\} \in \pow^\uparrow_{\preceq}(Z)$.
\end{itemize}

%We claim that there must be a upward-closed subset $Z$ such that $a_0+a_1 u_1 + \dots + a_m u_m > a_0 + c_1\sum \limits_{z_j \not \in Z} a_j + c_2 \sum \limits_{z_j \in Z} a_j \ge 0$. Next we present the arguments for this claim.
 
 We introduce some additional notations. A \emph{zone} of of the sequence $a_{1}, \dots, a_{m}$ is a maximal subsequence $a_{i}, \dots, a_{i'}$ (where $i \le i'$)  such that all the rationals in the subsequence are of the same sign, that is, either all of them are positive or all of them are negative. A zone is called \emph{positive} (resp. \emph{negative}) if all the rationals in the subsequence are positive (resp. negative).
 
% A \emph{positive (resp. negaive) zone} of the sequence $a_{j_1}, \dots, a_{j_t}$ is a \emph{maximal} subsequence $a_{j_{i}}, \dots, a_{j_{i'}}$ (where $i < i'$) such that for each $ i'' \in [i, i']$, $a_{j_{i''}}$ is positive (resp. for each $i'' \in [i, i']$, $a_{j_{i''}}$ is negative). A zone is  either a positive zone or a negative zone.
 
We present the arguments by an induction on the number of zones of the sequence $a_{1}, \dots, a_{m}$.

\medskip

\noindent {\it Induction base}: There is exactly one zone. 
\begin{itemize}
\item If the zone is positive, then for each $i \in [m]$, $a_{i} > 0$. Therefore, $c \sum \limits_{i  \in [m]} a_{i}  < \sum \limits_{i \in [m]} a_{i} \eta(z_i)  <  d \sum \limits_{i \in [m]} a_{i}$. 
Let $Z'_{1} = \emptyset$ an $Z'_{2} = Z$. Then 
$$c \sum \limits_{z_j  \in Z \setminus Z'_{1}} a_j + d \sum \limits_{z_j \in Z'_{1}} a_j  <  \sum \limits_{z_j \in Z} a_{j} \eta(z_{j})  < c \sum \limits_{z_j  \in Z \setminus Z'_{2}} a_j + d \sum \limits_{z_j \in Z'_{2}} a_j.$$ 
\item Otherwise, for each $i \in [m]$, $a_{i} < 0$. Then $d \sum \limits_{i \in [m]} a_{i}  < \sum \limits_{i \in [m]} a_{i} \eta(z_{i})  < c \sum \limits_{i \in [m]} a_{i}$. 
Let $Z'_{1} = Z$ an $Z'_{2} = \emptyset$. Then 
$$c \sum \limits_{z_j  \in Z \setminus Z'_{1}} a_j + d \sum \limits_{z_j \in Z'_{1}} a_j  <  \sum \limits_{z_j \in Z} a_{j} \eta(z_{j})  < c \sum \limits_{z_j  \in Z \setminus Z'_{2}} a_j + d \sum \limits_{z_j \in Z'_{2}} a_j.$$ 
\end{itemize}

\medskip

\noindent {\it Induction step}: There are at least two zones.

Suppose $a_{i}, \dots, a_{m}$ is the zone containing $a_{m}$.  
We present the arguments for the situation that $a_{i}, \dots, a_{m}$ is a positive zone. The arguments for the situation that  $a_{i}, \dots, a_{m}$ is a negative zone are symmetric.

\medskip

\noindent {\bf 
There is $Z'_{2} \in \pow^\uparrow_{\preceq}(Z)$ such that 
$ \sum \limits_{z_j \in Z} a_{j} \eta(z_{j}) < c \sum \limits_{z_j \in Z \setminus  Z'_{2}} a_j + d \sum \limits_{z_j \in Z'_{2}} a_j .
$}

\smallskip

For each $i' \in [i, m]$, $a_{i'} > 0$, thus $a_{i'} \eta(z_{i'}) < a_{i'}d$. Therefore,   
\begin{eqnarray}\label{eqn-inf-sup-1}
\sum \limits_{z_j \in Z} a_j \eta(z_j) = \sum \limits_{i' \in [i-1]} a_{i'} \eta(z_{i'}) +   \sum \limits_{i' \in [i, m]} a_{i'} \eta(z_{i'}) <  \sum \limits_{i' \in [i-1]} a_{i'} \eta(z_{i'}) +  d \sum \limits_{i' \in [i, m]} a_{i'}.
\end{eqnarray}

Since $a_i, \dots, a_m$ is a positive zone and $a_{i-1} < 0$, the number of zones of the sequence $a_1, \dots, a_{i-1}$ is strictly less than that of $a_{1}, \dots, a_{m}$. 
%Let $\preceq' = \preceq \cap\ (\{z_1,\dots, z_{i-1}\} \times \{z_1,\dots, z_{i-1}\})$. 
By the induction hypothesis, there is $Z''_{2} \in \pow^\uparrow_{\preceq}(\{z_{1},\dots, z_{i-1}\})$ such that  
\begin{eqnarray}\label{eqn-inf-sup-2}
\sum \limits_{i' \in [i-1]} a_{i'} \eta(z_{i'})  < c \sum \limits_{z_j \in \{z_{1},\dots, z_{i-1}\} \setminus  Z''_{2}} a_j + d \sum \limits_{z_j \in Z''_{2}} a_j.
\end{eqnarray}

Let $Z'_{2} = Z''_{2} \cup \{z_{i}, \dots, z_{m}\}$.  Then $Z'_{2} \in \pow^\uparrow_{\preceq}(Z)$ and
$$
%\begin{array}{l c l}
\sum \limits_{z_j \in Z} a_j \eta(z_j)  < \sum \limits_{i' \in [i-1]} a_{i'} \eta(z_{i'}) +  d \sum \limits_{i' \in [i, m]} a_{i'} 
%\\
%& < & 
< c \sum \limits_{z_j \in \{z_{1},\dots, z_{i-1}\} \setminus  Z''_{2}} a_j + d \sum \limits_{z_j \in Z''_{2}} a_j +  d \sum \limits_{i' \in [i, m]} a_{i'}
%&= & 
= c \sum \limits_{z_j \in Z \setminus  Z'_{2}} a_j + d \sum \limits_{z_j \in Z'_{2}} a_j,
%\end{array}
$$
where the first and second inequality above follow from the inequality~(\ref{eqn-inf-sup-1}) and (\ref{eqn-inf-sup-2}) respectively, and the equality above follows from the fact $\{z_1,\dots, z_{i-1}\} \setminus Z''_2 = \{z_1,\dots, z_{i-1}, z_i, \dots, z_m\} \setminus (Z''_2 \cup \{z_i,\dots, z_m\})= Z \setminus Z'_2$.

\medskip

\noindent {\bf 
There is $Z'_{1} \in \pow^\uparrow_{\preceq}(Z)$ such that 
$c \sum \limits_{z_j \in Z \setminus  Z'_{1}} a_j + d \sum \limits_{z_j \in Z'_{1}} a_j < \sum \limits_{z_j \in Z} a_j \eta(z_j) .
$}

\smallskip

We distinguish between the following two cases.
\begin{itemize}
\item Case I: For each $i' \in [i-1]$, $ \sum \limits_{i'' \in [i', m]} a_{i''} \ge 0$.
%
%\item $i=2$, $\sum \limits_{i' \in [m]} a_{i'}>0$,
%
%\item $i=2$, $\sum \limits_{i' \in [m]} a_{i'} = 0$,
%
\item Case II: There is $i' \in [i-1]$ such that $ \sum \limits_{i'' \in [i', m]} a_{i''} < 0$.
\end{itemize}

\smallskip

\noindent {\it Case I: For each $i' \in [i-1]$, $ \sum \limits_{i'' \in [i', m]} a_{i''} \ge 0$}. 

\smallskip

In this case, since $a_i,\dots, a_m$ is a positive zone, actually  for each $i' \in [m]$, we have $\sum \limits_{i'' \in [i', m]} a_{i''} \ge 0$.

\smallskip

\noindent {\it Subcase I.I: There is $i' \in [2, m]$ such that $\sum \limits_{i'' \in [i', m]} a_{i''} > 0$ and $\eta(z_{i'-1}) < \eta(z_{i'})$}.

\smallskip

From $\eta(z_m) \ge \dots \ge \eta(z_{i'})$ and the fact that  for each $i'' \in [i', m]$, $ \sum \limits_{i''' \in [i'', m]} a_{i'''} \ge 0$, we deduce that $\sum \limits_{i'' \in [i', m]} a_{i''} \eta(z_{i''}) \ge \eta(z_{i'}) \sum  \limits_{i'' \in [i', m]} a_{i''}$. Then
$$
\begin{array} {l c l }
\sum \limits_{z_j \in Z} a_j \eta(z_j) & =& \sum \limits_{i'' \in [i'-1]} a_{i''}   \eta(z_{i''}) +\sum  \limits_{i'' \in [i', m]} a_{i''}  \eta(z_{i''}) \ge  \sum \limits_{i'' \in [i'-1]} a_{i''}   \eta(z_{i''}) + \eta(z_{i'}) \sum  \limits_{i'' \in [i', m]} a_{i''}\\
& > &  \sum \limits_{i'' \in [i'-1]} a_{i''}   \eta(z_{i''}) + \eta(z_{i'-1}) \sum  \limits_{i'' \in [i', m]} a_{i''}  =  \sum \limits_{i'' \in [i'-2]} a_{i''}   \eta(z_{i''}) + \eta(z_{i'-1}) \sum  \limits_{i'' \in [i'-1, m]} a_{i''} \\
& \ge  &
 \dots  \ge \eta(z_{1}) \sum \limits_{i'' \in [m]} a_{i''} \ge c \sum \limits_{i' \in [m]} a_{i''},
\end{array}
$$
where the strict inequality above follows from the fact that $\sum \limits_{i'' \in [i', m]} a_{i''} > 0$ and $\eta(z_{i'-1}) < \eta(z_{i'})$.

Let $Z'_{1} = \emptyset$. Then 
$c  \sum \limits_{z_j \in Z \setminus Z'_{1}} a_j + d \sum \limits_{z_j \in Z'_{1}} a_j = c \sum \limits_{i'' \in [m]} a_{i''} < \sum \limits_{z_j \in Z} a_j \eta(z_j) $.

\smallskip

\noindent {\it Subcase I.II: 
For each $i' \in [2, m]$, either $\sum \limits_{i'' \in [i', m]} a_{i''} = 0$ or $\eta(z_{i'-1}) = \eta(z_{i'})$.}

\smallskip

If $\sum \limits_{i' \in [m]} a_{i'} > 0$, then from $\eta(z_m) \ge \dots \ge \eta(z_{1})$ and the fact that  for each $i' \in [m]$, $ \sum \limits_{i'' \in [i', m]} a_{i''} \ge 0$, we deduce that  $\sum \limits_{z_j \in Z} a_j \eta(z_j) \ge \eta(z_1) \sum \limits_{i' \in [m]} a_{i'} > c \sum \limits_{i' \in [m]} a_{i'}$, where the last equality follows from $\eta(z_1) > c$ and $\sum \limits_{i' \in [m]} a_{i'} > 0$. Let $Z'_1 = \emptyset$. Then
$ c  \sum \limits_{z_j \in Z \setminus Z'_{1}} a_j + d \sum \limits_{z_j \in Z'_{1}} a_j = c \sum \limits_{i' \in [m]} a_{i'} < \sum \limits_{z_j \in Z} a_j \eta(z_j) $.

\smallskip

Otherwise, $\sum \limits_{i' \in [m]} a_{i'} = 0$. Since $\sum \limits_{i' \in [2, m]} a_{i'} \ge 0$ and $a_1 \neq 0$, it follows that $a_1 < 0$ and $\sum \limits_{i' \in [2, m]} a_{i'} > 0$. 
\begin{itemize}
\item If $\eta(z_1) = \dots = \eta(z_m)$, then since $\eta$ is consistent with $\preceq$, we infer that the partial order $\preceq$ is trivial, i.e., $\preceq=\{(z_i, z_i) \mid i \in [m]\}$. Let $Z'_1 = \{z_1\}$. Since $\sum \limits_{i' \in [2,m]} a_{i'}  = -a_1$ and $a_1 < 0$, 
$$c  \sum \limits_{z_j \in Z \setminus Z'_{1}} a_j + d \sum \limits_{z_j \in Z'_{1}} a_j   = c \sum \limits_{i' \in [2,m]} a_{i'} + d a_1 = a_1 (d- c) < 0  = \eta(z_1) \sum \limits_{i' \in [m]} a_{i'} = \sum \limits_{z_j \in Z} a_j \eta(z_j).$$ 
\item Otherwise, let $i' \in [2, m]$ such that $\eta(z_{i'-1}) < \eta(z_{i'})$. Then $\sum \limits_{i'' \in [i', m]} a_{i''} = 0$. From this, we deduce that
$$
%\begin{array}{l c l }
\sum \limits_{z_j \in Z} a_j \eta(z_j)   =   \sum \limits_{i'' \in [i'-1]} a_{i''}   \eta(z_{i''}) +\sum  \limits_{i'' \in [i', m]} a_{i''}  \eta(z_{i''})  
 \ge  \sum \limits_{i'' \in [i'-1]} a_{i''}   \eta(z_{i''}) + \eta(z_{i'}) \sum  \limits_{i'' \in [i', m]} a_{i''} 
 =  \sum \limits_{i'' \in [i'-1]} a_{i''}   \eta(z_{i''}).
%\end{array}
$$
Since $\sum \limits_{i'' \in [i', m]} a_{i''} = 0$ and $a_i, \dots, a_m$ is a positive zone, it follows that $i' \le i-1$. Therefore, the number of zones in the sequence $a_1, \dots , a_{i'-1}$ is strictly less than that of $a_1,\dots, a_m$. \\
 Consider the expression $a_{1}  z_{1} + \dots + a_{i'-1} z_{i'-1}$. By the induction hypothesis, there is $Z''_1 \in \pow^\uparrow_{\preceq}(\{z_1,\dots, z_{i'-1}\})$ such that 
 $ c  \sum \limits_{z_j \in Z \setminus Z''_{1}} a_j + d \sum \limits_{z_j \in Z''_{1}} a_j <  \sum \limits_{i'' \in [i'-1]} a_{i''} \eta(z_{i''})$. \\
 Let $Z'_1 = Z''_1 \cup \{z_{i'}, \dots, z_m\}$. Then $Z'_1 \in \pow^\uparrow_{\preceq}(Z)$, and 
 $$
 \begin{array} {l c l}
c  \sum \limits_{z_j \in Z \setminus Z'_{1}} a_j + d \sum \limits_{z_j \in Z'_{1}} a_j   & = & c  \sum \limits_{z_j \in \{z_1,\dots, z_{i'-1}\} \setminus Z''_{1}} a_j + d \sum \limits_{z_j \in Z''_{1}} a_j  + d \sum \limits_{i'' \in [i', m]} a_{i''} \\
& = & c  \sum \limits_{z_j \in \{z_1,\dots, z_{i'-1}\} \setminus Z''_{1}} a_j  + d \sum \limits_{z_j \in Z''_{1}} a_j   < \sum \limits_{i'' \in [i'-1]} a_{i''} \eta(z_{i''}) \\
& \le &  \sum \limits_{z_j \in Z} a_j \eta(z_j).
\end{array}
 $$
\end{itemize}

\smallskip

\noindent {\it Case II: There is $i' \in [i-1]$ such that $ \sum \limits_{i'' \in [i', m]} a_{i''} < 0$}. 

\smallskip

Let $a_{i'}, \dots, a_{i-1}$ be the \emph{maximal} suffix of $a_1,\dots, a_{i-1}$ such that for each $i'' \in [i', i-1]$, $\sum \limits_{i''' \in [i'', m]} a_{i'''} \ge 0$.  Then $i' > 1$ and $\sum \limits_{i'' \in [i'-1, m]} a_{i''} < 0$. From this, we know that $a_{i'-1} < 0$. 

Since $a_i, \dots, a_m$ is a positive zone, we deduce that for each $i'' \in [i', m]$,  $\sum \limits_{i''' \in [i'', m]} a_{i'''} \ge 0$.
In addition, $\eta(z_{i'-1}) \le \dots \le \eta(z_{m})$. Then we have 
$$
\begin{array}{l c l }
\sum \limits_{z_j \in Z} a_j \eta(z_j) & = &  \sum \limits_{i'' \in [i'-1]} a_{i''} \eta(z_{i''}) +  \sum \limits_{i'' \in [i', m]} a_{i''} \eta(z_{i''-1})  
 \ge  \sum \limits_{i'' \in [i'-1]} a_{i''} \eta(z_{i''}) +  \eta(z_{i'-1}) \sum \limits_{i'' \in [i', m]} a_{i''} \\
 & = & \sum \limits_{i'' \in [i'-2]} a_{i''} \eta(z_{i''}) + \eta(z_{i' -1})  \sum \limits_{i'' \in [i'-1, m]} a_{i''}.
\end{array}
$$
Consider the expression $ \sum \limits_{i'' \in [i'-1]} a'_{i''} z_{i''}$ such that for each $i'' \in [i'-2]$, $a'_{i''} = a_{i''}$, and $a'_{i'-1} = \sum \limits_{i'' \in [i'-1, m]} a_{i''}$. 

Since $a'_{i'-1} =  \sum \limits_{i'' \in [i'-1, m]} a_{i''} < 0$ and $a_{i'-1} < 0$, the number of zones of $a'_{1}, \dots,  a'_{i'-1}$ is the same as  that of $a_{1}, \dots, a_{i'-1}$, thus strictly less than that of $a_{1},\dots, a_{m}$.

By the induction hypothesis, there is $Z''_{1} \in \pow^\uparrow_{\preceq}(\{z_{1},\dots, z_{i'-1}\})$ such that  
$$ c \sum \limits_{z_j \in \{z_{1},\dots, z_{i'-1}\} \setminus  Z''_{1}} a'_j + d \sum \limits_{z_j \in Z''_{1}} a'_j < \sum \limits_{i'' \in [i'-1]} a'_{i''} \eta(z_{i''}).$$
%$$ a'_0 + c_1\sum \limits_{z_j \in \{z_1,\dots, z_{i-1}\} \setminus  Z''_1} a'_j + c_2 \sum \limits_{z_j \in Z''_1} a'_j < a_0 + \sum \limits_{j \in [i-2]} a_j \eta(z_j) +  \left(a_{i-1} + \sum \limits_{j \in [i, m]} a_j \right)\ \eta(z_{i-1}).$$
%
We can assume that $z_{i'-1} \in Z''_{1}$, since otherwise, let $Z'''_{1} = Z''_{1} \cup \{z_{i'-1}\}$, then $Z'''_{1} \in \pow^\uparrow_{\preceq}(\{z_{1},\dots, z_{i'-1}\})$, and 
$$
\begin{array}{l c l}
& & c \sum \limits_{z_j \in \{z_{1},\dots, z_{i'-1}\} \setminus  Z'''_{1}} a'_j + d \sum \limits_{z_j \in Z'''_{1}} a'_j  
= c \sum \limits_{z_j \in \{z_{1},\dots, z_{i'-2}\} \setminus  Z''_{1}} a'_j + d \sum \limits_{z_j \in Z''_{1}} a'_j + d  a'_{i'-1}\\
& < &  c \sum \limits_{z_j \in \{z_{1},\dots, z_{i'-2 }\} \setminus  Z''_{1}} a'_j + d \sum \limits_{z_j \in Z''_{1}} a'_j + c  a'_{i'-1} 
=  c \sum \limits_{z_j \in \{z_{1},\dots, z_{i'-1}\} \setminus  Z''_{1}} a'_j + d \sum \limits_{z_j \in Z''_{1}} a'_j \\
& < &  \sum \limits_{i'' \in [i'-1]} a'_{i''} \eta(z_{i''}),
\end{array}
$$
where the first inequality follows from the fact $a'_{i'-1} < 0$.

Therefore, we assume that $z_{i'-1} \in Z''_{1}$. Let $Z'_{1} = Z''_{1} \cup \{z_{i'}, \dots, z_{m}\}$. Then $Z'_{1} \in \pow^\uparrow_{\preceq}(Z)$, and
$$
\begin{array}{l c l}
 c \sum \limits_{z_j \in Z \setminus  Z'_{1}} a_j + d \sum \limits_{z_j \in Z'_{1}} a_j    & = &  c \sum \limits_{z_j \in \{z_{1},\dots, z_{i'-1}\} \setminus  Z''_{1}} a_j + d (\sum \limits_{z_j \in Z''_{1}} a_j + \sum \limits_{i'' \in [i', m]} a_{i''}) \\
& = & c \sum \limits_{z_j \in \{z_{1},\dots, z_{i'-1}\} \setminus  Z''_{1}} a_j + d (\sum \limits_{z_j \in (Z''_{1} \setminus \{z_{i'-1}\})} a_j + \sum \limits_{i'' \in [i'-1, m]} a_{i''})\\
& = &  c \sum \limits_{z_j \in \{z_{1},\dots, z_{i'-1}\} \setminus  Z''_{1}} a'_j + d \sum \limits_{z_j \in Z''_{1}} a'_j  < \sum \limits_{i'' \in [i'-1]} a'_{i''} \eta(z_{i''})\\
&= & \sum \limits_{i'' \in [i'-2]} a_{i''} \eta(z_{i''}) + \eta(z_{i'-1})  \sum \limits_{i'' \in [i'-1, m]} a_{i''}  
\le \sum \limits_{z_j \in Z} a_j \eta(z_j).
\end{array}
$$

\medskip

\noindent {\bf \large The second fact}. 

\smallskip

Suppose $d' \in (\alpha, \beta)$. 
Our goal is to construct an assignment $\eta: Z \rightarrow (c, d)$ such that $\eta$ is consistent with $\preceq$ and $d' = \sum \limits_{z_j \in Z} a_j \eta(z_j)$. 
%$\sum \limits_{i \in [t]} a_{j_i}  \eta(z_{j_i})$.

Let $Z'_{1}, Z'_{2} \in \pow^\uparrow_{\preceq}(Z)$ such that $\alpha = c \sum \limits_{z_j  \in Z \setminus Z'_{1}} a_j + d \sum \limits_{z_j \in Z'_{1}} a_j$ and $\beta =  c \sum \limits_{z_j  \in Z \setminus Z'_{2}} a_j + d \sum \limits_{z_j \in Z'_{2}} a_j$.  Then 
$$c \sum \limits_{z_j  \in Z \setminus Z'_{1}} a_j + d \sum \limits_{z_j \in Z'_{1}} a_j < d' <  c \sum \limits_{z_j  \in Z \setminus Z'_{2}} a_j + d \sum \limits_{z_j \in Z'_{2}} a_j.$$

Let $Z'_{3} = Z'_{1} \cap Z'_{2}$ and $Z'_{4} =  Z \setminus (Z'_{1} \cup Z'_{2})$. Then  
$$
\begin{array}{l c l}
& & c \sum \limits_{z_j  \in Z \setminus Z'_{1}} a_j +  d \sum \limits_{z_j \in Z'_{1}} a_j  = c \sum \limits_{z_j  \in Z'_{4}} a_j + c \sum \limits_{z_j  \in Z'_{2} \setminus Z'_{3}} a_j + d \sum \limits_{z_j \in Z'_{3}} a_j + d \sum \limits_{z_j \in Z'_{1} \setminus Z'_{3}} a_j  \\
& < & d' \\
&< & c \sum \limits_{z_j  \in Z \setminus Z'_{2}} a_j + d \sum \limits_{z_j \in Z'_{2}} a_j
= c \sum \limits_{z_j  \in Z'_{4}} a_j + c \sum \limits_{z_j  \in Z'_{1} \setminus Z'_{3}} a_j + d \sum \limits_{z_j \in Z'_{3}} a_j + d \sum \limits_{z_j \in Z'_{2} \setminus Z'_{3}} a_j.
\end{array}
$$
This implies that 
$$c \sum \limits_{z_j  \in Z'_{2} \setminus Z'_{3}} a_j  + d \sum \limits_{z_j \in Z'_{1} \setminus Z'_{3}} a_j < d' -  (c \sum \limits_{z_j  \in Z'_{4}} a_j + d \sum \limits_{z_j \in Z'_{3}} a_j) < d \sum \limits_{z_j \in Z'_{2} \setminus Z'_{3}} a_j + c \sum \limits_{z_j  \in Z'_{1} \setminus Z'_{3}} a_j.$$

Let 
$$d'_1 = \frac{d' -  (c \sum \limits_{z_j  \in Z'_{4}} a_j + d \sum \limits_{z_j \in Z'_{3}} a_j)  - (c + d) \sum \limits_{z_j \in Z'_{1} \setminus Z'_{3}} a_j}{ \sum \limits_{z_j  \in Z'_{2} \setminus Z'_{3}} a_j - \sum \limits_{z_j \in Z'_{1} \setminus Z'_{3}} a_j }$$
and
$$d'_2 =\frac{(c + d) (\sum \limits_{z_j  \in Z'_{2} \setminus Z'_{3}} a_j) -  d' +  (c \sum \limits_{z_j  \in Z'_{4}} a_j + d \sum \limits_{z_j \in Z'_{3}} a_j)}{\sum \limits_{z_j  \in Z'_{2} \setminus Z'_{3}} a_j - \sum \limits_{z_j \in Z'_{1} \setminus Z'_{3}} a_j }.$$

It is a routine to verify that $c < d'_1, d'_2 < d$ and 
$$d'_1 \sum \limits_{z_j  \in Z'_{2} \setminus Z'_{3}} a_j   + d'_2 \sum \limits_{z_j \in Z'_{1} \setminus Z'_{3}} a_j = d' -  (c \sum \limits_{z_j  \in Z'_{4}} a_j + d \sum \limits_{z_j \in Z'_{3}} a_j).$$
Therefore,
$$c \sum \limits_{z_j  \in Z'_{4}} a_j + d \sum \limits_{z_j \in Z'_{3}} a_j+ d'_1 \sum \limits_{z_j  \in Z'_{2} \setminus Z'_{3}} a_j   + d'_2 \sum \limits_{z_j \in Z'_{1} \setminus Z'_{3}} a_j = d'.$$

%$u = (v-(c_1+c_2)b)/(a-b)$ and $u' = ((c_1 + c_2)a - v)/(a-b)$. Then 

%Consider the function $f(x_1, x_2) = x_1 \sum \limits_{z_j  \in Z'_{I, 2} \setminus Z'_{I, 3}} a_j + x_2 \sum \limits_{z_j \in Z'_{I, 1} \setminus Z'_{I, 3}} a_j $ over the rectangle $[c, d] \times [c, d]$. Then $f(c, d) < d -  (c \sum \limits_{z_j  \in Z'_{I, 4}} a_j + d \sum \limits_{z_j \in Z'_{I, 3}} a_j) < f(d, c)$. By the intermediate value theorem,  we know that there is $(u, u') \in [c, d] \times [c, d]$

%Therefore, there are $u, u' \in \ratnum$ such that $c < u < d$, $c < u' < d$, and 
%$$u \sum \limits_{z_j \in Z'_{I, 2} \setminus Z'_{I, 3}} a_j + u' \sum \limits_{z_j \in Z'_{I, 1} \setminus Z'_{I, 3}} a_j = d -  (c \sum \limits_{z_j  \in Z'_{I,4}} a_j + d \sum \limits_{z_j \in Z'_{I, 3}} a_j) ,$$ that is,  
%$$c \sum \limits_{z_j  \in Z'_{I, 4}} a_j + d \sum \limits_{z_j \in Z'_{I, 3}} a_j  + u \sum \limits_{z_j \in Z'_{I, 2} \setminus Z'_{I, 3}} a_j + u' \sum \limits_{z_j \in Z'_{I, 1} \setminus Z'_{I, 3}} a_j  = d.$$

Define an assignment $\eta: Z_I \rightarrow (c, d)$ as follows: 
\begin{itemize}
\item for each $z_j \in Z'_{4}$, let $\eta(z_j) = c + \varepsilon_j$,   
\item for each $z_j \in Z'_{3}$, let $\eta(z_j) = d + \varepsilon_j$, 
\item for each $z_j \in Z'_{2} \setminus Z'_{3}$, let $\eta(z_j) = d'_1 + \varepsilon_j$, 
\item for each $z_j \in Z'_{1} \setminus Z'_{3}$, let $\eta(z_j) = d'_2 + \varepsilon_j$,
\end{itemize}
where
\begin{itemize}
\item for each $z_j \in Z'_{4}$, $0 < \varepsilon_j < \min(d'_1 - c, d'_2 - c)/2$,
\item for each $z_j \in Z'_{3}$, $- \min(d - d'_1, d - d'_2)/2 < \varepsilon_j <0$,
\item for each $z_j \in Z'_{2} \setminus Z'_{3}$, $|\varepsilon_j| < \min(d - d'_1, d - d'_2, d'_1- c, d'_2 - c)/2$,
\item for each $z_j \in Z'_{1} \setminus Z'_{3}$, $|\varepsilon_j| < \min(d - d'_1, d - d'_2, d'_1 - c, d'_2 - c)/2$,
\item for each $z_{j}, z_{j'} \in Z'_{4}$  (resp. $z_{j}, z_{j'} \in Z'_{3}$, $z_{j}, z_{j'} \in Z'_{2} \setminus Z'_{3}$, $z_{j}, z_{j'} \in Z'_{1} \setminus Z'_{3}$) such that $z_{j} \prec z_{j'}$, we have $\varepsilon_{j} < \varepsilon_{j'}$, 
\item $\sum \limits_{i' \in [m]} a_{i'} \varepsilon_{i'} = 0$.
\end{itemize}
It is not hard to see that such rational numbers $\varepsilon_j$ for $z_j \in Z$ exist.

Then $\eta$ is consistent with $\preceq$, and 
$$
\begin{array}{l c l}
& &  \sum \limits_{i' \in [m]} a_{i'} \eta(z_{i'})  = \sum \limits_{ i' \in [m]} (a_ {i'} \varepsilon_{i'}) + c \sum \limits_{z_j \in Z'_{4}} a_j  + d \sum \limits_{z_j \in Z'_{3}} a_j + d'_1 \sum \limits_{z_j \in Z'_{2} \setminus Z'_{3}} a_j + d'_2 \sum \limits_{z_j \in Z'_{1} \setminus Z'_{3}} a_j \\
& = &  c \sum \limits_{z_j \in Z'_{4}} a_j  + d \sum \limits_{z_j \in Z'_{3}} a_j + d'_1 \sum \limits_{z_j \in Z'_{2} \setminus Z'_{3}} a_j + d'_2 \sum \limits_{z_j \in Z'_{1} \setminus Z'_{3}} a_j = d'. 
\end{array}
$$
\end{proof}

Next, we generalize Proposition~\ref{prop-po-one-interval} to the special intervals $(-\infty, c)$ with $c \le 0$ and $(c, +\infty)$ with $c \ge 0$.
For technical reasons, we introduce two special symbols $c_{-\infty}$ and $c_{+\infty}$ to represent the arbitrarily small negative resp.  large positive rational numbers. The two special symbols $c_{-\infty}$ and $c_{+\infty}$ can also be seen as two special variables.
In addition, we will use the expressions of the form $d_0 + d_1  c_{+\infty}$ and $d_0 + d_1  c_{-\infty}$ for $d_0, d_1 \in \ratnum$. We extend the order relation on $\ratnum$ to $\{d_0 + d_1  c_{+\infty} \mid d_0, d_1 \in \ratnum\}$ resp. $\{d_0 + d_1 c_{-\infty}  \mid d_0, d_1 \in \ratnum\}$ as follows: 
\begin{itemize}
\item $d_0 + d_1 c_{+\infty} < d'_0 + d'_1 c_{+\infty}$  iff either $d_1 < d'_1$ or $d_1 = d'_1$ and $d_0 < d'_0$,
\item $d_0 + d_1 c_{-\infty}  < d'_0 + d'_1 c_{-\infty}$ iff either $d_1 > d'_1$ or $d_1= d'_1$ and $d_0 < d'_0$.
\end{itemize}

Let $I = (-\infty, c)$ with $c \le 0$ or $(c, +\infty)$ with $c \ge 0$, we define $\infr_I$, $\supr_I$, $\alpha_I$, and $\beta_I$ as follows.
\begin{itemize}
\item If $I = (-\infty, c)$, then $\infr_I = c_{-\infty}$, $\supr_I = c$, and $\alpha_I, \beta_I$ are defined as follows.
\begin{itemize}
\item Suppose $\min \limits_{Z' \in \pow^\uparrow_{\preceq}(Z)} \left( \infr_I \sum \limits_{z_j  \in Z \setminus Z' } a_j + \supr_I \sum \limits_{z_j \in Z'} a_j \right) = d_0 + d_1 c_{-\infty}$. Then $\alpha_I = -\infty$  if $d_1 > 0$, and $\alpha_I = d_0$ otherwise (note that from the fact $d_0 + d_1 c_{-\infty}$ is the minimum, it holds that $d_0 + d_1 c_{-\infty}  \le c \sum \limits_{z_j \in Z} a_j $, thus $d_1 \ge 0$).
\item Suppose $ \max \limits_{Z' \in \pow^\uparrow_{\preceq}(Z)} \left(\infr_I \sum \limits_{z_j  \in Z \setminus Z' } a_j + \supr_I \sum \limits_{z_j \in Z'} a_j \right) = d'_0 + d'_1 c_{-\infty}$. Then $\beta_I = + \infty$ if $d'_1 < 0$, and $\beta_I = d'_0$ otherwise (note that from the fact $d'_0 + d'_1 c_{-\infty}$ is the maximum, it holds that $c \sum \limits_{z_j \in Z} a_j  \le d'_0 + d'_1 c_{-\infty}$, thus $d'_1 \le 0$).
\end{itemize}
\item If $I = (c, +\infty)$, then $\infr_I = c$, $\supr_I = c_{+\infty}$, and $\alpha_I, \beta_I$ are defined as follows.
\begin{itemize}
\item Suppose $\min \limits_{Z' \in \pow^\uparrow_{\preceq}(Z)} \left( \infr_I \sum \limits_{z_j  \in Z \setminus Z'} a_j + \supr_I \sum \limits_{z_j \in Z'} a_j \right) = d_0 + d_1 c_{+\infty}$. Then $\alpha_I = -\infty$  if $d_1 < 0$, and $\alpha_I = d_0$ otherwise (note that from the fact $d_0 + d_1 c_{+\infty}$ is the minimum, it holds that $d_0 + d_1 c_{+\infty}  \le c \sum \limits_{z_j \in Z} a_j $, thus $d_1 \le 0$).
\item Suppose $ \max \limits_{Z' \in \pow^\uparrow_{\preceq}(Z)} \left(\infr_I \sum \limits_{z_j  \in Z \setminus Z' } a_j + \supr_I \sum \limits_{z_j \in Z'} a_j \right) = d'_0 + d'_1 c_{+\infty}$. Then $\beta_I = + \infty$ if $d'_1 > 0$, and $\beta_I = d'_0$ otherwise (note that from the fact $d'_0 + d'_1 c_{+\infty}$ is the maximum, it holds that $c \sum \limits_{z_j \in Z} a_j  \le d'_0 + d'_1 c_{+\infty}$, thus $d'_1 \ge 0$).
\end{itemize}
\end{itemize}

Note that the minimum and maximum operator in the definition of $\alpha_I$ and $\beta_I$ above are based on the aforementioned order relation on $\{d_0 + d_1  c_{+\infty} \mid d_0, d_1 \in \ratnum\}$ resp. $\{d_0 + d_1 c_{-\infty}  \mid d_0, d_1 \in \ratnum\}$.

\smallskip

\begin{proposition}\label{prop-po-one-interval-gen}
Suppose $I = (-\infty, c)$ with $c \le 0$ or $(c, +\infty)$ with $c \ge 0$, $Z=\{z_1,\dots, z_m\}$ is a set of variables, $\preceq$ is a partial order on $Z$, and $a_1, \dots, a_m \in \ratnum$. Then 
$
\left\{ \sum \limits_{j \in [m]} a_j\eta(z_j) \  \big\vert\  \eta: Z \rightarrow I \mbox{ is consistent with } \preceq \right\}= (\alpha_I, \beta_I).$
\end{proposition}

\begin{proof}
Similar to the proof of Proposition~\ref{prop-po-one-interval}, we still prove the following two facts.
\begin{enumerate}
\item For each assignment $\eta: Z \rightarrow I$ such that $\eta$ is consistent with $\preceq$, there are $Z'_{1}, Z'_{2} \in \pow^\uparrow_{\preceq}(Z)$ such that 
$$\infr_I \sum \limits_{z_j  \in Z \setminus Z'_{1}} a_j + \supr_I \sum \limits_{z_j \in Z'_{1}} a_j  < \sum \limits_{z_j \in Z} a_j \eta(z_j) <  \infr_I \sum \limits_{z_j \in Z \setminus  Z'_{2}} a_j + \supr_I \sum \limits_{z_j \in Z'_{2}} a_j.$$

\item For each $d' \in (\alpha_I, \beta_I)$, 
there is an assignment $\eta: Z \rightarrow I$ such that $\eta$ is consistent with $\preceq$ and $d' = \sum \limits_{z_j \in Z} a_j  \eta(z_j)$.
\end{enumerate}

The proof of the first fact is essentially the same as in Proposition~\ref{prop-po-one-interval}. 

For the second fact, 
we will illustrate the arguments for  $I = (c, +\infty)$. The arguments for  $I = (-\infty, c)$ are symmetric.

\begin{claim}
 Let $I = (c, +\infty)$ with $c \ge 0$. Then one of the following holds: $\alpha_I = -\infty$ and $\beta_I \in \ratnum$, or $\alpha_I \in \ratnum$ and $\beta_I = +\infty$, or $\alpha_I = -\infty$ and $\beta_I = +\infty$. 
\end{claim}

\begin{proof} [Proof of the claim] 
Let $Z' = \{z_{i}\}$ be the upward-closed subset of $Z$ such that $z_{i}$ is a $\preceq$-maximal element of $Z$. Then $\infr_I \sum \limits_{z_j \in Z \setminus Z' } a_j + \supr_I  \sum \limits_{z_j \in Z'} a_j = c  \sum \limits_{z_j \in Z \setminus Z' } a_j + a_{i} c_{+\infty}$. If $a_{i} > 0$, then $\beta_I = +\infty$, otherwise, $\alpha_I = -\infty$. Therefore, we always have $\beta_I = +\infty$ or $\alpha_I = -\infty$. From the definition of $\alpha_I$ and $\beta_I$, we conclude that the claim holds.
\end{proof}

\medskip

We resume the proof of the proposition and distinguish between the three cases in the claim.

Suppose $d' \in (\alpha_I, \beta_I)$. 
Our goal is to construct an assignment $\eta: Z \rightarrow I$ such that $\eta$ is consistent with $\preceq$ and $d' = \sum \limits_{z_j \in Z} a_j \eta(z_j)$. 

\medskip

\noindent {\it Case $\alpha_I = -\infty$ and $\beta_I \in \ratnum$}.

\smallskip

From the definition of $\alpha_I, \beta_I$ and the fact that $\alpha_I = -\infty$ and $\beta_I \in \ratnum$, we know that there are $Z'_{1}, Z'_{2} \in \pow^\uparrow_{\preceq}(Z)$ such that $\sum \limits_{z_j \in Z'_{1}} a_j < 0$, $\sum \limits_{z_j \in Z'_{2}} a_j = 0$, and 
%$c \sum \limits_{z_j  \in Z_I \setminus Z'_{I, 1}} a_j + d \sum \limits_{z_j \in Z'_{I, 1}} a_j$ satisfies that and 
%
$\beta_I =  c \sum \limits_{z_j  \in Z \setminus Z'_{2}} a_j $.  
This implies that there is a sufficiently large number $D \in (c, +\infty)$ such that  
$c \sum \limits_{z_j  \in Z \setminus Z'_{1}} a_j + D \sum \limits_{z_j \in Z'_{1}} a_j < d'.$
Therefore, we have 
$$c \sum \limits_{z_j  \in Z \setminus Z'_{1}} a_j + D \sum \limits_{z_j \in Z'_{1}} a_j < d' < c \sum \limits_{z_j  \in Z \setminus Z'_{2}} a_j + D \sum \limits_{z_j \in Z'_{2}} a_j.$$

Let $Z'_{3} = Z'_{1} \cap Z'_{2}$ and $Z'_{4} =  Z \setminus (Z'_{1} \cup Z'_{2})$. The rest of the arguments are the same as those for the second fact in the proof of Proposition~\ref{prop-po-one-interval}, with  $d$ replaced by $D$. 

\medskip

\noindent {\it Case $\alpha_I \in \ratnum$ and $\beta_I  = +\infty$}.

\smallskip

%$\sum \limits_{i \in [t]} a_{j_i}  \eta(z_{j_i})$.

From the definition of $\alpha_I, \beta_I$ and the fact that $\alpha_I \in \ratnum$ and $\beta_I = +\infty$, we know that there are $Z'_{1}, Z'_{2} \in \pow^\uparrow_{\preceq}(Z)$ such that $\sum \limits_{z_j \in Z'_{1}} a_j = 0$, $\alpha_I =  c \sum \limits_{z_j  \in Z \setminus Z'_{1}} a_j $, and $\sum \limits_{z_j \in Z'_{2}} a_j > 0$.  
This implies that there is a sufficiently large number $D \in (c, +\infty)$ such that  
$d' < c \sum \limits_{z_j  \in Z \setminus Z'_{2}} a_j + D \sum \limits_{z_j \in Z'_{2}} a_j.$
Therefore, we have 
$$c \sum \limits_{z_j  \in Z \setminus Z'_{1}} a_j + D \sum \limits_{z_j \in Z'_{1}} a_j < d' < c \sum \limits_{z_j  \in Z \setminus Z'_{2}} a_j + D \sum \limits_{z_j \in Z'_{2}} a_j.$$

Let $Z'_{3} = Z'_{1} \cap Z'_{2}$ and $Z'_{4} =  Z \setminus (Z'_{1} \cup Z'_{2})$. The rest of the arguments are the same as those for the second fact in the proof of Proposition~\ref{prop-po-one-interval}, with  $d$ replaced by  $D$. 

\medskip

\noindent {\it Case $\alpha_I = - \infty$ and $\beta_I  = +\infty$}.

\smallskip

From the definition of $\alpha_I, \beta_I$ and the fact that $\alpha_I = -\infty$ and $\beta_I = +\infty$, we know that there are $Z'_{1}, Z'_{2} \in \pow^\uparrow_{\preceq}(Z)$ such that $\sum \limits_{z_j \in Z'_{1}} a_j < 0$ and $\sum \limits_{z_j \in Z'_{2}} a_j > 0$.  
This implies that there is a sufficiently large number $D \in (c, +\infty)$ such that  
$$c \sum \limits_{z_j  \in Z \setminus Z'_{1}} a_j + D \sum \limits_{z_j \in Z'_{1}} a_j < d' < c \sum \limits_{z_j  \in Z \setminus Z'_{2}} a_j + D \sum \limits_{z_j \in Z'_{2}} a_j.$$

Let $Z'_{3} = Z'_{1} \cap Z'_{2}$ and $Z'_{4} =  Z \setminus (Z'_{1} \cup Z'_{2})$. The rest of the arguments are the same as those for the second fact in the proof of Proposition~\ref{prop-po-one-interval}, with $d$ replaced by $D$. 
\end{proof}

The results in Proposition~\ref{prop-po-one-interval}-\ref{prop-po-one-interval-gen} can be generalized naturally to multiple intervals as follows. 

In the rest of this subsection, we fix 
\begin{itemize}
\item a finite set of open intervals $\intval$ which comprises $(-\infty, c_{-r})$, $(c_s, -\infty)$, and the intervals $(c_i, c_{i+1})$ with $i \in [-r, s-1]$ such that $c_0=0$ and $c_{-r} < \dots < c_{-1} < c_0 < c_1 < \dots < c_s$ (let $C$ denote $\{c_{-r}, \dots, c_{-1}, c_0, c_1, c_s\}$),
\item a finite set of variables $Z=\{z_1,\dots, z_m\}$, 
\item a partial order $\preceq$ on $Z \cup C$ such that the restriction of $\preceq$ to $C$ is identical to the restriction of the order relation of $\ratnum$ to $C$, in addition, for each $z_j \in Z$, either $z_j \preceq c_{-r}$, or $c_s \preceq z_j$, or $c_i \preceq z_j \preceq c_{i+1}$ for some $i \in [-r, s-1]$.
\item and $a_1, \dots, a_m \in \ratnum$.
\end{itemize}

An assignment $\eta: Z \rightarrow \ratnum$ is said to be \emph{consistent with $\preceq$} if for each $z_i, z_j \in Z$, $z_j \preceq z_j$ implies $\eta(z_i) \le \eta(z_j)$, and for each $z_i \in Z$ and $c_j \in C$, $z_i \preceq c_j$ implies $\eta(z_i) \le c_j$. For each $I \in \intval$, define $\infr_I$, $\supr_I$, $Z_I$, $\alpha_I$, and $\beta_I$ as follows:
\begin{itemize}
\item if $I = (-\infty, c_{-r})$, then $\infr_I = c_{-\infty}$, $\supr_I = c_{-r}$, $Z_I = \{z_i \in Z \mid z_i \preceq c_{-r}\}$, $\alpha_I$ and $\beta_I$ are those in Proposition~\ref{prop-po-one-interval-gen} (with $c$ replaced by $c_{-r}$ and $Z$ replaced by $Z_I$),
\item if $I = (c_s, +\infty)$, then $\infr_I = c_s$, $\supr_I = c_{+\infty}$, $Z_I = \{z_i \in Z \mid c_s \preceq z_i \}$, $\alpha_I$ and $\beta_I$ are those in Proposition~\ref{prop-po-one-interval-gen} (with $c$ replaced by $c_{s}$ and $Z$ replaced by $Z_I$),
\item if $I = (c_j, c_{j+1})$ for $j \in [-r, s-1]$, then $\infr_I = c_j$, $\supr_I = c_{j+1}$, $Z_I = \{z_i \in Z \mid c_j \preceq z_i \preceq c_{j+1}\}$, and $\alpha_I$ and $\beta_I$ are $\alpha, \beta$ in Proposition~\ref{prop-po-one-interval-gen} (with $c,d$ replaced by $c_j, c_{j+1}$ and $Z$ replaced by $Z_I$).
\end{itemize}

\begin{corollary}\label{cor-po-mult-intval}
The following equation holds. 
$$\left\{ \sum \limits_{z_j \in Z} a_j \eta(z_j) \  \big\vert\   \eta: Z \rightarrow \ratnum \mbox{ is consistent with } \preceq \right\}= \left(\sum \limits_{I \in \intval} \alpha_I, \sum \limits_{I \in \intval}  \beta_I \right).$$
\end{corollary}
\begin{remark}
Note that  to define the meanings of the expressions $\sum \limits_{I \in \intval} \alpha_I$ and $\sum \limits_{I \in \intval}  \beta_I$ in Corollary~\ref{cor-po-mult-intval}, the arithmetics of $\ratnum$ should be extended to $\ratnum \cup \{-\infty, +\infty\}$ as follows: $-\infty + c = -\infty$ for $c \in \ratnum$ and $+\infty + c = +\infty$ (The expressions like $+\infty +(-\infty)$ will not occur in $\sum \limits_{I \in \intval} \alpha_I$ and $\sum \limits_{I \in \intval}  \beta_I$). 
\end{remark}

From Corollary~\ref{cor-po-mult-intval}, we can also deduce directly the following result.

\begin{corollary}\label{cor-po-mult-intval-ls-gr-0}
The following two conditions are equivalent,
\begin{itemize}
\item  there is an assignment $\eta: Z \rightarrow \ratnum$ such that $\eta$ is consistent with $\preceq$ and $a_0 + \sum \limits_{i \in [m]} a_i \eta(z_i) = 0$,
\item for each $I \in \intval$, there are two upward-closed subsets $Z'_I, Z''_I$ of $Z_I$, such that 
\begin{eqnarray}\label{eqn-po-ls-0-lem}
a_0 +  \sum \limits_{I \in \intval} \left(d \sum \limits_{z_i \in Z'_I} a_i +  c \sum \limits_{z_i \in Z_I \setminus Z'_I} a_i \right) < 0
\end{eqnarray}
 and 
\begin{eqnarray}\label{eqn-po-gr-0-lem}
a_0 + \sum \limits_{I \in \intval} \left(d \sum \limits_{z_i \in Z''_I} a_i + c \sum \limits_{z_i \in Z_I \setminus Z''_I} a_i \right) > 0.
\end{eqnarray}
\end{itemize}
\end{corollary}

\begin{remark}\label{rem-po}
Note that the expressions in the left-hand-side of the inequality (\ref{eqn-po-ls-0-lem}) and (\ref{eqn-po-gr-0-lem}) can be rewritten into the expressions of the form $d_0 + d_1 c_{-\infty} + d_2 c_{+\infty}$, where $d_0, d_1, d_2 \in \ratnum$.  The meanings of the inequalities $d_0 + d_1  c_{-\infty} + d_2  c_{+\infty} < 0$ and $d_0 + d_1  c_{-\infty} + d_2  c_{+\infty} > 0$ are defined as follows. 
\begin{itemize}
\item $d_0 + d_1 c_{-\infty} + d_2 c_{+\infty} < 0$ iff one of the following three conditions holds: 
1) $d_1 > 0$, 2) $d_2 < 0$, 3) $d_1 \le 0$, $d_2 \ge 0$, and $d_0 + d_1c_{-r}   + d_2c_s   < 0$. Intuitively, the first condition means that if $c_{-\infty}$  is replaced by a sufficiently small negative rational number, then the value of $d_0 + d_1 c_{-\infty} + d_2 c_{+\infty}$ can be made negative, similarly for the second condition. The third condition means that if $c_{-\infty}$ is replaced by a rational number which is slightly smaller than $c_{-r}$, and $c_{+\infty}$ is replaced by a rational number which is slightly greater than $c_s$, then the value of $d_0 + d_1 c_{-\infty} + d_2 c_{+\infty}$ can be made negative. 
\item Symmetrically, $d_0 + d_1 c_{-\infty} + d_2 c_{+\infty} > 0$ iff one of the following three conditions holds: 
1) $d_1 < 0$, 2) $d_2 > 0$, 3) $d_1 \ge 0$, $d_2 \le 0$, and $d_0 + d_1c_{-r}   + d_2c_s   > 0$.
\end{itemize}
\end{remark}

%\begin{proof}[Proof of Lemma~\ref{lem-po}]
%From Corollary~\ref{cor-po}, we deduce that there is an assignment $\eta: Z \rightarrow \ratnum$ such that $\eta$ is consistent with $\preceq_P$ and $a_0 + a_1 \eta(z_1) + \dots + a_m \eta(z_m) = 0$ iff $0 \in  \left(a_0 +\sum \limits_{I \in \intval, Z_I \neq \emptyset} \alpha_I,\ a_0+ \sum \limits_{I \in \intval, Z_I \neq \emptyset}  \beta_I \right)$. 
%
%From the definition of $\alpha_I, \beta_I$ for $I \in \intval$, we know that $0 \in  \left(a_0 +\sum \limits_{I \in \intval, Z_I \neq \emptyset} \alpha_I,\ a_0+ \sum \limits_{I \in \intval, Z_I \neq \emptyset}  \beta_I \right)$ iff for each $I \in \intval$ such that $Z_I \neq \emptyset$,  there are upward-closed subsets $Z'_{I}, Z''_{I}$ of $Z_I$, satisfying that 
%
%$$a_0 +  \sum \limits_{I \in \intval} \big(d \sum \limits_{z_i \in Z'_I} a_i +  c \sum \limits_{z_i \in Z_I \setminus Z'_I} a_i \big) < 0 \mbox{ and } 
%a_0 + \sum \limits_{I \in \intval} \big(d \sum \limits_{z_i \in Z''_I} a_i + c \sum \limits_{z_i \in Z_I \setminus Z''_I} a_i \big) > 0.$$
%The proof of the lemma is complete.
%\end{proof}

For briefness, we will use the expression (\ref{eqn-po-ls-0-lem}) resp. (\ref{eqn-po-gr-0-lem}) to denote the expression in the left-hand-side of the inequality (\ref{eqn-po-ls-0-lem}) resp. (\ref{eqn-po-gr-0-lem}).
}
%%%%%%%%%%%%%%%%%%%%%%%%%%%%%%%%%%%%%%%%%%%%%%%%%%%%%%%%
%%%%%%%%%%%%%%%%%%%%%%%%%%%%%%%%%%%%%%%%%%%%%%%%%%%%%%%%

\subsection{Proof of Theorem~\ref{thm-reach-dec}: The decision procedure} \label{app-pf-thm-reach-dec}

\newcommand{\cnst}{{\sf cnst}}

\newcommand{\pmt}{{\sf pmt}}

\newcommand{\row}{{\sf row}}

\newcommand{\storedIn}{{\sf storedIn}}

In the following, we present a proof of Theorem~\ref{thm-reach-dec}.

Let $\cA=\langle Q,q_0,F,\vu_0,\delta,\zeta\rangle$
be a copyless $\raq$ with non-strict transition guards over $(X,Y)$, where $X=\{x_1,\ldots,x_k\}$
and $Y=\{y_1,\ldots,y_l\}$.
Let $\cN$ be the set of constants found in $\vu_0\ssX$. W.l.o.g., we assume that $0 \in \cN$, that is, the constant $0$ is in the initial contents control variables. In addition, let $c_{\min}$ and $c_{\max}$ be the minimum and maximum constant in $\cN$ and $\cN_\infty = \{-\infty, + \infty\} \cup \cN$.

To simplify the presentation, we first apply a normalization procedure to $\cA$. By abuse the notation, we still use $\cA$ to denote the normalized $\raq$. The normalized $\raq$ $\cA$ satisfies the following properties.
\begin{enumerate}
\item The set of control variables $X$ is partitioned into $X_1, X_2$ such that $X_1$ is the set of read-only control variables
 and $X_1$ holds all the constants in $\cN$, more specifically, there is a bijection $\cnst$ between $X_1$ and $\cN$  (intuitively, each variable $x \in X_1$ stores the constant $\cnst(x)$).
\item For each state $q \in Q$, a total preorder over $X$, denoted by $\preceq_q$, is associated with $q$, which is consistent with the rational order relation on $\cN$, that is, for every $x, x' \in X_1$, $x \preceq_q x'$ iff $\cnst(x) \le \cnst(x')$. We will use $\simeq_q$ to denote the equivalence relation induced by $\preceq_q$, that is, $x \simeq_q x'$ iff $x \preceq_q x'$ and $x' \preceq_q x$. In addition, let $\prec_q = \preceq_q \setminus \simeq_q$. Note that the total preorders $\preceq_q$ for $q \in Q$ can be seen as a reformulation of the orderings used in Section~\ref{sec:non-zero}.
\item For each transition $(p, \varphi(\vec{x}, \cur)) \rightarrow (q, A, B, \vec{b})$, let $[z'_1]_p, \dots, [z'_{k'}]_p$ be an enumeration of the equivalence classes of $\simeq_p$ such that $z'_1 \prec_p \dots \prec_p z'_{k'}$, then the guard $\varphi(\vec{x}, \cur)$ is of the form $\cur = z'_i$ with $i \in [k']$, or $z'_i \le \cur \le z'_{i+1}$ with $i \in [k'-1]$, or $\cur \le z'_1$, or $z'_{k'} \le \cur$. 
%The guards of the form $\cur = z'_i$ are called \emph{rigid} guards, and the guards of the other forms are called \emph{flexible} guards.
%
% In addition, if $\varphi(\vec{x}, \cur)$ is $\cur = z'_i$ with $i \in [t]$, then $A(j, k+1) = 0$ for each $j \in [k]$, and $B(j, k+l+1) = 0$ for each $j \in [l]$ (intuitively, if the value of $\cur$ is not fresh, then it is replaced by a corresponding control variable in the expressions to update the control or data variables).
\end{enumerate}
For an $\raq$ $\cA$ of $n$ states and $k$ control variables, $k$ additional read-only control variables may be introduced to store the constants, then the $\raq$ after normalization has at most $O(n 2^{2k-1} (2k)! )$ states, where $2^{2k-1} (2k)!$ is the number of orderings for $2k$ control variables (cf. Section~\ref{sec:non-zero}), which is an upper bound on the number of total preorders for $2k$ control variables.

\smallskip

From now on, we assume that $\cA$ is normalized.

Suppose there is a word $w=d_1\cdots d_n$ that leads to $0$.
Let the run be
$(q_0,\vu_0)\vdash_{t_1,d_1}
(q_1,\vu_1)\vdash_{t_2,d_2}\cdots \vdash_{t_n,d_n} (q_n,\vu_n)$.
By Proposition~\ref{prop:linear-raq}, there are $M$ and $\vb$ such that
\begin{eqnarray*}
\vu_n & = & M \myvec {d_1\\ \vdots \\ d_n} + \vb.
\end{eqnarray*}
Now, these values $d_1,\ldots,d_n$ satisfies a set of inequalities
imposed by the transitions $t_1,\ldots,t_n$. Let $\Phi(d_1,\ldots,d_n)$ denote the conjunction of those inequalities. 
Note that due to the fact that $\cA$ is normalized, the guards in $t_1,\dots, t_n$ contain \emph{no disjunctions}, which means that the set of points (vectors) satisfying $\Phi(\vz)$ is a convex polyhedron.

Suppose the output function of $q_n$ be $\va\cdot \myvec{\vx\\ \vy} + a'$.
Define the following function:
\begin{eqnarray*}
f(\vz) & = & \va \cdot M \myvec {z_1\\ \vdots \\ z_n} +\va\cdot \vb +a'.
\end{eqnarray*}
Thus, by our assumption that $d_1\cdots d_n$ leads to zero,
we have:
\begin{eqnarray*}
f((d_1,\ldots,d_n)^t) = 0& \wedge & \Phi((d_1,\ldots,d_n)^t)=\ltrue.
\end{eqnarray*}
It follows that 
\begin{eqnarray}\label{eq:reach_bound-2}
	\exists\vz_1,\vz_2 \in \bbQ^{n}: f(\vz_1)  \leq   0  \leq  f(\vz_2) \wedge  \Phi(\vz_1)\wedge\Phi(\vz_2).
\end{eqnarray}
Observe that (\ref{eq:reach_bound-2}) holds iff the following two constraints hold simultaneously:
\begin{description}\itemsep=0pt
\item[{\bf [F1]}]
the infimum of $f(\vz)$ w.r.t. $\Phi(\vz)$ is $\leq 0$,
\item[{\bf [F2]}]
the supremum of $f(\vz)$ w.r.t. $\Phi(\vz)$ is $\geq 0$.
\end{description}
By the Simplex algorithm for linear programming~\cite{chvatal},
we know that the points that yield the optimum, i.e., the infimum and the supremum,
are at the ``corner'' points of convex polyhedra.
The constraints in $\Phi(\vz)$ contain the constants from $\cN$ (as a result of the fact that the initial contents of control variables are a fixed vector of constants),
so the corner points of the convex polyhedron represented by $\Phi(\vz)$ have components from $\cN_\infty$.

To establish F1 and F2, it is sufficient to find 
two corner points $\vz_1$ and $\vz_2$ such that:
$$
f(\vz_1)\leq 0 \leq f(\vz_2)
$$
To find these two points, we will construct
a $\bbQ$-VASS $\cB$, whose reachability is decidable in $\nptime$. (See 
Appendix~\ref{app-rat-vass}.)

Let $q \in Q$. A {\em specification} of $X$ w.r.t. $q$ is a mapping $\eta$ from $X$ to $\cN_{\infty}$
that respects $\preceq_q$,
i.e., 1) for each $x_i \in X_1$, $\eta(x_i) = \cnst(x_i)$, and 2) for each $x_i, x_j \in X$, if $x_i \preceq_q x_j$, then $\eta(x_i) \le \eta(x_j)$.
Intuitively, $\eta$ encodes the value of $x_i$ of a corner point, and $\eta(x_i)=c \in \cN$ means that $x_i$ is assigned with $c$, 
%a value arbitrarily close to $c_j$ from above,
%whereas $\eta(x_i)=c_j^=$ means that $x_i$
%is assigned with $c_j$.

For $m,n \in \bbN$, we will use $[m]$ to denote $\{1,\dots, m\}$, and $[m, n]$ to denote $\{m, \dots, n\}$, provided that $m \le n$. 

We will construct a two-dimensional $\bbQ$-VASS $\cB=(S,\Delta)$
with two rational variables $C_1, C_2$. 
%$\vy_1 = (y_{1,1},\ldots,y_{1,l})$ 
%and $\vy_2 = (y_{2,1},\ldots,y_{2,l})$ as follows.
The set $S$ of states comprises two special states $q'_0, q'_f$ (whose purpose will become clear later) and the set of tuples $(q, q_f, \eta_1, \eta_2, \pi)$ such that 
\begin{itemize}
\item $q \in Q$, $q_f \in F$, 
\item $\eta_1, \eta_2$ are the specifications of $X$ w.r.t. $\preceq_q$,
\item $\pi$ is a mapping from $[l]$ to $[l] \cup \{0\}$.  
\end{itemize}
Intuitively, 
\begin{itemize}
\item $q$ represents the current state of $\cA$,
\item $q_f$ represents the final state of $\cA$ that we will reach eventually, 
\item $\eta_1$ and $\eta_2$ represent the two points $\vz_1$ and $\vz_2$ we are looking for,
\item for each $i \in [l]$, if $\pi(i) \in [l]$, then the current value of $y_i$ will be stored in $y_{\pi(i)}$ when reaching the state $q_f$, otherwise, the current value of $y_i$ will be lost eventually,
\item $C_1, C_2$ represent the values of $\zeta(q_f)(\vz_1)$ and $\zeta(q_f)(\vz_2)$ respectively.
\end{itemize}
%$Q\times \{(\eta_1,\eta_2) \mid \eta_1,\eta_2 \ \textrm{are specifications}\}$.
%A configuration is of the form $(q,\eta_1,\eta_2,\vy_1,\vy_2)$,
%where in fact $(\eta_1,y_1)$ and $(\eta_2,\vy_2)$
%represent the two points that we are looking for.

Before defining $\Delta$, we introduce some additional notations.

Suppose $t=(p,\varphi)\to(q,A,B,\vb)$ in $\delta$ is a transition of $\cA$ and  $(p, q_f, \eta_1,\eta_2, \pi)\in S$.

\medskip

\paragraph*{The mapping $\storedIn_t$} Let $B = \myvec{B_X\ B_Y\ B_\cur}$ such that $B_X \in \bbQ^{l \times k}$, $B_Y \in \bbQ^{l \times l}$ and $B_\cur \in \bbQ^{l \times 1}$. Intuitively, $B_X, B_Y, B_\cur$ are divided according to control variables, data variables, and $\cur$ respectively. In addition, for $j \in [l]$, we will use $\row_j(B_Y)$ to denote the $j$-th row of $B_Y$.
We also define a mapping $\storedIn_t: [l] \rightarrow [l] \cup \{0\}$ to record where the original value of each data variable is stored after executing $t$, as follows: For each $i \in [l]$,
\begin{itemize}
\item if there is $i' \in [l]$ such that $(\row_{i'}(B_Y))(i) = 1$ (intuitively, the original value of $y_i$ is stored in $y_{i'}$ after executing $t$), then $\storedIn_t(i) = i'$, 
\item otherwise, $\storedIn_t(i) = 0$. 
\end{itemize}
Note that because of the copyless constraint, for each $i \in [l]$, there is at most one $i' \in [l]$ such that $(\row_{i'}(B_Y))(i)= 1$, thus the mapping $\storedIn_t$ is well-defined. 

%Before presenting the transitions of $\Delta$ corresponding to $t$ and $(p, q_f, \eta_1,\eta_2, \pi)$, we introduce some additional notations.

\medskip

\paragraph*{$\pi$ is compatible with $t$} We say that $\pi$ are \emph{compatible} with $t$, if the following constraints are satisfied.
\begin{itemize}
\item For each $i \in [l]$ such that $\pi(i) \neq 0$, it holds that $\storedIn_t(i) \neq 0$. \\
Intuitively, if the current value of $y_i$ will be stored in some data variable eventually, then the current value of $y_i$ should not be lost after executing $t$.
\item For each $i_1, i_2 \in [l]$ such that $\storedIn_t(i_1) = \storedIn_t(i_2) \neq 0$, it holds that $\pi(i_1) = \pi(i_2)$. \\
Intuitively, if the current value of $y_{i_1}$ and $y_{i_2}$ will be stored in the same data variable after executing $t$, then eventually, either both of them will be lost, or otherwise they will be stored in the same data variable. 
\end{itemize}

% and
%that $\vf(\vx,\cur)=(f_1(\vx,\cur),\ldots,f_l(\vx,\cur))$.

\medskip

\paragraph*{$\pi'$ is consistent with the application of $t$ on $\pi$} We say that a mapping $\pi': [l] \rightarrow [l] \cup \{0\}$ is \emph{consistent with the application of $t$ on $\pi$} if for each $i \in [t]$
%\begin{itemize}
%\item 
such that $\storedIn_t(i')=i$ for some $i' \in [l]$, we have $\pi'(i)=\pi(i')$ (intuitively, if the original value of $y_{i'}$ is stored in $y_i$ after executing $t$ and the original value of $y_{i'}$ will be stored in $y_{\pi(i')}$ eventually, then the value of $y_i$ after executing $t$ should be stored in $y_{\pi(i')}$ as well, thus $\pi'(i)$ takes the value of $\pi(i')$),
%
%\item otherwise,  $\cB$ nondeterministically assigns to $\pi(i)$ a value from  $[l] \cup \{0\}$ (in order to guess where the value assigned to $y_i$ in $t$ will be stored eventually). 
%\end{itemize}

\smallskip

For each transition $t=(p,\varphi)\to(q,A,B,\vb)$ in $\delta$ and each $(p, q_f, \eta_1,\eta_2, \pi)\in S$ such that \emph{$\pi$ is compatible with $t$}, we will define a set of transitions in $\Delta$.

In the following, let $\zeta(q_f) = \va\cdot \myvec{\vx\\ \vy} + a'$, and the reassignments of the data variables in $t$ be 
$y_j:= (\row_j(B_Y))^t \cdot \vy + f_j(\vx,\cur)$. In addition, let $\vf(\vx,\cur)=(f_1(\vx,\cur),\ldots,f_l(\vx,\cur))$.
 
Let $[z'_1], \dots, [z'_{k'}]$ be an enumeration of the equivalence classes of $\simeq_p$ on $X$ such that $z'_1 \prec_p \dots \prec_p z'_{k'}$. Then $\varphi(\vx,\cur)$ is of one of the following forms, $\cur = z'_i$ for $i \in [k']$, $\cur \le z'_1$, $z'_i \le \cur \le z'_{i+1}$ for $i \in [k'-1]$, and $z'_{k'} \le \cur$. 

% the updates by considering the typical situation 
%that $\varphi(\vec{x},\cur)$ is of the form $z'_i < \cur < z'_{i+1}$ for some $i \in [t-1]$. For the other situations, the updates are similar.  In addition, since the updates of $idx''_I$ and $\overrightarrow{C''}$ are analogous to those of $idx'_I$ and $\overrightarrow{C'}$, we will only illustrate how to update $idx'_I$ and $\overrightarrow{C'}$.

%Note that since $A \in \bbP^{k\times (k+1)}$ simply selects for each $x_i \in X$, an element from $X \cup \{\cur\}$,
%we can infer a set of specifications that are consistent
%with the application of $A$ on $\eta_1$. 
We make the following conventions below.
\begin{itemize}
\item Let $\eta'_1$ denote a specification that is consistent with
the application of $A$ on $\eta_1$, where a specification $\eta'_1$ is \emph{consistent with the application of $A$ on $\eta_1$} if for each $j, j' \in [k]$ such that $A(j, j') = 1$, we have $\eta'_1(x_j) = \eta_1(x_{j'})$. Note that this definition does not take the assignments of $\cur$ to control variables into consideration.
\item Similarly, let $\eta'_2$ denote a specification that is consistent with the application of $A$ on $\eta_2$. 
\item In addition, let $\pi': [l] \rightarrow [l] \cup \{0\}$ denote a mapping that is consistent with the application of $t$ on $\pi$.
\end{itemize}

We will define the transitions of $\Delta$ according to the different forms of $\varphi$.

\begin{enumerate}\itemsep=1pt
\item[{\bf [T1]}]
If $\varphi$ is of the form $\cur=z'_i$ for $i \in [k']$, then
\begin{eqnarray*}
((p, q_f, \eta_1,\eta_2, \pi), (c'_{1}, c'_{2}),  (q, q_f, \eta'_{1},\eta'_{2}, \pi'))
& \in & \Delta
\end{eqnarray*}
where $c'_1 = \sum \limits_{\pi'(j) \neq 0} \va(\pi'(j)) \cdot c_{1,j}$ and $c'_2 = \sum \limits_{\pi'(j) \neq 0} \va(\pi'(j)) \cdot  c_{2,j}$ such that for each $j=1,\ldots,l$:
\begin{eqnarray*}
c_{1,j} & = & f_j(\eta_1(\vx),\eta_1(z'_i)), 
\\
c_{2,j} & = & f_j(\eta_2(\vx),\eta_2(z'_i)). 
\end{eqnarray*}
%Here the image of $\eta_1,\eta_2$ are all
%integers where each $c_i^{-},c_i^=, c_i^+$
%are treated as $c_i$.
Intuitively, $c_{1,j}$ represents the value added to $y_j$ by $t$, in addition, if $\pi'(j) \neq 0$, then this value will be stored in $y_{\pi'(j)}$ eventually, thus $\va(\pi'(j)) \cdot c_{1,j}$ will be added to the final output. Similarly for $c_{2,j}$.

\smallskip

\item[{\bf [T2]}]
If $\varphi$ is of the form $\cur \le z'_1$,
then the following transitions are in $\Delta$:
\begin{itemize}\itemsep=0pt
\item
$((p, q_f, \eta_1,\eta_2, \pi),  (c'_{1,low}, c'_{2,low}),  (q, q_f, \eta'_{1},\eta'_{2}, \pi'))$, such that for each $j \in [k]$ with $A(j, k+1) = 1$, we have $\eta'_1(x_j)=-\infty$ and $\eta'_2(x_j)=-\infty$,
\item
$((q,\eta_1,\eta_2),  (c'_{1,low}, c'_{2,up}), (q, q_f, \eta'_{1},\eta'_{2}, \pi'))$, such that for each $j \in [k]$ with $A(j, k+1) = 1$, we have $\eta'_1(x_j)=-\infty$ and $\eta'_2(x_j)=\eta_2(z'_1)$,
\item
$((q,\eta_1,\eta_2),  (c'_{1,up}, c'_{2,low}),  (q, q_f, \eta'_{1},\eta'_{2}, \pi'))$, such that for each $j \in [k]$ with $A(j, k+1) = 1$, we have $\eta'_1(x_j)=\eta_1(z'_1)$ and $\eta'_2(x_j)=-\infty$,
\item
$((q,\eta_1,\eta_2),  (c'_{1,up}, c'_{2,up}),  (q, q_f, \eta'_{1},\eta'_{2}, \pi'))$, such that for each $j \in [k]$ with $A(j, k+1) = 1$, we have $\eta'_1(x_j)=\eta_1(z'_1)$ and $\eta'_2(x_j)=\eta_2(z'_1)$,
\end{itemize}
where
\begin{itemize}
\item  $c'_{1, low} = \sum \limits_{\pi'(j) \neq 0} \va(\pi'(j)) \cdot \vc_{1,low}(j)$, and $c'_{2, low} = \sum \limits_{\pi'(j) \neq 0} \va(\pi'(j)) \cdot  \vc_{2,low}(j)$, 
\item $c'_{1, up} = \sum \limits_{\pi'(j) \neq 0} \va(\pi'(j)) \cdot \vc_{1,up}(j)$, and $c'_{2, up} = \sum \limits_{\pi'(j) \neq 0} \va(\pi'(j)) \cdot  \vc_{2,up}(j)$,
\item $\vc_{1,low}$ and $\vc_{2,low}$
denote the vectors $\vf(\vx,\cur)$
with $(\vx,\cur)$ being substituted with $(\eta_1(\vx),-\infty)$
and $(\eta_2(\vx),-\infty)$ respectively,
\item $\vc_{1,up}$ and $\vc_{2,up}$
denote the vectors $\vf(\vx,\cur)$
with $(\vx,\cur)$ being substituted with $(\eta_1(\vx),\eta_1(z'_1))$
and $(\eta_2(\vx),\eta_2(z'_1))$ respectively.
\end{itemize}

Above, when we substitute $\cur$ with $-\infty$
inside $\vf(\vx,\cur)$,
we do not evaluate the $\pm\infty$ and take them as two special variables.
For example, we may have expression $3x_1 -2 \cur$.
Under the substitution $(x_1,\cur)\mapsto (-\infty,-\infty)$,
we obtain a $3(-\infty) -2(-\infty)=(3-2)(-\infty)=-\infty$.
Intuitively, $+\infty$ and $-\infty$ are to be interpreted
as arbitrary rational numbers bigger and smaller than
$c_{\max}$ and $c_{\min}$ respectively.  In Appendix~\ref{app-rat-vass}, we will show that the reachability problem for $\bbQ$-VASS that contain special symbols $\pm \infty$ can also be reduced to the satisfiability of existential Presburger formula, similar to the normal $\bbQ$-VASS (without $\pm \infty$).

\item[{\bf [T3]}] 
If $\varphi$ is of the form $z'_i \le \cur \le z'_{i+1}$ for some $i \in [k'-1]$,
then the following transitions are in $\Delta$:
\begin{itemize}\itemsep=0pt
\item
$((p, q_f, \eta_1,\eta_2, \pi),  (c'_{1,low}, c'_{2,low}),  (q, q_f, \eta'_{1},\eta'_{2}, \pi'))$, such that for each $j \in [k]$ with $A(j, k+1) = 1$, we have $\eta'_1(x_j)=\eta_1(z'_i)$ and $\eta'_2(x_j)=\eta_2(z'_i)$,
\item
$((q,\eta_1,\eta_2),  (c'_{1,low}, c'_{2,up}), (q, q_f, \eta'_{1},\eta'_{2}, \pi'))$, such that for each $j \in [k]$ with $A(j, k+1) = 1$, we have $\eta'_1(x_j)=\eta_1(z'_i)$ and $\eta'_2(x_j)=\eta_2(z'_{i+1})$,
\item
$((q,\eta_1,\eta_2),  (c'_{1,up}, c'_{2,low}),  (q, q_f, \eta'_{1},\eta'_{2}, \pi'))$, such that for each $j \in [k]$ with $A(j, k+1) = 1$, we have $\eta'_1(x_j)=\eta_1(z'_{i+1})$ and $\eta'_2(x_j)=\eta_2(z'_i)$,
\item
$((q,\eta_1,\eta_2),  (c'_{1,up}, c'_{2,up}),  (q, q_f, \eta'_{1},\eta'_{2}, \pi'))$, such that for each $j \in [k]$ with $A(j, k+1) = 1$, we have $\eta'_1(x_j)=\eta_1(z'_{i+1})$ and $\eta'_2(x_j)=\eta_2(z'_{i+1})$,
\end{itemize}
where
\begin{itemize}
\item  $c'_{1, low} = \sum \limits_{\pi'(j) \neq 0} \va(\pi'(j)) \cdot \vc_{1,low}(j)$, $c'_{2, low} = \sum \limits_{\pi'(j) \neq 0} \va(\pi'(j)) \cdot  \vc_{2,low}(j)$, 
\item and $c'_{1, up} = \sum \limits_{\pi'(j) \neq 0} \va(\pi'(j)) \cdot \vc_{1,up}(j)$, and $c'_{2, up} = \sum \limits_{\pi'(j) \neq 0} \va(\pi'(j)) \cdot  \vc_{2,up}(j)$,
\item $\vc_{1,low}$ and $\vc_{2,low}$
denote the vectors $\vf(\vx, \cur)$
with $(\vx, \cur)$ being substituted with $(\eta_1(\vx), \eta_1(z'_i))$
and $(\eta_2(\vx),\eta_2(z'_i))$ respectively,
\item $\vc_{1,up}$ and $\vc_{2,up}$
denote the vectors $\vf(\vx,\cur)$
with $(\vx,\cur)$ being substituted with $(\eta_1(\vx), \eta_1(z'_{i+1}))$
and $(\eta_2(\vx), \eta_2(z'_{i+1}))$ respectively.
\end{itemize}

\smallskip

\item[{\bf [T4]}] 
If $\varphi$ is of the form $z'_{k'} \le \cur$,
then the following transitions are in $\Delta$:
\begin{itemize}
\item
$((p, q_f, \eta_1,\eta_2, \pi),  (c'_{1,low}, c'_{2,low}),  (q, q_f, \eta'_{1},\eta'_{2}, \pi'))$, such that for each $j \in [k]$ with $A(j, k+1) = 1$, we have $\eta'_1(x_j)=\eta_1(z'_{k'})$ and $\eta'_2(x_j)=\eta_2(z'_{k'})$,
\item
$((q,\eta_1,\eta_2),  (c'_{1,low}, c'_{2,up}), (q, q_f, \eta'_{1},\eta'_{2}, \pi'))$, such that for each $j \in [k]$ with $A(j, k+1) = 1$, we have $\eta'_1(x_j)=\eta_1(z'_{k'})$ and $\eta'_2(x_j)=+\infty$,
\item
$((q,\eta_1,\eta_2),  (c'_{1,up}, c'_{2,low}),  (q, q_f, \eta'_{1},\eta'_{2}, \pi'))$, such that for each $j \in [k]$ with $A(j, k+1) = 1$, we have $\eta'_1(x_j)=+\infty$ and $\eta'_2(x_j)=\eta_1(z'_{k'})$,
\item
$((q,\eta_1,\eta_2),  (c'_{1,up}, c'_{2,up}),  (q, q_f, \eta'_{1},\eta'_{2}, \pi'))$, such that for each $j \in [k]$ with $A(j, k+1) = 1$, we have $\eta'_1(x_j)=+\infty$ and $\eta'_2(x_j)=+\infty$,
\end{itemize}
where
\begin{itemize}
\item  $c'_{1, low} = \sum \limits_{\pi'(j) \neq 0} \va(\pi'(j)) \cdot \vc_{1,low}(j)$, and $c'_{2, low} = \sum \limits_{\pi'(j) \neq 0} \va(\pi'(j)) \cdot  \vc_{2,low}(j)$, 
\item $c'_{1, up} = \sum \limits_{\pi'(j) \neq 0} \va(\pi'(j)) \cdot \vc_{1,up}(j)$, and $c'_{2, up} = \sum \limits_{\pi'(j) \neq 0} \va(\pi'(j)) \cdot  \vc_{2,up}(j)$,
\item $\vc_{1,low}$ and $\vc_{2,low}$
denote the vectors $\vf(\vx,\cur)$
with $(\vx,\cur)$ being substituted with $(\eta_1(\vx),\eta_1(z'_{k'}))$
and $(\eta_2(\vx), \eta_1(z'_{k'}))$ respectively,
\item $\vc_{1,up}$ and $\vc_{2,up}$
denote the vectors $\vf(\vx,\cur)$
with $(\vx,\cur)$ being substituted with $(\eta_1(\vx), +\infty)$
and $(\eta_2(\vx), +\infty)$ respectively.
\end{itemize}

\item[{\bf [T5]}] Finally, $\Delta$ includes the following transitions involving $q'_0$ and $q'_f$.
\begin{itemize}
\item For each $q_f \in F$ and $\pi: [l] \rightarrow [l] \cup \{0\}$ such that $\zeta(q_f) = \va \cdot \myvec{\vx \\ \vy} + a'$, the transition $(q'_0, (c_1, c_2), (q_0,q_f, \eta, \eta, \pi)) \in \Delta$, where $c_1 = \sum \limits_{\pi(j) \neq 0} \va(\pi(j)) \cdot \vu_0\ssY(j)$, $c_2 = \sum \limits_{\pi(j) \neq 0} \va(\pi(j)) \cdot \vu_0\ssY(j))$, and $\eta(x_i)=\vu_0\ssX(i)$ for each $x_i \in X$.
\item For each state $(q_f, q_f, \eta_1, \eta_2, \pi) \in S$ such that $\pi(j)=j$ for each $j \in [l]$, suppose $\zeta(q_f) = \va \cdot \myvec{\vx \\ \vy} + a'$, then the transition $((q_f, q_f, \eta_1, \eta_2, \pi), (c_1, c_2), q'_f) \in \Delta$, where $c_1 = a' +  \va\ssX \cdot \eta_1(\vx)$ and $c_2 = a' + \va\ssX \cdot \eta_2(\vx)$. \\
The requirement that $\pi(j) = j$ for each $j \in [l]$ is consistent with the intuition of $\pi$: When reaching the final state $q_f$, for each $j \in [l]$, the current value of $y_j$ is stored in $y_{j=\pi(j)}$. Note that the contents of $y_j \in Y$ have been added to the final output, this is why $c_1, c_2$ only need take the contents of $\vx$ into consideration.
\end{itemize}
\end{enumerate}

%Let us define initial and final states, initial and final configurations, of $\cB$, respectively.
%\begin{itemize}
%\item An \emph{initial state} of $\cB$ is $(q_0,q_f, \eta, \eta, \pi) \in S$ such that $q_f \in F$, $\eta(x_i)=\vu_0\ssX(i)$ for each $x_i \in X$, and $\pi: [l] \rightarrow [l] \cup \{0\}$. 
%
%\item An \emph{initial configuration} of $\cB$ is $((q_0,q_f, \eta, \eta, \pi),  c_1, c_2)$ such that $(q_0,q_f, \eta, \eta, \pi)$ is an initial state, $c_1 = a' + \sum \limits_{\pi(j) \neq 0} \va(\pi(j)) \cdot \vu_0\ssY(j)$, and  $c_2 = a'+ \sum \limits_{\pi(j) \neq 0} \va(\pi(j)) \cdot \vu_0\ssY(j))$. 
%
%\item A \emph{final state} of $\cB$ is $(q_f,q_f, \eta_1, \eta_2, \pi) \in S$ such that $q_f \in F$, and $\pi(j)=j$ for each $j \in [l]$. The requirement that $\pi(j) = j$  for each $j \in [l]$ is consistent with the intuition of $\pi$: When reaching the final state $q_f$, for each $j \in [l]$, the current value of $y_j$ is stored in $y_{j=\pi(j)}$.
%
%\item 
%A \emph{final configuration} of $\cB$ is $((q_f,q_f, \eta_1, \eta_2, \pi),  c_1, c_2)$ such that $(q_f,q_f, \eta_1, \eta_2, \pi)$ is a final state. 
%\end{itemize}

\medskip

Then there is $w$ such that $0 \in \cA(w)$
iff there is a configuration
$(q'_f, c'_1, c'_2)$ reachable from $(q'_0, 0, 0)$ in $\cB$ such that either $c'_1 \le 0 \le c'_2$ or $c'_2 \le 0 \le c'_1$.
%$$
%\zeta(q')(\eta_1(\vx),\vv_1) \ \leq \ 0
%\ \leq \
%\zeta(q')(\eta_2(\vx),\vv_2),
%$$
%or
%$$
%\zeta(q')(\eta_2(\vx),\vv_2) \ \leq \ 0
%\ \leq \
%\zeta(q')(\eta_1(\vx),\vv_1).
%$$

Therefore, the reachability problem for the $\raq$ $\cA$ is reduced to the  configuration coverability problem for $\bbQ$-VASS $\cB$ (via a Karp reduction),
which in turn, can be reduced to satisfiability problem for 
existential Presburger formula.
See Appendix~\ref{app-rat-vass} for the details.

\medskip

\paragraph*{Complexity analysis} Suppose a non-normalized $\raq$ $\cA$ is given as the input, with $n$ states, $k$ control variables, and $l$ data variables respectively. Then as mentioned before, the number of states of the normalized $\raq$ $\cA'$ is at most $n2^{2k-1}(2k)!$. The number of specifications $\eta$ is at most $k^{2k}$ (at most $2k$ control variables and at most $k$ constants in $\cA'$), while the number of mappings $\pi$ is at most $(l+1)^l$. Since  each state from $S$ is of the form $(q, q_f, \eta_1, \eta_2, \pi)$ (except two special states), it follows that the cardinality of $S$ is at most $(n2^{2k-1}(2k)!)^2 (k^{2k})^2 (l+1)^l$, which is exponential over the size of $\cA$. Because the reachability problem for $\bbQ$-VASS is in $\nptime$ (cf. Theorem~\ref{thm-str-reach-q-vass} in Appendix~\ref{app-rat-vass}), we conclude that the reachability problem for copyless $\raq$ is in $\nexptime$.

%%%%%%%%%%%%%%%%%%%%%%%%%%%%%%%%%%%%%%%%%%%%%%%%%%%%%%%%%%
%%%%%%%%%%%%%%%%%%%%%%The original reduction by Zhilin%%%%%%%%%%%%%%%%%%%%%
%%%%%%%%%%%%%%%%%%%%%%%%%%%%%%%%%%%%%%%%%%%%%%%%%%%%%%%%%%
\hide{
Let us fix a copyless $\raq$ $\cA=\langle Q, q_0, \vec{u}_0, F, \delta, \zeta \rangle$.  

%According to the definition of copyless $\raq$, for each transition $t=(p, \varphi(\vec{x}, \cur)) \rightarrow (q, A, B, \vec{b})$ of $\cA$, let $B=(B_1\ B_2\ B_3)$ such that $B_1 \in \ratnum^{l \times k}$, $B_2 \in \ratnum^{l \times l}$ and $B_3 \in \ratnum^{l \times 1}$, 
%then each entry of $B_2$ is $0$ or $1$, in addition, there is at most one $1$-entry in \emph{each column} of $B_2$.

Let $\cN$ denote the set of rational constants occurring in $\{\vec{u}_{0}(i) \mid i \in [k]\}$. 

Without loss of generality, we assume that $\cN = \{ c_{-r}, \dots, c_{-1}, c_0, c_1, \dots, c_s\}$ such that $c_0 = 0$, and $ c_{-r} < \dots < c_{-1} < c_0 < c_1 < \dots < c_s$. The constants in $\const(\cA)$ separate the set of rational numbers $\ratnum$ into $r + s + 2$ open intervals, i.e. $(-\infty, c_{-r})$, $(c_{-r}, c_{-r+1}), \dots, (c_{s-1}, c_s)$, and $(c_s, +\infty)$.  Let $\intval$ denote the set of these intervals. 

By utilizing Corollary~\ref{cor-po-mult-intval-ls-gr-0}, we show that the reachability problem of $\cA$ can be reduced to the configuration coverability problem of a rational VASS $\cB$.

Let $P$ be a path from $q_0$ to some $q_f \in F$. 
A run of $\cA'$ on an input $w$ is said to \emph{follow} $P$ if the sequence of transitions in the run is exactly $P$. For each transition $(p, \varphi(\vx, \cur))  \to  (q, A, B, \vb)$ on $P$ such that $A(i, k+1) \neq 0$  for some $i \in [k]$ or $B(i, k+l + 1) \neq 0$ for some $i \in [l]$ (intuitively, this means that $\varphi(\vx, \cur)$ does not require that $\cur$ is equivalent to any control variable, thus $\cur$ occurs in the expressions to update the control or data variables),   a fresh variable is introduced to represent the input data value read when executing this transition. Suppose $m$ fresh variables, say $z_1,\dots, z_m$, are introduced when traversing $P$. Then from Proposition~\ref{prop:linear-raq}, the output of each run of $\cA'$ that follows $P$ can be specified by an expression of the form $a_0 + a_1 z_1 + \dots + a_m z_m$, where $a_0, a_1, \dots, a_m$ are rational constants. Let $Z=\{z_1,\dots, z_m\}$. Then the guards and assignments of the transitions on $P$ induce a partial order $\preceq_P$ on $Z \cup \const(\cA)$. Note that the restriction of $\preceq_P$ to $\const(\cA)$ is identical to the rational order relation on $\const(\cA)$. Let us use $\prec_P$ to denote $\preceq_P \setminus \{(z, z) \mid z \in Z \cup \const(\cA)\}$. For each $I \in \intval$, we use $Z_I$ to denote the set of variables $z_i \in Z$ such that $z_i$ should take a value from $I$.

The main idea of the reduction is as follows: 
\begin{quote}
\it Guess \emph{on the fly} a path $P$ of $\cA'$ from $q_0$ to some $q_f \in F$ and two upward-closed subsets $Z'_I, Z''_I$ of $Z_I$ with respect to $\preceq_P$ for each $I \in \intval$, use the rational variables $\overrightarrow{C'} = (C'_0, C'_1, C'_2)$ to store the value of the expression (\ref{eqn-po-ls-0-lem}) and $\vec{C''} = (C''_0, C''_1, C''_2)$ to store the value of the expression  (\ref{eqn-po-gr-0-lem}), update $\overrightarrow{C'}, \overrightarrow{C''}$ according to the transitions of $P$, and when $q_f$ is reached, check that the values of $C', C''$ are less than resp. greater than zero. 
\end{quote}
The three variables $C'_0, C'_1, C'_2$ correspond to $d_0, d_1,d_2$ in Remark~\ref{rem-po}, where $d_1,d_2$ are the coefficients of $c_{-\infty}$ and $c_{+\infty}$ respectively. Similarly for $C''_0, C''_1, C''_2$.

We are ready to present the construction of $\cB$.

%We are ready to present the construction of $\cB$, which comprises two steps.

At first, $\cB$ nondeterministically chooses a state $q_f \in F$. Suppose $\zeta(q_f) = c + \sum^{k}_{i=1} \alpha_i x_i + \sum^l_{i=1} \beta_i y_i$. 
Then $\cB$ guesses nondeterministically a path $P$ of $\cA$ from $q_0$ to $q_f$.

%The three variables $C'_0, C'_1, C'_2$ correspond to $c'_0, c'_1, c'_2$ in Remark~\ref{rem-po} (where $c'_1, c'_2$ are the coefficients of $c_{-\infty}$ and $c_{+\infty}$ respectively), similarly for $C''_0, C''_1, C''_2$. We will use $\overrightarrow{C'}$ and $\overrightarrow{C''}$ to denote the vector $(C'_0, C'_1, C'_2)$ and $(C''_0, C''_1, C''_2)$ respectively.

%$\cB$ records the last state guessed on $P$, as a component of its control state.

For each guessed transition $t$ of $P$, $\cB$ simulates the execution of $t$, and updates the rational variables $\overrightarrow{C'}, \overrightarrow{C''}$ to record the values of the two expressions (\ref{eqn-po-ls-0-lem}) and (\ref{eqn-po-gr-0-lem}). 
%
%The updates are based on the following observation, resulted from the copyless and independently evolving constraint: 
%\begin{quote}
%\it For each guessed transition $t$ on $P$, where a fresh variable $z_{i'}$ is introduced, and each $i \in [l]$,  if $d$ copies of $z_{i'}$ are added to $y_i$ in $t$, then these $d$ copies of $z_{i'}$ will persist in $y_i$, thus eventually $\beta_i d$ copies of $z_{i'}$ will be added to the final output, unless $y_i$ is reset in the future.
%\end{quote}
%

When a transition $t$ of $\cA$ is guessed,  suppose that a fresh variable $z_{i'}$ is introduced to represent the data value read by $t$, and $z_{i'}$ should take a value from some interval $I \in \intval$, in order to update $\overrightarrow{C'}, \overrightarrow{C''}$, we need the following two types of information, 
\begin{enumerate}
\item whether $z_{i'}$ belongs to $Z'_I$ resp. $Z''_I$,
\item for each $y_i \in Y$, how the copies of $z_{i'}$ added to $y_i$ in $t$ will evolve in the future, that is, when reaching the state $q_f$ and the guessing of $P$ is over, whether these copies of $z_{i'}$ have been lost, or these copies of $z_{i'}$ are still kept in some data variable $y_{j}$. 
\end{enumerate}
To acquire the first type of information, for each $I \in \intval$, $\cB$ introduces the integer variables $idx'_I$ and $idx''_I$ that range over $[k] \cup \{0\}$, with the following intuition:
%Similarly,  for each $I \in \intval$, $\cB$ introduces the integer variable $idx''_I$ that ranges over $[k] \cup \{0\}$, to acquire the information whether $z_{i'}$ belongs to $Z''_I$.  
%The intuition of the variables $idx'_I, idx''_I$ is explained below:
%\begin{quote}
%In each configuration reachable from the initial configuration, for each control variable $x_j \in X$, the data value stored in $x_j$ is either equal to some constant from $\const(\cA) \setminus \{c_{-\infty}, c_{+\infty}\}$, or belongs to $I$ for some $I \in \intval$. 
If the current state of $\cA$ is $q$, then $idx'_I$ records the index $j$ of the control variable $x_{j}$ such that  the current value of $x_{j}$, which is represented by some fresh variable $z_{i'}$, belongs to $Z'_I$, in addition, $[x_{j}]_q$ is the \emph{$\preceq_q$-minimum} equivalence class of $\simeq_q$ satisfying this condition ($idx'_I = 0$ means that no values stored in control variables belong to $Z'_I$). The intuitive meaning of $idx''_I$ is similar, with $Z'_I$ replaced by $Z''_I$. 
%\end{quote}
Note that the integer variables $idx'_I, idx''_I$ should be taken as components of the control states of $\cB$.

To acquire the second type of information, we introduce a mapping $pmt: [l] \rightarrow [l] \cup \{0\}$ to keep track of the evolvement of the values of data variables. Intuitively, for each $i \in [l]$, if $pmt(i) \in [l]$, then the current value of $y_i$ will be stored in $y_{pmt(i)}$ eventually when reaching the state $q_f$ and the guessing is over, otherwise, the current value of $y_i$ will be lost. Note that $pmt$ should still be understood as a component of the control states of $\cB$.

\smallskip

\noindent {\bf The initial values of the variables}. \\ 
The initial values of the variables in $\cB$ are set as follows: 
\begin{itemize}
\item for each $I \in \intval$, let $(idx'_I, idx''_I) := (0,  0)$,  
\item for each $i \in [l]$, nondeterministically choose the value of $pmt(i)$ from $[l] \cup \{0\}$, 
\item in addition, let $\overrightarrow{C'} := (c, 0, 0)$ and $\overrightarrow{C''}:=(c, 0, 0)$ (recall that $c$ is the constant in $\zeta(q_f)$).
\end{itemize}
%$C'_1:=0$, $C'_2 = 0$, $C''_0 := c$, $C''_1:=0$, and $C''_2:=0$.

\smallskip

\noindent {\bf The transition rules of $\cB$ to simulate the transitions of $\cA$}. \\  
For each transition $t$ of $\cA$, $\cB$ contains the transition rules to update the variables  according to $t$.

Let us fix a transition $t =  (p, \varphi(\vec{x}, \cur)) \rightarrow (q, A, B, \vec{b})$ of $\cA$ from now on. 

Suppose $\myvec{x_1 := r_1; \ldots; x_k := r_k; y_1 := s_1; \ldots; y_l := s_l}$ is the assignments in $t$. 
Then for each $i \in [l]$, 
$$s_i =b_i + \sum \limits_{j \in [l]} B(i, k+j) y_j + \sum \limits_{j \in [k]} B(i, j) x_j + B(i, k+l+1)\ \cur.$$

The updates of the variables in $\cB$ comprise three phases.
\begin{itemize}
\item {\bf Phase 1}:  Update the mapping $pmt$, according to $ \sum \limits_{j \in [l]} B(i, k+j) y_j $  in $s_i$ for $i \in [l]$.
\item {\bf Phase 2}: Update $\overrightarrow{C'}$ and $\overrightarrow{C''}$, according to $b_i$ and the expressions $\sum \limits_{j \in [k]} B(i, j) x_j$  in $s_i$  for $i \in [l]$.
\item {\bf Phase 3}: Update $idx'_I$, $idx''_I$ for $I \in \intval$, $\overrightarrow{C'}$ and $\overrightarrow{C''}$, according to the guard $\varphi(\vec{x}, \cur)$ and the expressions $B(i, k+l+1)\ \cur$  in $s_i$ for $i \in [l]$. 
\end{itemize}

\smallskip

\noindent {\bf Phase 1}.

For preparation, we define a mapping $storedIn_t: [l] \rightarrow [l] \cup \{0\}$, to record where the original value of each data variable is stored after executing $t$, as follows: For each $i \in [l]$,
\begin{itemize}
\item if there is $i' \in [l]$ such that $B(i', k+i) = 1$ (intuitively, the original value of $y_i$ is stored in $y_{i'}$ after executing $t$), then $storedIn_t(i) = i'$, 
\item otherwise, $storedIn_t(i) = 0$. 
\end{itemize}
Note the because of the copyless constraint, for each $i \in [l]$, there is at most one $i' \in [l]$ such that $B(i', k+i) = 1$, thus $storedIn_t$ is well-defined. 

\smallskip

In Phase 1, $\cB$ first makes sure that the following two constraints are satisfied:
\begin{itemize}
\item For each $i \in [l]$ such that $pmt(i) \neq 0$, it holds that $storedin_t(i) \neq 0$. \\
Intuitively, if the current value of $y_i$ will be stored in some data variable eventually, then the current value of $y_i$ should not be lost after executing $t$.
\item For each $i_1, i_2 \in [l]$ such that $storedin_t(i_1) = storedin_t(i_2) \neq 0$, it holds that $pmt(i_1) = pmt(i_2)$. \\
Intuitively, if the current value of $y_{i_1}$ and $y_{i_2}$ will be stored in the same data variable after executing $t$, then eventually, either both of them will be lost, or otherwise they will be stored in the same data variable. 
\end{itemize}

Then the mapping $pmt$ is updated as follows: For each $i \in [t]$,
\begin{itemize}
\item if there is $i' \in [l]$ such that $storedin_t(i')=i$, then let $pmt(i):=pmt(i')$ (intuitively, if the original value of $y_{i'}$ is stored in $y_i$ after executing $t$ and the original value of $y_{i'}$ will be stored in $y_{pmt(i')}$ eventually when the guessing of $P$ is over, then the value of $y_i$ after executing $t$ should be stored in $y_{pmt(i')}$ as well, thus $pmt(i)$ is updated by taking the original value of $pmt(i')$),
\item otherwise,  $\cB$ nondeterministically assigns to $pmt(i)$ a value from  $[l] \cup \{0\}$ (in order to guess where the value assigned to $y_i$ in $t$ will be stored eventually). 
\end{itemize}
The description of Phase I is finished.

\medskip

From now on, for each $i \in [l]$, when we refer to $pmt(i)$, we mean the (possibly) updated value of $pmt(i)$.

%For convenience, we use $Y_{nr}$ to denote the set of data variables $y_i \in Y$ that will not be reset any more, that is, $reset_i = \lfalse$

%At first, $\cB$ makes sure that for each $i \in [l]$, if $reset_{i} = \lfalse$, then $B(i, k+i) = 1$ (i.e. $y_i$ is not reset). For each $i \in [l]$ such that $reset_{i} = \ltrue$, if $B(i, k+i) = 0$ (i.e. $y_i$ is reset in $t$), then $\cB$ nondeterministically changes the value of $reset_{i}$ to $\lfalse$ (thus $y_i$ will not be reset any more). Then from now on, for each $i \in [l]$, when we refer to $reset_{i}$, we mean the (possibly) updated value of $reset_{i}$. For convenience, we use $Y_{nr}$ to denote the set of data variables $y_i \in Y$ that will not be reset any more, that is, $reset_i = \lfalse$.

%Let $[z'_1]_q, \dots, [z'_t]_q$ be an enumeration of the equivalence classes of $\simeq_q$ such that $z'_1 \prec_q \dots \prec_q z'_t$.  

Before describing Phase 2 and 3, we introduce some additional notations.
\begin{itemize}
\item For $x_j \in X$, $x_j$ is called \emph{$\preceq_p$-rigid} if there is $x_{j'} \in X_1$ such that $x_j \simeq_p x_{j'}$.  We use $R_{p}$ to denote the set of $x_j \in X$ such that $x_j$ is $\preceq_p$-rigid. Note that $X_1 \subseteq R_p$. Intuitively, $x_j$ is $\preceq_p$-rigid if its value is required by $\preceq_p$ to be equivalent to some rational constant stored in a control variable from $X_1$.
\item We extend the mapping $cnst$ from $X_1$ to $R_p$ as follows: For each $x_j \in R_p \setminus X_1$, $cnst(x_j) = cnst(x_{j'})$ such that $x_{j'} \in X_1$ and $x_j \simeq_p x_{j'}$.   
\item In addition, we define a mapping $ass: [k] \rightarrow [k+1]$ to represent the assignment of control variables in $t$ as follows: For each $i \in [k]$, $ass(i) = k+1$ iff $A(i, k+1)  = 1$, and $ass(i) = j \in [k]$ iff $A(i, j)  = 1$. Intuitively, $ass(i)= k+1$ if $x_i$ is assigned the value of $\cur$, and $ass(i) = j \in [k]$ if $x_i$ is assigned the value of $x_j$.
\end{itemize}

%We first update $\overrightarrow{C'}, \overrightarrow{C''}$ as follows. 

\smallskip

\noindent {\bf Phase 2}.
When updating the rational variables $\overrightarrow{C'}$ and $\overrightarrow{C''}$, we will ignore the assignments $s_i$ such that $pmt(i)=0$, since the value of $y_i$ will be lost eventually and $s_i$ will not contribute to the final output.

We first update $\overrightarrow{C'}$ and $\overrightarrow{C''}$ for the $\preceq_p$-rigid control variables by letting 
$$C'_0:= C'_0 + \sum \limits_{pmt(i) \neq 0} \beta_{pmt(i)} \left(b_{i} + \sum \limits_{x_j \in R_p} B(i, j)\ cnst(x_j)\right) \mbox{\ \ and\ \ } C''_0:=C''_0 + \sum \limits_{pmt(i) \neq 0} \beta_{pmt(i)} \left(b_{i} + \sum \limits_{x_j \in R_p} B(i, j)\ cnst(x_j)\right).$$
%$$C''_0 := C''_0+ \sum \limits_{i' \in [l], reset_{i'}  = \lfalse} \beta_{i'} (b_{i'} + \sum \limits_{j \in R_q} A(i', j) val(j)).$$
Intuitively, for each data variable $pmt(i) \neq 0$, since $t$ adds the constant $b_{i} + \sum \limits_{x_j \in R_p} B(i, j)\ cnst(x_j)$ to $y_{i}$, which will be stored in $y_{pmt(i)}$ eventually,  the constant $ \beta_{pmt(i)} \left(b_{i} + \sum \limits_{x_j \in R_p} B(i, j)\ cnst(x_j) \right)$ will be added to the final output.

For each non-$\preceq_p$-rigid control variable $x_j$, its value, which is represented by some fresh variable $z_{i}$, belongs to some $I \in \intval$. We wil show how to update $\overrightarrow{C'}$. The updates of $\overrightarrow{C''}$ are analogous. We  distinguishing whether $z_i$ belongs to $Z'_I$ or not (characterised by the condition $idx'_I > 0$ and $x_{idx'_I} \preceq_p x_j$):
\begin{itemize}
\item If $idx'_I > 0$ and $x_{idx'_I} \preceq_p x_j$, then 
\begin{itemize}
\item if $\supr_I \neq c_{+\infty}$, then let $C'_0: = C'_0 + \supr_I \sum \limits_{pmt(i') \neq 0} \beta_{pmt(i')} B(i', j)$,
\item otherwise, let $C'_2: = C'_2 +  \sum \limits_{pmt(i') \neq 0} \beta_{pmt(i')} B(i', j)$.
\end{itemize}
\item otherwise, 
\begin{itemize}
\item if $\infr_I \neq c_{-\infty}$, then let
$C'_0: = C'_0 + \infr_I \sum \limits_{pmt(i') \neq 0} \beta_{pmt(i')} B(i', j)$, 
\item otherwise, let
$C'_1: = C'_1 + \sum \limits_{pmt(i') \neq 0} \beta_{pmt(i')} B(i', j)$.
\end{itemize}
\end{itemize}
Intuitively, since the value of $x_j$, represented by some $z_{i}$, belongs to $Z'_I$, and $B(i', j)$ copies of $z_{i}$ are added to $y_{i'}$ in $t$, in addition, the value of $y_{i'}$ will be stored in $pmt(i')$ eventually, it follows that $\beta_{pmt(i')} B(i', j)$ copies of $z_{i}$ will be added to the final output. Because $z_{i}$ can take a value arbitrarily close to $\supr_I$ resp. $\infr_I$, we add $\supr_I \beta_{pmt(i')} B(i', j)$ resp. $\infr_I \beta_{pmt(i')} B(i', j)$ to $C'_0$, if $\supr_I \neq c_{+\infty}$ resp. $\infr_I \neq c_{-\infty}$, and add  $\beta_{pmt(i')} B(i', j)$ to $C'_2$ and $C'_1$ respectively otherwise.

%\item If $idx'_I = 0$ and $c \neq c_{-\infty}$, then let 
%$$C'_0 := C'_0 + c \sum \limits_{y_{i'} \in Y_{nr}} \beta_{i'} B(i', j).$$
%
%Intuitively, since $idx'_I = 0$, we know that the value of $x_j$, represented by some $z_i$, does not belong to $Z'_I$. Therefore,  $z_i$ can take a value arbitrarily close to $c$ and we add $c\ \beta_{i'} B(i', j)$ to $C'_0$.
%
%\item If $idx'_I = 0$ and $c = c_{-\infty}$, then let 
%$$C'_1 := C'_1 + \sum \limits_{y_{i'} \in Y_{nr}} \beta_{i'} B(i', j).$$
%The intuition is the same as the previous case, except we add to $C'_1$, instead of $C'_0$.
%\end{itemize}

\smallskip

\noindent {\bf Phase 3}.
Let $[z'_1]_p, \dots, [z'_t]_p$ be an enumeration of the equivalence classes of $\simeq_p$ on $X$. We will exemplify the updates by considering the typical situation 
that $\varphi(\vec{x},\cur)$ is of the form $z'_i < \cur < z'_{i+1}$ for some $i \in [t-1]$. For the other situations, the updates are similar.  In addition, since the updates of $idx''_I$ and $\overrightarrow{C''}$ are analogous to those of $idx'_I$ and $\overrightarrow{C'}$, we will only illustrate how to update $idx'_I$ and $\overrightarrow{C'}$.

%Let $\simeq$ denote the equivalence relation induced by $\preceq$. For each $i \in [k]$, let $[x_i]$ denote the equivalence class of $\simeq$ containing $x_i$.

Suppose $\varphi(\vec{x},\cur)$ is of the form $z'_i < \cur < z'_{i+1}$ for some $i \in [t-1]$.
Since $\preceq_p$ is consistent with the rational order relation on $\const(\cA)$ and $[z'_{i+1}]_p$ is the $\preceq_p$-successor of $[z'_i]_p$, we know that there is a \emph{unique} $I \in \intval$ such that $cnst(x_j) = c$, $cnst(x_{j'}) = d$, and $x_{j} \preceq_p z'_i \prec_p z'_{i+1} \preceq_p x_{j'}$ for some $x_j, x_{j'} \in X_1$. 

In order to update $\overrightarrow{C'}$, we need know the information whether the value of $\cur$ belongs to $Z'_I$. Such information is maintained as follows.
\begin{itemize}
\item If $idx'_I > 0$ and $x_{idx'_I} \preceq_p z'_i$, then the value of $\cur$ has to belong to $Z'_I$, since the value of $x_{idx'_I}$ belongs to $Z'_I$ and $Z'_I$ is upward-closed. 
\item Otherwise, we make a nondeterministic choice to let $\cur$ belong to $Z'_I$. Since it is necessary to guarantee the consistency of these nondeterministic choices, we only allow these choices for the the following two cases: $z'_{i+1}$ is $\preceq_p$-rigid (i.e. $z'_{i+1} \simeq_p x_{j'}$), or otherwise the value of $z'_{i+1}$ belongs to $Z'_I$ (i.e. $z'_{i+1} \prec_p x_{j'}$, $idx'_I > 0$, and $x_{idx'_I} \simeq_p z'_{i+1}$). For contradiction, if we choose $\cur$ to belong to $Z'_I$ in the case that $z'_{i+1}$ is non-$\preceq_p$-rigid and the value of $z'_{i+1}$ does not belong to $Z'_I$, then from the fact that $\cur$ belongs to $Z'_I$ and  $Z'_I$ is upward-closed, we deduce that $z'_{i+1}$ has to belong to $Z'_I$ as well,  which would contradict to the nondeterministic choice made before, i.e. the value $z'_{i+1}$ does not belong to $Z'_I$.
%$idx'_I > 0$ and $x_{idx'_I} \simeq_p z'_{i+1}$ or $idx'_I = 0$ and $z'_{i+1} \simeq_p x_{j'}$.
\end{itemize}

%We distinguish between the following cases.
%\begin{itemize}
%\item $idx'_I > 0$ and $x_{idx'_I} \preceq_p z'_i$,
%\item $idx'_I > 0$ and $x_{idx'_I} \simeq_p z'_{i+1}$,
%\item $idx'_I > 0$ and $ z'_{i+1} \prec_p x_{idx'_I}$,
%\item $idx'_I = 0$ and $z'_{i+1} \simeq_p x_{j'}$,
%\item $idx'_I = 0$ and $z'_{i+1} \prec_p x_{j'}$.
%\end{itemize}

More specifically, we update $idx'_I$ and $\overrightarrow{C'}$ as follows.

\smallskip

\noindent {\it Case $idx'_I > 0$ and $x_{idx'_I} \preceq_p z'_i$}.

\smallskip

\begin{itemize}
\item Let $post'(t, X_2)$ denote the set of control variables $x_{j''} \in X_2$ such that after the execution of $t$, $x_{j''}$ stores the value of $\cur$ or the original value of some control variable belonging to $Z'_I$. More specifically, $post'(t, X_2)=\{x_{j''} \in X_2 \mid ass(j'') = k+1, \mbox{ or } ass(j'') \in [k] \mbox{ and } x_{idx'_I} \preceq_p x_{ass(j'')} \preceq_p x_{j'}\}$. 
\begin{itemize}
\item If $post'(t, X_2) \neq \emptyset$, 
then let $idx'_I$ be some $j'' \in [k]$ such that $x_{j''}$ is a $\preceq_{q}$-minimal element in $post'(t, X_2)$. Note that $\preceq_{q}$ is the total preorder associated with $q$, instead of $p$.
\item Otherwise, let $idx'_I := 0$.
\end{itemize}
Intuitively, if after the execution of $t$, the value of $\cur$ or the original value of some control variable belonging to $Z'_I$, remains in some control variable $x_{j''}$, then let $idx'_I$ be the index of  a $\preceq_{q}$-minimal such $x_{j''}$. 
%\begin{itemize}
%\item if $\{j'' \in \rng(\pi) \mid x_{idx'_I} \preceq_q x_{j''} \preceq_q x_j \} \neq \emptyset$, let $idx'_I: = j'''$ for some $j''' \in \pi^{-1}(j'')$  such that $x_{j''}$ is $\preceq_q$-minimal in $\{x_{j''} \mid j'' \in \rng(\pi), x_{idx'_I} \preceq_q x_{j''} \preceq_q x_j \}$, 
%
%\item otherwise, if $k+1 \in \rng(\pi)$, then let $idx'_I: = j''$ for some $j'' \in \pi^{-1}(k+1)$, 
%
%\item otherwise, if $\{j'' \mid j'' \in \rng(\pi), x_{j'} \preceq x_{j''} \} \neq \emptyset$, let $idx'_I: = j'''$ for some $j''' \in \pi^{-1}(j'')$  such that $x_{j''}$ is $\preceq$-minimal in $\{x_{j''} \mid j'' \in \rng(\pi), x_{j'} \preceq x_{j''}\}$, 
%
%\item otherwise, let $idx'_I:=0$,
%\end{itemize}
%
\item Let
$
C'_0  : = C'_0 + 
d \sum \limits_{pmt(i') \neq 0} \beta_{pmt(i')} B(i', k + l + 1).
$ \\
Intuitively, $C'_0$ is updated by considering how many copies of $\cur$ are added to $y_{i'}$. 
\end{itemize}

\smallskip

\noindent {\it Case $idx'_I = 0$ and $z'_{i+1} \simeq_p x_{j'}$, or $idx'_I > 0$, $z'_{i+1} \prec_p x_{j'}$ and $x_{idx'_I} \simeq_p z'_{i+1}$}.

\smallskip

Then $\cB$ nondeterministically chooses to do one of the following (intuitively, $\cB$ guesses whether $\cur$ belongs to $Z'_I$ or not).
\begin{enumerate}
\item $\cur$ is chosen to belong to $Z'_I$: 
\begin{itemize}
\item Let $post'(t, X_2)=\{x_{j''} \in X_2 \mid ass(j'') = k+1, \mbox{ or } ass(j'') \in [k], idx'_I > 0, \mbox{ and } x_{idx'_I} \preceq_p x_{ass(j'')} \preceq_p x_{j'}\}$. 
\begin{itemize}
\item If $post'(t, X_2)$ is nonempty, 
then let $idx'_I$ be some $j'' \in [k]$ such that $x_{j''}$ is a $\preceq_{q}$-minimal element in $post'(t, X_2)$. 
\item Otherwise, let $idx'_I:=0$.
\end{itemize}
%
%\item update $idx'_I$ as follows: if $k+1 \in \rng(\pi)$, then let $idx'_i: = j''$ for some $j'' \in \pi^{-1}(k+1)$, 
% otherwise, if $\{j'' \mid j'' \in \rng(\pi), x_{j'} \preceq x_{j''} \} \neq \emptyset$, let $idx'_i: = j'''$ for some $j''' \in \pi^{-1}(j'')$  such that $x_{j''}$ is $\preceq$-minimal in $\{x_{j''} \mid j'' \in \rng(\pi), x_{j'} \preceq x_{j''}\}$, otherwise, let $idx'_i:=0$,
%
\item Let 
$
C'_0   : = C'_0 + d \sum \limits_{pmt(i') \neq 0} \beta_{pmt(i')} B(i', k + l + 1).
$ 
\end{itemize}

\item $\cur$ is not chosen to belong to $Z'_I$:   
\begin{itemize}
\item Let $post'(t, X_2) = \{x_{j''} \in X_2 \mid ass(j'') \in [k], idx'_I > 0, x_{idx'_I} \preceq_p x_{ass(j'')} \preceq_p x_{j'}\}$. 
\begin{itemize}
\item If $post'(t, X_2)$ is nonempty, 
then let $idx'_I$ be some $j'' \in [k]$ such that $x_{j''}$ is a $\preceq_{q}$-minimal element in $post'(t, X_2)$. 
\item Otherwise, let $idx'_I:=0$.
\end{itemize}
\item Let 
$
C'_0   : = C'_0 + c \sum \limits_{pmt(i') \neq 0} \beta_{pmt(i')} B(i', k + l + 1).
$ 
\end{itemize}
\end{enumerate}

\smallskip

\noindent {\it Case $idx'_I = 0$ and $z'_{i+1} \prec_q x_{j'}$, or $idx'_I > 0$ and $ z'_{i+1} \prec_q x_{idx'_I}$}.
\begin{itemize}
\item Let $post'(t, X_2) = \{x_{j''} \in X_2 \mid ass(j'') \in [k], idx'_I >0,  x_{idx'_I} \preceq_p x_{ass(j'')} \preceq_p x_{j'} \}$. 
\begin{itemize}
\item If $post'(t, X_2)$ is nonempty, 
then let $idx'_I$ be $j'' \in [k]$ such that $x_{j''}$ is a $\preceq_{q}$-minimal element in $post'(t, X_2)$. 
\item Otherwise, let $idx'_I := 0$.
\end{itemize}
\item Let $
C'_0   : = C'_0 + c \sum \limits_{pmt(i') \neq 0} \beta_{pmt(i')} B(i', k + l + 1)
$. 
\end{itemize}

%We distinguish between the following cases.
%\begin{itemize}
%\item $idx'_I > 0$ and $x_{idx'_I} \preceq_p z'_i$,
%\item $idx'_I > 0$ and $x_{idx'_I} \simeq_p z'_{i+1}$,
%\item $idx'_I > 0$ and $ z'_{i+1} \prec_p x_{idx'_I}$,
%\item $idx'_I = 0$ and $z'_{i+1} \simeq_p x_{j'}$,
%\item $idx'_I = 0$ and $z'_{i+1} \prec_p x_{j'}$.
%\end{itemize}

%Note that above for the case $idx'_I > 0$ and $z'_{i+1} \prec_q x_{idx'_I}$,  it is disallowed to nondeterministically guess that $\cur$ belongs to $Z'_I$.  The reason for this is that we need guarantee the consistency of these nondeterministic choices: If $z'_{i+1}$ was not chosen to belong to $Z'_I$, then $\cur$ should not be chosen to belong to $Z'_I$ as well, since otherwise, $z'_{i+1}$ would be forced to belong to $Z'_I$ (as a result of $\cur < z'_{i+1}$), contradicting to the choice made before.

\smallskip

\noindent {\bf The distinguished state $q_{acc}$}.\\
The rational VASS $\cB$ enters a distinguished state $q_{acc}$ if the state $q_f$ is reached, and for each $i \in [l]$, $pmt(i) = i$. 

The requirement that for each $i \in [l]$, $pmt(i) = i$ is consistent with the intuition of $pmt$: When the guessing is over, for each $i \in [l]$, the current value of $y_i$ is stored in $y_{i=pmt(i)}$.

\smallskip

\noindent The construction of $\cB$ is finished.

\smallskip

From the construction, there is an input $w$ such that the run of  $\cA$ stops at $q_f$ and produces a zero output iff there is a run of $\cB$ that reaches the state $q_{acc}$ and at the same time, the values of  $\overrightarrow{C'}$ and $\overrightarrow{C''}$ satisfy the inequality (\ref{eqn-po-ls-0-lem}) and (\ref{eqn-po-gr-0-lem}) respectively. %that is, 
% constraints stated in Remark~\ref{rem-po}, that is,  
%
%$\overrightarrow{C'}$ satisfies one of the following constraints:
% 1) $C'_1 > 0$, 2) $C'_2 < 0$, 3) $C'_1 \le 0$, $C'_2 \ge 0$, and $C'_0 + c_{-r} C'_1  + c_s C'_2 < 0$, on the other hand, $\overrightarrow{C''}$ satisfies one of the following constraints: 1) $C''_1 < 0$, 2) $C''_2 > 0$, 3) $C''_1 \ge 0$, $C''_2 \le 0$, and $C'_0 + c_{-r}  C'_1 + c_s C'_2  > 0$.

\smallskip

\noindent {\bf Complexity analysis}. Suppose a non-normalized $\raq$ $\cA$ is given as the input, with $k$ control variables and $l$ data variables respectively. The number of states of normalized $\raq$ is exponential over the number of control variables in $\cA$ (i.e. $k$). The number of different values of the variables $idx'_I, idx''_I$  for $I \in \intval$ in $\cB$ is $(k+1)^{r+s+2}$, while that of the mapping $pmt$ of $\cB$ is $(l+1)^l$. It follows that the number of states of $\cB$ is at most exponential over the size of $\cA$. Since the reachability problem for rational VASS is in \nptime\ (cf. Theorem~\ref{thm-str-reach-q-vass} in Appendix~\ref{app-rat-vass}), we conclude that the reachability problem for copyless $\raq$ is in \nexptime.
}

%%%%%%%%%%%%%%%%%%%%%%%%%%%%%%%%%%%%%%%%%%%%%%%%%%%%%%%%%%
%%%%%%%%%%%%%%%%%%%%%%The original reduction by Zhilin%%%%%%%%%%%%%%%%%%%%%
%%%%%%%%%%%%%%%%%%%%%%%%%%%%%%%%%%%%%%%%%%%%%%%%%%%%%%%%%%

%%%%%%%%%%%%%%%%%%%%%%%%%%%%%%%%%%%%%%%%%%%%%%%%%%%%%%%%%%%%%%%%%%%%
%%%%%%%%%%%%%%%%%%%%%%%%%%%appendix for rat-vass%%%%%%%%%%%%%%%%%%%%%%%%%%%%%%%
%%%%%%%%%%%%%%%%%%%%%%%%%%%%%%%%%%%%%%%%%%%%%%%%%%%%%%%%%%%%%%%%%%%%

\input{app-rat-vass.tex}

%%%%%%%%%%%%%%%%%%%%%%%%%%%%%%%%%%%%%%%%%%%%%%%%%%%%%%%%%%%%%%%%%%%%
%%%%%%%%%%%%%%%%%%%%%%%%%%%appendix copied from app-yufang%%%%%%%%%%%%%%%%%%%%%%%%
%%%%%%%%%%%%%%%%%%%%%%%%%%%%%%%%%%%%%%%%%%%%%%%%%%%%%%%%%%%%%%%%%%%%
%%%%%%%%%%%%%%%%%%%%%%%%%%%%%%%%%%%%%%%%%%%%%%%%%%%%%%%%%%%%%%%%%%%%

\subsection{Results on the coverability problem}
\label{app:cov}

\begin{theorem}
	The invariant problem of $\raq$ can be reduced to the coverability problem of $\raq$, which can be further reduced to the reachability problem of $\raq$. Both reductions are in polynomial-time.
\end{theorem}

Let $\cA$ be an $\raq$.
The first reduction is done by creating (1) an $\raq$ $\cA'$ such that $\exists v \in\cA(w): v\geq 0$ for some $w$ iff $\exists v \in\cA'(w): v> 0$ for some $w$ by adding an arbitrary positive value to the output of $\cA$, e.g., add a $\cur>0$ and
(2) another $\raq$ $\cA''$ by negating all output expressions of $\cA'$. The answer to the non-zero problem of $\cA$ is positive iff the answers to the coverability problems of $\cA'$ and $\cA''$ are both positive.
The second reduction is done by adding a transition $(q,\cur\geq 0)\rightarrow (q', y_1:=\zeta(q)-\cur)$ from all final states $q$ to a new state $q'$ for some data variable $y_1$ and setting $q'$ as the only new final state with the output expression $y$.

\begin{corollary}
	The coverability problem of copyless $\raq$s is in \nexptime. 
\end{corollary}

\subsection{Results on configuration reachability and coverability}
\label{app:conf-cov-rea}

For Petri-net or VASS, people are also interested in configuration reachability and configuration coverability problems. The corresponding problems in $\raq$ are defined as follows.
Given an $\raq$ $\cA=\langle Q,q_0,F,\vu_0,\delta,\zeta\rangle$ and a configuration $(q_n,\vu_n)$, the configuration reachability problem asks if $(q_0,\vu_0)\vdash^{\ast} (q_n,\vu_n)$ and the configuration coverability problem asks if 
$(q_0,\vu_0)\vdash^{\ast} (q_n,\vu'_n)$ and $\vu'_n \geq \vu_n$. 
We show that the two problems in $\raq$ are inter-reducible in polynomial-time and they are not easier than the reachability problem, which is undecidable (Theorem~\ref{thm-reach-und}).

\begin{lemma}
	The configuration reachability problem of $\raq$ can be reduced to the configuration coverability problem of $\raq$, and vice versa. 
\end{lemma}
Let $\cA$ be an $\raq$ over $(X,Y)$ and the two problem are targeting the configuration $(q,\vv)$. 
The first reduction is done by creating a $\raq$ $\cA'$ over $(X\cup X', Y\cup Y')$ such that $X'$ and $Y'$ compute the negation of $X$ and $Y$, using a similar construction of the proof of Lemma 2 in~\cite{Haase14}. The reverse direction is done by adding the transition $(q, \bigwedge_{x_i\in X} x=\vv(X)[i])\rightarrow(q', \{\})$ and for all $y\in Y$ the transition $(q',\cur\geq 0) \rightarrow (q', \{y:=y-cur\})$ and targeting the configuration $(q',\vv)$ instead.

\begin{lemma}
	The coverability and reachability problems of $\raq$ can be reduced to the configuration coverability and reachability problems of $\raq$, respectively.
\end{lemma} 
The reduction is done by adding a transition with the assignment $y:=\zeta(q)$ from all final states $q$ to a new state $q'$ and use the configuration $(q',0)$ in the corresponding configuration coverability and reachability problems.

%\tony{Why is there reachability problem here. It should be after Ondra's proof of undecidability?}
%\yfc{Becasue here is a section to discuss ``configuration'' reachailbity and coverability problems, which is defined only here.}

%% file: app-rat-vass.tex
\subsection{Rational VASS ($\bbQ$-VASS), with or without $\pm \infty$}
\label{app-rat-vass}

We will use the following theorem.
\begin{theorem}
\label{theo:presburger}~\cite{VermaSS05}
For every finite state automaton $\cA$ over the alphabet $\{\alpha_1,\ldots,\alpha_k\}$,
one can construct in polynomial time an existential Presburger formula $\Psi_{\cA}(x_1,\ldots,x_k)$
such that for every $(a_1,\ldots,a_k)\in \bbN^k$,
$\Psi_{\cA}(a_1,\ldots,a_k)$ holds if and only if
there is a word $w \in L(\cA)$ with Parikh image $(a_1,\ldots,a_k)$.
\end{theorem}

In fact, Theorem~\ref{theo:presburger} holds also for context-free grammar,
but for our purpose, finite state automata is sufficient.
One can easily modify it so that instead of Parikh image, 
it talks about the number of transitions in an accepting run,
as stated below.

\begin{theorem}\label{theo:presburger-two}~\cite{VermaSS05}
For every finite state automaton $\cA$ over transitions $\{t_1,\ldots,t_m\}$,
one can construct in polynomial time an existential Presburger formula $\Psi_{\cA}(x_1,\ldots,x_m)$
such that for every $(a_1,\ldots,a_m)\in \bbN^m$,
$\Psi_{\cA}(a_1,\ldots,a_m)$ holds if and only if
there is an accepting run of $\cA$ in which transition $t_i$ appears $a_i$ times.
\end{theorem}

Indeed, Theorem~\ref{theo:presburger-two} is a straightforward consequence of Theorem~\ref{theo:presburger}
by labelling the transitions with new labels, so that all of them have different labels.

\paragraph*{Rational VASS ($\bbQ$-VASS)}
A $k$-dimensional $\bbQ$-VASS is a pair $(S,\Delta)$, where $S$ is a finite set of states
and $\Delta$ is a finite subset of $S\times \bbQ^k\times S$.
A configuration is a pair $(s,\vu)\in S\times \bbQ^k$.
A configuration $(r,\vv)$ is reachable from $(s,\vu)$,
if there is a sequence of transitions $(s_0,\vv_0,s_1),(s_1,\vv_1,s_2),\ldots,(s_n,\vv_n,s_{n+1})$ of $\Delta$
such that $s_0=s$, $s_{n+1}=r$ and $\vv = \vu + \sum_{i=0}^{n} \vv_i$.
We say that $(r,\vv)$ is coverable by $(s,\vu)$,
if $\vv \leq \vu + \sum_{i=0}^{n} \vv_i$.

Two popular problems for $\bbQ$-VASS are:
\begin{itemize}\itemsep=0pt
\item
{\bf $\bbQ$-VASS configuration reachability:}

On input $\bbQ$-VASS $(S,\Delta)$ and two configurations $(s,\vv)$ and $(r,\vu)$,
decided whether $(s,\vv)$ is reachable from $(r,\vu)$.

\item
{\bf $\bbQ$-VASS configuration coverability:}

On input $\bbQ$-VASS $(S,\Delta)$ and two configurations $(s,\vv)$ and $(r,\vu)$,
decided whether $(s,\vv)$ is coverable from $(r,\vu)$.

\end{itemize}
The two problems above are special cases of the following problem:
\begin{itemize}\itemsep=0pt
\item
{\bf $\bbQ$-VASS strong configuration reachability:}

On input $k$-dimensional $\bbQ$-VASS $(S,\Delta)$, a configuration $(r,\vu)$, a state $s$ and 
a Boolean combination of atomic Presburger formula $\varphi(x_1,\ldots,x_k)$,
decide whether there is a configuration $(s,\vv)$
such that $\varphi(\vv)$ holds and $(s,\vv)$ is reachable from $(r,\vu)$.
\end{itemize}
Note that reachability and coverability are special cases of strong reachability,
where $\varphi(z_1,\ldots,z_m)$ is $(z_1,\ldots,z_m) = \vv$
and $(z_1,\ldots,z_m) \geq \vv$, respectively.

\begin{theorem}\label{thm-str-reach-q-vass}
$\bbQ$-VASS strong configuration reachability is in $\nptime$.
\end{theorem}

\begin{proof}
Let the input be $(S,\Delta)$, $(r,\vu)$, $s$ and $\varphi(x_1,\ldots,x_k)$ as above.
Let $\Delta= \{(s_1,\vv_1,t_1),\ldots,(s_m,\vv_m,t_m)\}$.
Deciding strong reachability is equivalent to deciding the satisfiability of the following formula:
\begin{eqnarray*}
\Phi & := & \exists x_1 \cdots \exists x_m \ \Psi_{\cA}(x_1,\ldots,x_m) \ \wedge \
\varphi\Big(\vu + \sum_{i=1}^m x_i\vv_i \Big)
\end{eqnarray*}
where $\Psi_{\cA}(x_1,\ldots,x_m)$ is the formula constructed via Theorem~\ref{theo:presburger-two},
by viewing $(S,\Delta)$ as a finite state automaton, $r$ the initial state and $s$ the final state.
Note that each $\vv_i$ may contain rational numbers.
%, so the factor of $x_i$ maybe fractionals.
However, we can get rid of the denominators easily by scalar multiplication with integers, which only increases the size of the formula polynomially (since the constants are encoded in binary). Then the satisfiability for $\Phi$ is reduced to integer linear programming problem, which is well-known to be  $\nptime$-complete. 
\end{proof}

\paragraph*{The conversion of $\bbQ$-VASS with $\pm\infty$ to existential Presburger formula}

Let the input be $(S,\Delta)$, $(r,\vu)$, $s$ and $\varphi(x_1,\ldots,x_k)$ as above. Let $\Delta= \{(s_1,\vv_1,t_1),\ldots,(s_m,\vv_m,t_m)\}$ such that  $\vv_1, \dots, \vv_m$ may include the special symbols $\pm \infty$. In addition, suppose we know that $+\infty$ and $-\infty$ represent an arbitrary rational number bigger and smaller than $c_{\max}$ and $c_{\min}$ respectively.
Let $\Phi$ be the formula constructed in the proof of Theorem~\ref{thm-str-reach-q-vass}, that is, 
\begin{eqnarray*}
\Phi & := & \exists x_1 \cdots \exists x_m \ \Psi_{\cA}(x_1,\ldots,x_m) \ \wedge \
\varphi\Big(\vu + \sum_{i=1}^m x_i\vv_i \Big).
\end{eqnarray*}
Then we can transform $\Phi$ into a formula $\Phi'$ without $\pm \infty$ as follows.
\begin{itemize}
\item
For term $x_i(+\infty)$, we replace it with a fresh existential variable $z$
and add a conjunct $ z \geq c_{\max} x_i$.
\item
For term $x_i(-\infty)$, we replace it with a fresh new existential variable $z$
and add a conjunct $z\leq c_{\min} x_i$.
\item
Likewise for terms such as $a(+\infty)$ and $a(-\infty)$.
\end{itemize}